\documentclass[11pt]{article}

\usepackage[utf8]{inputenc}
\usepackage[T1]{fontenc}
\usepackage{amsmath}
\usepackage{amsthm}
\usepackage{amsfonts}
\usepackage{amssymb}
\usepackage{mathtools}
\usepackage{xcolor}
\usepackage{graphicx}
\usepackage[version=4]{mhchem}
\usepackage[flushleft]{threeparttable}
\usepackage{booktabs}       
\usepackage{makecell}
\usepackage{siunitx}
\usepackage{setspace}
\usepackage{array}
\usepackage{enumitem}
\usepackage{bm}
\usepackage{wrapfig}
\usepackage{placeins}
\usepackage{fix-cm}
\usepackage{pifont}
\usepackage{tocloft}

\newcommand{\cmark}{\ding{51}}
\newcommand{\xmark}{\ding{55}}

\newcommand{\PaperTableStyle}{%
    \centering
    \footnotesize
    \setlength{\tabcolsep}{3pt}%
    \renewcommand{\arraystretch}{1.10}%
}
\newcommand{\PaperTableNotesStyle}{\scriptsize}

\newcommand{\PreserveBackslash}[1]{\let\temp=\\#1\let\\=\temp}
\newcolumntype{C}[1]{>{\PreserveBackslash\centering}p{#1}}
\newcolumntype{R}[1]{>{\PreserveBackslash\raggedleft}p{#1}}
\newcolumntype{L}[1]{>{\PreserveBackslash\raggedright}p{#1}}
\renewcommand{\thefootnote}{\fnsymbol{footnote}}

\providecommand{\Sph}{S^2}

\theoremstyle{plain}
\newtheorem{theorem}{Theorem}[section]

\newtheorem{lemma}[theorem]{Lemma}
\theoremstyle{definition}

\theoremstyle{remark}
\newtheorem{remark}[theorem]{Remark}

\usepackage[letterpaper]{geometry}
\usepackage[superscript,biblabel,nomove]{cite}
\usepackage[colorlinks=true,allcolors=blue]{hyperref}
\usepackage{mathptmx}
\usepackage{microtype}

\makeatletter
\renewcommand{\section}{%
    \@startsection{section}{1}{\z@}%
    {-2.0ex \@plus -0.5ex \@minus -0.2ex}%
    {1.5ex \@plus 0.3ex \@minus 0.2ex}%
    {\large\bfseries\raggedright}%
}
\renewcommand{\subsection}{%
    \@startsection{subsection}{2}{\z@}%
    {-1.8ex \@plus -0.5ex \@minus -0.2ex}%
    {0.8ex \@plus 0.2ex}%
    {\normalsize\bfseries\raggedright}%
}
\renewcommand{\subsubsection}{%
    \@startsection{subsubsection}{3}{\z@}%
    {-1.5ex \@plus -0.5ex \@minus -0.2ex}%
    {0.5ex \@plus 0.2ex}%
    {\normalsize\bfseries\raggedright}%
}
\renewcommand{\paragraph}{%
    \@startsection{paragraph}{4}{\z@}%
    {1.5ex \@plus 0.5ex \@minus 0.2ex}%
    {-1em}%
    {\normalsize\bfseries}%
}

\newcommand{\@toptitlebar}{%
    \hrule height 4pt
    \vskip 0.25in
    \vskip -\parskip
}
\newcommand{\@bottomtitlebar}{%
    \vskip 0.29in
    \vskip -\parskip
    \hrule height 1pt
    \vskip 0.09in
}

\renewcommand{\maketitle}{%
    \par
    \begingroup
    \renewcommand{\thefootnote}{\fnsymbol{footnote}}%
    \renewcommand{\@makefnmark}{\hbox to \z@{$^{\@thefnmark}$\hss}}%
    \long\def\@makefntext##1{%
        \parindent 1em\noindent
        \hbox to 1.8em{\hss $\m@th ^{\@thefnmark}$}##1%
    }%
    \thispagestyle{empty}%
    \vbox{%
        \hsize\textwidth
        \linewidth\hsize
        \vskip 0.1in
        \@toptitlebar
        \centering
        {\LARGE\bfseries \@title\par}
        \@bottomtitlebar
        \def\And{%
            \end{tabular}\hfil\linebreak[0]\hfil
            \begin{tabular}[t]{c}\bfseries\rule{\z@}{24pt}\ignorespaces
        }%
        \def\AND{%
            \end{tabular}\hfil\linebreak[4]\hfil
            \begin{tabular}[t]{c}\bfseries\rule{\z@}{24pt}\ignorespaces
        }%
        \begin{tabular}[t]{c}\bfseries\rule{\z@}{24pt}\@author\end{tabular}%
        \vskip 0.3in \@minus 0.1in
    }%
    \@thanks
    \endgroup
    \let\maketitle\relax
    \let\thanks\relax
}

\renewenvironment{abstract}%
{%
    \vskip 0.075in
    \centerline{\large\bfseries Abstract}
    \vspace{0.5ex}
    \begin{quote}
}
{
    \par
    \end{quote}
    \vskip 1ex
}
\makeatother

\title{
DPA4: Pushing the Accuracy--Cost Frontier of Interatomic Potentials with EMFA SO(2) Convolution
}

\author{
    \begin{minipage}{0.95\textwidth}
        \centering
        Tiancheng Li$^{1,2}$ \quad
        Wentao Li$^{3}$ \quad
        Anyang Peng$^{2}$\\
        Jianming Xue$^{1,4,*}$ \quad
        Linfeng Zhang$^{2,5,*}$ \quad
        Duo Zhang$^{2,5,6,*}$ \quad
        Han Wang$^{4,7,*}$ \\[0.6em]
        {\normalfont\small
            $^{1}$State Key Laboratory of Nuclear Physics and Technology, School of Physics, Peking University, Beijing 100871, China\\
            $^{2}$AI for Science Institute, Beijing 100080, P. R. China\\
            $^{3}$Department of Chemical Engineering, Tsinghua University, Beijing 100084, P. R. China\\
            $^{4}$HEDPS, CAPT, College of Engineering, Peking University, Beijing 100871, P. R. China\\
            $^{5}$DP Technology, Beijing 100080, P. R. China\\
            $^{6}$Academy for Advanced Interdisciplinary Studies, Peking University, Beijing 100871, P. R. China\\
            $^{7}$National Key Laboratory of Computational Physics, Institute of Applied Physics and Computational Mathematics, Fenghao East Road 2, Beijing 100094, P. R. China\\
            $^{*}$Correspondence: \href{mailto:jmxue@pku.edu.cn}{jmxue@pku.edu.cn};
            \href{mailto:linfeng.zhang.zlf@gmail.com}{linfeng.zhang.zlf@gmail.com};
            \href{mailto:zhduodyx@pku.edu.cn}{zhduodyx@pku.edu.cn};
            \href{mailto:wang_han@iapcm.ac.cn}{wang\_han@iapcm.ac.cn}
        }
    \end{minipage}
}

\date{}

\begin{document}

\maketitle

\begin{abstract}
    Machine-learning interatomic potentials now approach quantum-mechanical accuracy, but the most expressive equivariant architectures are costly to evaluate, and the leading ones depend on auxiliary denoising or direct-force pretraining.
    We introduce DPA4, an SE(3)-equivariant architecture spanning six size classes from 0.48 to 25 million parameters and reaching the accuracy of the strongest published models at several-fold to an order-of-magnitude higher inference throughput.
    Its convolution couples edge and node features across all angular degrees in an edge-local frame, and its Wigner bilinear nonlinearity is universal in its full form and, on an exact quadrature grid, equivariant to machine precision.
    On Matbench Discovery, DPA4 leads every ranked metric, and every variant evaluated lies on the accuracy--throughput Pareto frontier of the compliant leaderboard.
    DPA4-Pro attains the lowest energy error on OMat24 and lower total-energy and force errors than the strongest conservative baseline on the OMol25 composition-validation split.
    All variants are trained through the conservative energy-gradient path alone, made practical by a threefold-faster compiled implementation.
    DPA4 thus brings leaderboard-class accuracy within the routine compute budgets of molecular-dynamics and materials-screening workflows, for both inorganic crystals and organic molecules.
\end{abstract}

\section{Introduction}

Machine-learning interatomic potentials (MLIPs) are moving from case-specific models trained on dedicated density-functional-theory (DFT) data to pretrained atomistic foundation models~\cite{yuan2026foundation}, also called large atomistic models (LAMs), intended as broad DFT surrogates for simulation, materials discovery and molecular design~\cite{merchant2023scaling}.
The lineage runs from the descriptor- and kernel-based potentials of the case-specific era~\cite{behler2007generalized,bartok2010gaussian,schutt2017schnet,unke2019physnet,zhang2018deep,wang2018deepmd} to the graph neural network models that define the LAM era~\cite{chen2022universal,deng2023chgnet,batatia2022mace,kovacs2023mace,yang2024mattersim,neumann2024orb,rhodes2025orbv3,liao2023equiformerv2,liao2026equiformerv3,wood2025umafamilyuniversalmodels,zhang2024dpa2,zhang2025dpa3}.

That promise is expensive to realize: training UMA-M~\cite{wood2025umafamilyuniversalmodels}, built on the eSEN architecture~\cite{fu2025esen}, required 129{,}024 H200 GPU-hours across direct-force pretraining and conservative fine-tuning.
The barrier is not confined to pretraining: the per-step expense that makes pretraining costly is paid again in fine-tuning to a target domain, and then at every step of molecular dynamics or structure relaxation in deployment.
This motivates the question whether comparable accuracy can be reached at substantially lower cost, at inference as well as in training.

On the accuracy side, equivariant architectures carry directional information as first-class features that transform under SO(3), instead of compressing it immediately into rotation-invariant features; NequIP~\cite{batzner2022nequip}, MACE~\cite{batatia2022mace}, the Equiformer family~\cite{liao2022equiformer,liao2023equiformerv2,liao2026equiformerv3} and eSEN~\cite{fu2025esen} have shown that explicit directional features substantially improve data efficiency and benchmark accuracy.
On the efficiency side, the expense of SE(3)-equivariant models is dominated by Clebsch--Gordan tensor products, whose cost grows rapidly with angular order.
eSCN showed that SO(3)-equivariant convolutions can be reduced to equivalent edge-local SO(2) operations~\cite{passaro2023reducing}, a strategy also used by recent high-performing models such as eSEN~\cite{fu2025esen} and EquiformerV3~\cite{liao2026equiformerv3}.
In these constructions the edge modulates the node feature diagonally: the degree-$l$ edge feature multiplies only the degree-$l$ node feature, whereas a Clebsch--Gordan product couples every edge degree to every node degree.
That restriction, and not the reduction to SO(2), is where expressive power is lost; letting the edge couple different degrees recovers it.

Cost is not the only obstacle: the conservative training objective itself resists acceleration.
Because forces are the negative gradient of the energy, the force loss differentiates through the energy model and every training step requires a double-backward pass; training stacks developed for large language models, which are tuned for single-backward gradients, do not transfer directly to this setting.
Leading models such as EquiformerV3, eSEN and UMA work around this by pretraining with denoising (DeNS~\cite{liao2024dens}) or direct-force prediction, objectives whose targets are obtained from a single forward pass, and only then fine-tuning with the conservative energy-gradient objective.
This two-stage protocol spends an additional pretraining stage and the engineering effort to maintain two objectives, motivating an architecture in which conservative energy-gradient training is itself compiler-friendly from the start.

\begin{figure}[!t]
    \centering
    \includegraphics[width=.7\linewidth]{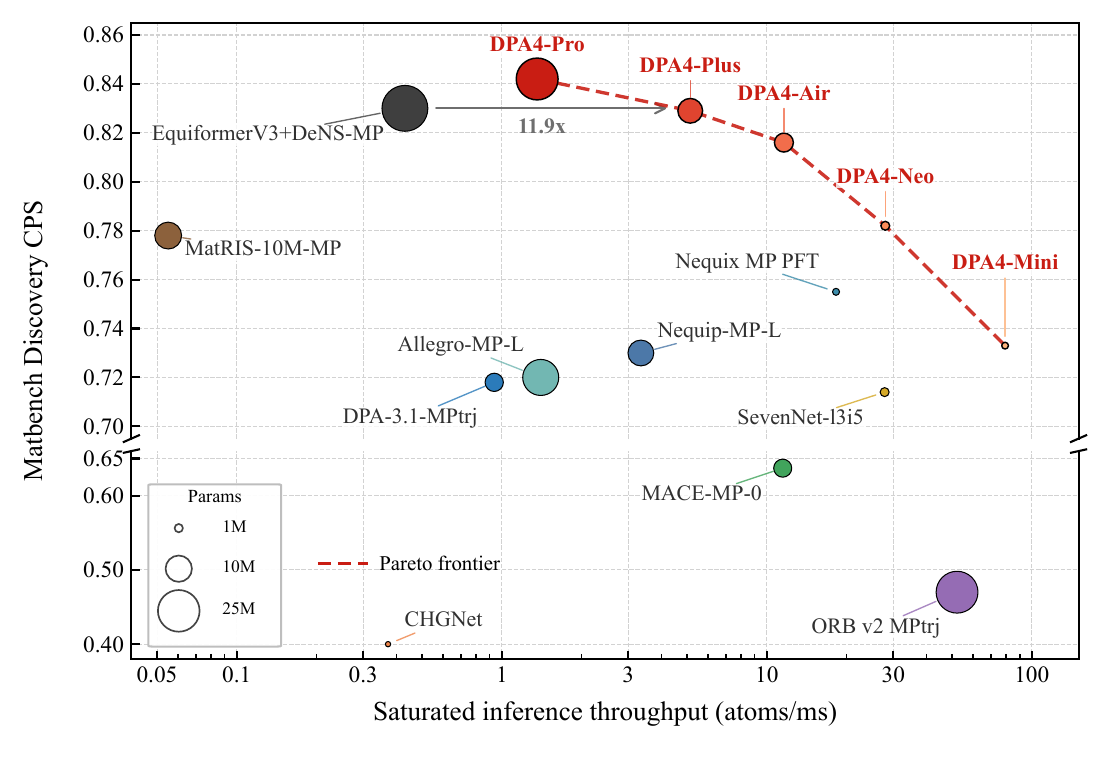}
    \caption{
        Accuracy--throughput Pareto frontier for compliant Matbench Discovery models~\cite{riebesell2025matbench}.
        The Combined Performance Score (CPS) is plotted against saturated inference throughput for the five DPA4 variants and ten representative baselines.
        Marker area is proportional to the number of model parameters.
        For each model, saturated throughput is the largest atom-normalized throughput among the successful trials of an adaptive diamond-supercell sweep; at each system size, throughput is computed from the mean wall time of 100 evaluations of energy, forces and stress after 10 warm-up evaluations on a single NVIDIA H20 GPU.
        Higher values are preferred on both axes; the dashed line joins the non-dominated DPA4 variants.
        The horizontal arrow gives the throughput ratio between EquiformerV3+DeNS-MP and DPA4-Plus at comparable CPS.
        The break in the CPS axis omits the empty interval between 0.66 and 0.695.
        The complete benchmark protocol is given in Supplementary Section~\ref{sec:sm-inference}.
    }
    \label{fig:cps-throughput}
\end{figure}

Here we introduce DPA4, an SE(3)-equivariant interatomic-potential architecture built on an EMFA (Edge--node degree coupling, Multi-Focus, Attention) SO(2) convolution and a Wigner bilinear network (WBN) nonlinearity, which achieves leading accuracy at substantially lower inference and training cost on inorganic-crystal and organic-molecule benchmarks.
Its architectural innovations, co-designed for efficiency and accuracy, are (A1) a low-rank edge--node SO(2)-equivariant degree coupling in an edge-local frame, (A2) a multi-focus design for message nonlinearity, (A3) envelope-gated attention for message aggregation --- the E, MF and A of EMFA --- and (A4) a Wigner bilinear network, a universal architecture for nonlinear SO(3)-equivariant maps between irreps, provably complete in its full-frame form.
All four raise generalization accuracy at low computational cost relative to standard SO(2)-equivariant baselines; the WBN, evaluated by exact Lebedev--uniform quadrature on the rotation group, is in addition equivariant to machine precision.
A shape-stable, compiler-friendly implementation of conservative energy-gradient training makes the energy-to-force path compatible with \texttt{torch.compile} and gives a threefold wall-clock training speedup in controlled ablations.

Across four datasets spanning inorganic crystals and organic molecules, DPA4 combines high accuracy with low inference and training cost.
On Matbench Discovery~\cite{riebesell2025matbench}, all five DPA4 variants evaluated on this benchmark lie on the CPS--throughput Pareto frontier among fifteen compliant models; DPA4-Pro leads every ranked metric, whereas DPA4-Plus matches the highest-CPS baseline with 11.9$\times$ higher saturated inference throughput (Fig.~\ref{fig:cps-throughput}).
On OMat24~\cite{barroso2024open}, a materials dataset comprising 118 million DFT calculations, DPA4-Pro attains the lowest mean absolute error (MAE) in energy (9.4~meV/atom).
Its force MAE is 42.7 versus 41.6~meV/\AA{} for EquiformerV3 with $L_{\max}=6$~\cite{liao2026equiformerv3}, a difference of only 1.1~meV/\AA{}; DPA4-Pro has about half as many parameters and requires one-sixth of the training compute.
On the OMol25 composition-validation split~\cite{levine2025open}, DPA4-Pro lowers both the total-energy and force MAEs relative to UMA-M-1.1~\cite{wood2025umafamilyuniversalmodels}, the strongest conservative baseline in the comparison.
DPA4-Pro has 25.2M parameters versus 1.4B total (50M active) for UMA-M-1.1 and required over an order of magnitude less training compute than the reported UMA-M-1.0 run; the cost of the subsequent fine-tuning to UMA-M-1.1 is not reported.
On~SPICE-MACE-OFF~\cite{kovacs2023mace}, DPA4-Plus lowers the aggregate energy and force errors of the best published model by 39\% and 36\% (Table~\ref{table:spice}).
All DPA4 models in these comparisons use conservative energy-gradient training alone, without auxiliary denoising or direct-force pretraining, whereas several leading baselines use one or both of these stages.
DPA4 is evaluated here in the single-task, per-dataset setting; extending this backbone to multitask LAM pretraining remains a direction for future work (Section~\ref{sec:discussion}).


\section{Results}
\label{sec:results}

\subsection{The DPA4 architecture}
\label{sec:arch_overview}

DPA4 is a conservative SE(3)-equivariant message-passing graph neural network that maps atomic species and positions to a scalar potential energy through a Geometry-Informed Embedding (GIE) stage, $N_{\mathrm{layer}}$ stacked equivariant interaction blocks, and an atomic energy head (Fig.~\ref{fig:arch}a).
Each interaction block has a residual structure with two skip connections: one over an EMFA SO(2) convolution followed by an equivariant RMSNorm, and one over a second equivariant RMSNorm followed by an equivariant feed-forward network (FFN).
Architectural designs A1, A2 and A3 introduced in the Introduction, namely the low-rank edge--node SO(2)-equivariant degree coupling, the multi-focus design for message nonlinearity, and the envelope-gated attention for message aggregation, operate inside the EMFA SO(2) convolution (Fig.~\ref{fig:arch}c).
Architectural design A4 supplies the SO(3)-equivariant nonlinearity: the equivariant FFN is implemented via a WBN, a stack of Wigner bilinear blocks mapping the node feature space to itself.
Forces and virials are obtained by automatic differentiation through the scalar energy.

A shared edge cache (Fig.~\ref{fig:arch}b) precomputes and reuses per-edge quantities across interaction blocks: the distance $r_{ij}$ and unit direction $\widehat{\mathbf{r}}_{ij}$, a radial-basis expansion with cutoff envelope, and per-pair species features.
It also caches the Wigner-D rotation $\mathbf{D}_{ij}$ and its inverse $\mathbf{D}_{ij}^{-1}$ that transport features between the global and the edge-local frame.
The GIE stage (Fig.~\ref{fig:arch}b) then injects both chemistry and geometry into the initial node representation before the first interaction block.
A scalar branch combines per-species type features with an SO(3)-invariant local-environment descriptor through Feature-wise Linear Modulation (FiLM)~\cite{perez2018film}, producing the initial $l=0$ slice.
In parallel, an equivariant branch projects each neighbor direction onto real spherical harmonics and weights the projection by a radial-species profile, producing the initial $l\geq 1$ slices.
This gives the model chemical and geometric context from the outset, rather than forcing all local-environment information to be discovered through iterated message passing.
The construction details are given in Section~\ref{sec:gie}.

The EMFA SO(2) convolution (Fig.~\ref{fig:arch}c) transports the source node features into an edge-local SO(2) frame, constructs an equivariant per-edge message in that frame, then lifts the message back and aggregates over neighbors to update the node features.
Mathematical details are given in Section~\ref{sec:so2}.
In the local frame, the convolution applies the low-rank edge--node SO(2)-equivariant degree coupling (A1), built on the eSCN reduction~\cite{passaro2023reducing} that replaces the costly SO(3) Clebsch--Gordan tensor product by equivalent edge-local SO(2) operations.
It couples every edge degree to every node degree, as an SO(3) Clebsch--Gordan product does, in contrast to similar SO(2)-equivariant constructions in eSEN~\cite{fu2025esen} and EquiformerV3~\cite{liao2026equiformerv3}, where the degree-$l$ edge feature multiplies only the degree-$l$ node feature.
Because the SO(2) Clebsch--Gordan coefficients are trivial, they are hard-wired into the block structure of the kernel rather than tabulated; its learnable part, a function of the invariant radial--species features, is factorized at low rank across channels, improving accuracy at modest additional training cost (Section~\ref{sec:ablation}).
The multi-focus design (A2) splits the hidden width into $F$ parallel focus streams, each processed by its own SO(2) stack and then reweighted by a cross-focus softmax competition, introducing message nonlinearity.
At matched SO(2) width, this parallel-focus structure matches the accuracy of the $F=1$ single-focus baseline with substantially fewer parameters and lower training and inference cost (Section~\ref{sec:ablation}).
Aggregation over neighbors uses an envelope-gated attention (A3) computed from the SO(3)-invariant $l=0$ slice, with a destination-side output gate (Section~\ref{sec:attn}).
This attention-weighted aggregation improves accuracy at small additional training cost relative to a plain envelope-weighted scatter sum (Section~\ref{sec:ablation}).
The cutoff envelope drives edge contributions smoothly to zero as $r_{ij}\to r_{\mathrm{c}}$, a requirement for stable molecular dynamics.

The equivariant FFN that follows the convolution is instantiated as a WBN (A4): a Wigner bilinear block expands the node feature into a pair of scalar functions on the rotation group, multiplies them point by point and re-expands the product, so that a single operation couples all degrees at once (Section~\ref{sec:s2-method}).
With the full range of frames retained, a stack of such blocks can approximate any transformation of the node features that respects rotations, to arbitrary accuracy (Supplementary Section~\ref{sec:sm-wbn-complete}).
The re-expansion integral is evaluated by a Lebedev rule on the sphere crossed with a uniform rule in the residual angle, which reaches a given algebraic order of accuracy with substantially fewer sample points than the tensor-product latitude--longitude grids used by Equiformer-family architectures~\cite{liao2023equiformerv2,liao2026equiformerv3} and keeps the numerical equivariance error of the nonlinearity at machine precision (Table~\ref{tab:s2_full_equivariance} and Supplementary Table~\ref{tab:sm_s2_truncated_equivariance}).

The remaining architectural component addresses close-contact physics.
DPA4 decomposes the potential energy as a sum of a learned equivariant branch and an analytical short-range branch, $E(Z,R)=E_{\Theta}^{\mathrm{NN}}(Z,R)+E^{\mathrm{ZBL}}(Z,R)$, with the analytical branch given by the Ziegler--Biersack--Littmark (ZBL) screened-Coulomb pair potential~\cite{ziegler1985}.
Rather than splicing the ZBL term onto the energy post hoc, DPA4 couples the two branches inside the energy model through native ZBL zone bridging: smooth bridging gates suppress the direct learned short-range pair channel so that the inner-zone pair interaction is handled exclusively by the analytical branch.
This removes the unphysical switching force, proportional to the energy mismatch of the two branches, that a coordinate-dependent splicing weight would introduce.
Both branches contribute to the same scalar total energy and are differentiated jointly, giving conservative forces with a smooth transition at close approach (Section~\ref{sec:zbl_method}).
Inside the inner zone the learned branch contributes no force, so the short-range force is by construction the analytical ZBL force; a C--Si dimer scan confirms this and shows the force-switching artifact that an external pair correction produces instead (Supplementary Section~\ref{sec:sm-zbl}).

\begin{figure}[!t]
    \centering
    \includegraphics[width=\textwidth]{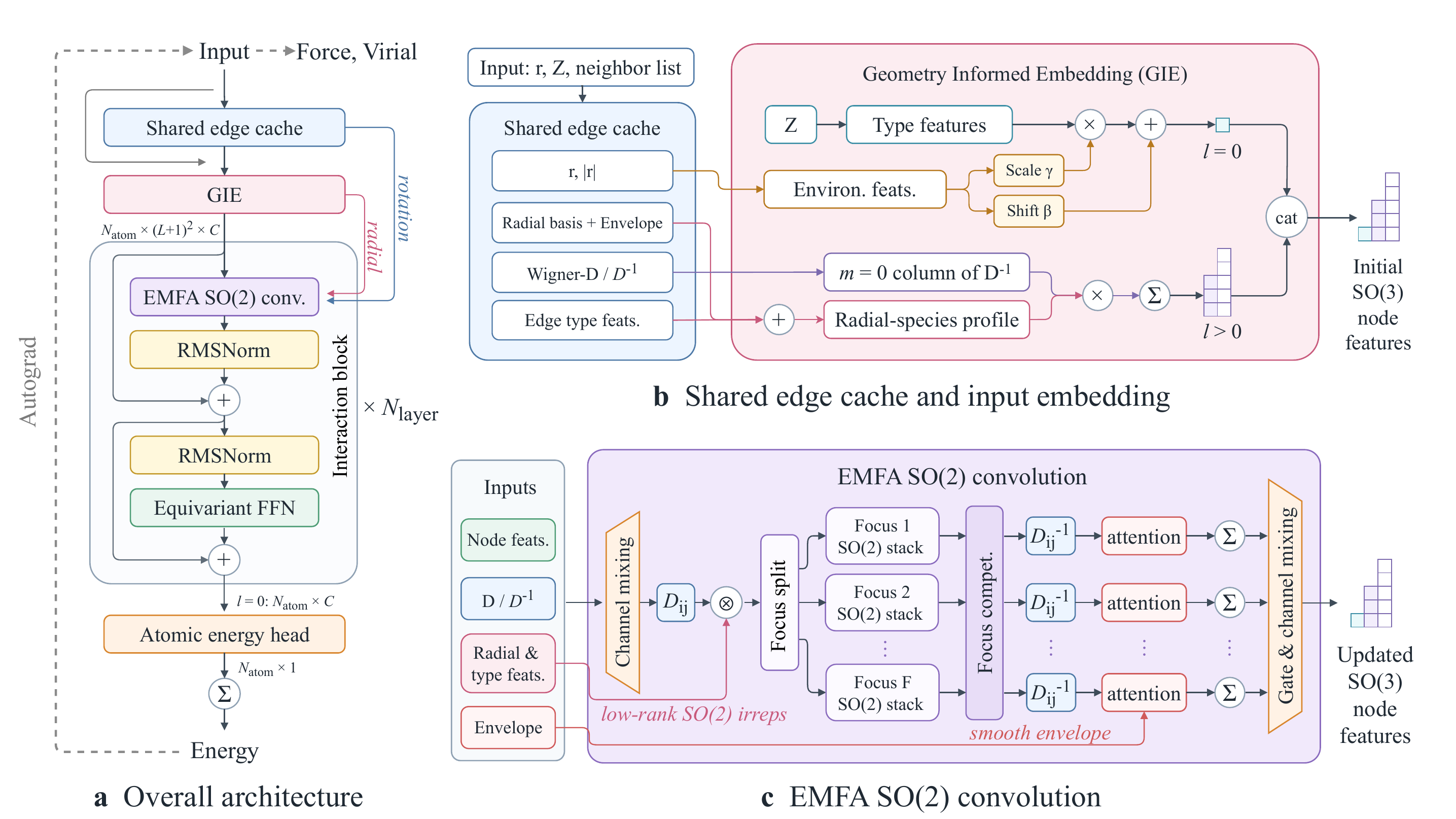}
    \caption{
        Overview of the DPA4 architecture.
        (a) The full model couples a shared edge cache, a Geometry-Informed Embedding (GIE) stage, $N_{\mathrm{layer}}$ stacked equivariant interaction blocks (each a residual stack of an EMFA (Edge--node degree coupling, Multi-Focus, Attention) SO(2) convolution and an equivariant feed-forward network (FFN) interleaved with equivariant RMS norms), and an atomic energy head.
        Architectural designs A1--A3 act inside the EMFA SO(2) convolution, and the equivariant FFN is instantiated as a Wigner bilinear network (A4).
        (b) The shared edge cache precomputes per-edge quantities; the GIE stage initializes the $l=0$ slice through Feature-wise Linear Modulation (FiLM) by a local-environment descriptor and the $l\ge 1$ slices by projecting neighbor directions onto real spherical harmonics weighted by a radial-species profile.
        (c) The EMFA SO(2) convolution transports features into an edge-local SO(2) frame, applies the low-rank edge--node SO(2)-equivariant degree coupling (A1) followed by the multi-focus design with cross-focus competition (A2), lifts back to the global frame, and aggregates over neighbors through envelope-gated attention (A3).
    }
    \label{fig:arch}
\end{figure}

\subsection{Matbench Discovery materials benchmark}
\label{sec:matbench}

\begin{table}[p]
    \PaperTableStyle
    \begin{threeparttable}
        \caption{Matbench Discovery leaderboard performance of DPA4 and compliant baseline models.}
        \label{tab:matbench}
        \begin{tabular*}{\linewidth}{@{\extracolsep{\fill}} l *{8}{c} @{}}
            \toprule
            \textbf{Model}       & \textbf{CPS}$\uparrow$\tnote{a} & \textbf{F1}$\uparrow$ & \textbf{$\kappa$SRME}$\downarrow$\tnote{b} & \textbf{RMSD}$\downarrow$\tnote{b} & \textbf{Params} & \textbf{Train time}\tnote{c} & \textbf{DeNS}\tnote{d} & \textbf{Targets}\tnote{e} \\
            \midrule
            DPA4-Pro             & \textbf{0.842}                  & \textbf{0.865}        & \textbf{0.227}                             & \textbf{0.068}                     & 25.2M           & 233                           & \xmark                 & EFSG                      \\
            DPA4-Plus            & 0.829                           & 0.847                 & \underline{0.240}                          & \underline{0.070}                  & 8.8M            & 31                            & \xmark                 & EFSG                      \\
            DPA4-Air             & 0.816                           & 0.837                 & 0.267                                      & 0.074                              & 5.1M            & 16                            & \xmark                 & EFSG                      \\
            DPA4-Neo             & 0.782                           & 0.810                 & 0.353                                      & 0.079                              & 1.1M            & 8.6                           & \xmark                 & EFSG                      \\
            DPA4-Mini            & 0.733                           & 0.776                 & 0.496                                      & 0.084                              & 0.66M           & 5.8                           & \xmark                 & EFSG                      \\
            \addlinespace[0.25em]
            \midrule
            \addlinespace[0.15em]
            EquiformerV3+DeNS-MP & \underline{0.830}               & \underline{0.863}     & 0.275                                      & \underline{0.070}                  & 30.3M           & 157                           & \cmark                 & EFSG                      \\
            eSEN-30M-MP          & 0.797                           & 0.831                 & 0.340                                      & 0.075                              & 30.1M           & 335                           & \cmark                 & EFSG                      \\
            MatRIS-10M-MP        & 0.778                           & 0.847                 & 0.489                                      & 0.072                              & 10.4M           & 224                           & \cmark                 & EFSGM                     \\
            Eqnorm MPtrj         & 0.756                           & 0.786                 & 0.408                                      & 0.084                              & 1.31M           & 83                            & \xmark                 & EFSG                      \\
            Nequix MP PFT        & 0.755                           & 0.748                 & 0.307                                      & 0.087                              & 708k            & 27\tnote{f}                   & \xmark                 & EFSHG                     \\
            Nequip-MP-L          & 0.730                           & 0.761                 & 0.466                                      & 0.086                              & 9.6M            & 41                            & \xmark                 & EFSG                      \\
            Nequix MP            & 0.729                           & 0.751                 & 0.446                                      & 0.085                              & 708k            & 21                            & \xmark                 & EFSG                      \\
            Allegro-MP-L         & 0.720                           & 0.751                 & 0.504                                      & 0.082                              & 18.7M           & 57                            & \xmark                 & EFSG                      \\
            DPA-3.1-MPtrj        & 0.718                           & 0.803                 & 0.650                                      & 0.080                              & 4.81M           & 104                           & \xmark                 & EFSG                      \\
            SevenNet-l3i5        & 0.714                           & 0.760                 & 0.550                                      & 0.085                              & 1.17M           & 381                           & \xmark                 & EFSG                      \\
            HIENet               & 0.707                           & 0.777                 & 0.642                                      & 0.080                              & 7.51M           & 120                           & \xmark                 & EFSG                      \\
            GRACE-2L-MPtrj       & 0.681                           & 0.691                 & 0.526                                      & 0.090                              & 15.3M           & --                            & \xmark                 & EFSG                      \\
            MACE-MP-0            & 0.637                           & 0.669                 & 0.682                                      & 0.092                              & 4.69M           & 108                           & \xmark                 & EFSG                      \\
            eqV2 S DeNS          & 0.522                           & 0.815                 & 1.676                                      & 0.076                              & 31.2M           & 228                           & \cmark                 & EFSD                      \\
            ORB v2 MPtrj         & 0.470                           & 0.765                 & 1.726                                      & 0.101                              & 25.2M           & --                            & \xmark                 & EFSD                      \\
            M3GNet               & 0.428                           & 0.569                 & 1.409                                      & 0.112                              & 228k            & --                            & \xmark                 & EFSG                      \\
            CHGNet               & 0.400                           & 0.613                 & 1.717                                      & 0.095                              & 413k            & --                            & \xmark                 & EFSGM                     \\
            \bottomrule
        \end{tabular*}
        \begin{tablenotes}[flushleft]
            \PaperTableNotesStyle
            \item[] The five DPA4 variants are trained on MPtrj in the compliant setting and compared with all compliant leaderboard entries on CPS, F1, $\kappa$SRME and RMSD, together with parameter counts and train times.
            \item[] Leaderboard entries comprise all compliant models accessed before August 1, 2026.
            \item[] Boldface and underlining denote the best and second-best values of each ranked metric.
            \item[a] CPS, the combined performance score of the leaderboard, aggregates F1, $\kappa$SRME and RMSD with weights 0.5, 0.4 and 0.1~\cite{riebesell2025matbench}.
            \item[b] $\kappa$SRME is the symmetric relative mean error of the phonon thermal conductivity~\cite{pota2024thermal}; RMSD is the mean structural deviation from the DFT-relaxed geometry, in \AA.
            \item[c] Train time is the aggregate accelerator time, reported in A100-equivalent GPU-days. It is calculated from the device count and wall-clock time recorded for each model~\cite{riebesell2025matbench} and rescaled by the theoretical FP32 peak throughput of the hardware used.
            A dash indicates that no train-time estimate is available.
            \item[d] \cmark{} marks models trained with the DeNS denoising objective~\cite{liao2024dens}.
            \item[e] Targets denotes the quantities each model is trained to predict, following the Matbench Discovery convention~\cite{riebesell2025matbench}: E, F, S, M and H denote energy, forces, stress, magnetic moments and Hessian, while the suffix G or D indicates whether forces and stress are obtained by energy gradient (conservative) or direct prediction.
            All DPA4 variants are EFSG.
            \item[f] Nequix MP PFT is obtained by fine-tuning Nequix MP on force constants from
            the MDR phonon database, co-trained with MPtrj~\cite{koker2026pft}; the reported value is
            the sum of the base-model and fine-tuning budgets, 500 and 140 GPU-hours, respectively.
        \end{tablenotes}
    \end{threeparttable}
\end{table}

Matbench Discovery is a widely used benchmark for evaluating machine-learning models in high-throughput inorganic-crystal discovery~\cite{riebesell2025matbench}.
In the compliant setting, models are trained on MPtrj~\cite{deng2023chgnet}, used to relax WBM candidate structures~\cite{WBM}, and then evaluated by their formation-energy predictions and derived convex-hull distances.
The resulting scores combine classification, regression and structural-relaxation metrics, making the benchmark a stringent test of both energy accuracy and relaxation quality.
The leaderboard also reports $\kappa$SRME, a thermal-conductivity metric sensitive to the smoothness and conservativeness of the learned potential~\cite{pota2024thermal}.
Here we focus on the compliant leaderboard to keep the training data fixed across models and more directly compare architectural differences.
Non-compliant leaderboard entries may additionally involve different training datasets or fine-tuning strategies and are left for future comparison.

Table~\ref{tab:matbench} reports the Matbench Discovery leaderboard metrics together with the parameter count and train time of each model.
Figure~\ref{fig:cps-throughput} plots the same score against saturated inference throughput.
DPA4-Pro establishes the best CPS in Table~\ref{tab:matbench}, reaching 0.842 with an F1 score of 0.865, a $\kappa$SRME of 0.227 and an RMSD of 0.068~\AA{} at 25.2M parameters.
It is the leading compliant model in every ranked metric of the leaderboard, spanning stability classification, thermal conductivity and relaxation quality.
Its CPS exceeds that of EquiformerV3+DeNS-MP~\cite{liao2026equiformerv3,liao2024dens} by 0.012 (0.842 versus 0.830) while using 17\% fewer parameters (25.2M versus 30.3M).

This comparison is notable because several high-ranking baselines use additional accuracy-enhancing training stages, marked in the DeNS column of Table~\ref{tab:matbench}.
DeNS~\cite{liao2024dens} has been shown to improve equivariant force fields by denoising non-equilibrium structures, and EquiformerV3+DeNS-MP~\cite{liao2026equiformerv3} further combines this strategy with direct-force pretraining.
DPA4-Pro uses neither DeNS nor direct-force training; it is trained through the conservative energy-gradient path and still surpasses the DeNS-assisted EquiformerV3 baseline in CPS.
This suggests that the accuracy advantage originates in the architecture itself rather than in an auxiliary denoising or direct-force objective.

The smaller DPA4 variants extend this trade-off across model scales: all five variants evaluated on Matbench Discovery lie on the Pareto frontier in CPS and saturated inference throughput (Fig.~\ref{fig:cps-throughput}), and each matches or surpasses much larger baselines at a small fraction of their model size and training budget.
DPA4-Plus reaches a CPS of 0.829 with 8.8M parameters, matching EquiformerV3+DeNS-MP to within 0.001 CPS while running 11.9$\times$ faster at saturated inference, using 5.1$\times$ less training compute (31 versus 157 A100-equivalent GPU-days) and 3.4$\times$ fewer parameters.
DPA4-Air reaches a CPS of 0.816 with 5.1M parameters, exceeding eSEN-30M-MP~\cite{fu2025esen} (0.797 CPS, 30.1M parameters) with 20.9$\times$ less training compute (16 versus 335 GPU-days) and a 5.9$\times$ smaller model.
DPA4-Neo contains only 1.1M parameters but still reaches a CPS of 0.782, above MatRIS-10M-MP~\cite{zhou2026matris} (0.778 CPS, 10.4M parameters) with a saturated throughput higher by more than two orders of magnitude (28.0 against 0.055~atoms/ms), 26$\times$ less training compute (8.6 versus 224 GPU-days) and a 9.5$\times$ smaller model.
DPA4-Mini, the smallest variant at 0.66M parameters, reaches a CPS of 0.733 at a training budget of 5.8 GPU-days, above the Nequip-MP-L and Nequix MP entries of Table~\ref{tab:matbench} (0.730 and 0.729 CPS), running 23.7$\times$ and 4.3$\times$ faster at saturated inference and at 7.1$\times$ and 3.6$\times$ less training compute.

Within the DPA4 series, the accuracy--throughput trade-off of Fig.~\ref{fig:cps-throughput} flattens sharply at the top end.
The Mini--Neo, Neo--Air and Air--Plus steps trade 2.8$\times$, 2.4$\times$ and 2.3$\times$ in saturated inference throughput for CPS gains of 0.049, 0.034 and 0.013, whereas the step from DPA4-Plus to DPA4-Pro costs 3.8$\times$ (5.1 to 1.4~atoms/ms) and gains only 0.013.
DPA4-Plus therefore sits at the knee of the frontier: it matches the accuracy of the strongest baseline at an order of magnitude higher throughput, while DPA4-Pro converts a disproportionate share of throughput into the remaining margin at the top of the leaderboard, itself still 3.2$\times$ faster than EquiformerV3+DeNS-MP.
On the training side, cost grows sublinearly with model size at the lower end, rising by only 1.9$\times$ from Neo to Air against a 4.6$\times$ increase in parameters, because neighbor-list construction, geometric preprocessing and memory traffic occupy a larger fraction of each step once the arithmetic workload no longer saturates the GPU.
For perspective, the A100 has a peak FP32 throughput of 19.5~TFLOPS, whereas recent single-card workstation GPUs can exceed 100~TFLOPS.
Rescaled by the peak-FLOP ratio, the 16 GPU-day budget of DPA4-Air corresponds to roughly three days on such hardware, which makes DPA4-Air and DPA4-Neo practical for high-throughput workflows.

With the leading compliant models now separated by only small CPS margins, substantial further gains in accuracy are likely to require training data beyond the fixed MPtrj set.
The compliant setting nevertheless remains the appropriate basis for the present comparison: holding the training data fixed across models attributes the observed CPS differences to the architectures themselves.
Within this setting, the accuracy--cost frontier remains open in the cost direction, namely how far equal accuracy can be attained at further reduced inference and training cost.

\subsection{OMat24 inorganic materials benchmark}
\label{sec:omat24}

\begin{table}[t]
    \PaperTableStyle
    \begin{threeparttable}
        \caption{Performance on the OMat24 validation set.}
        \label{tab:omat24}
        \begin{tabular*}{\linewidth}{@{\extracolsep{\fill}} l c c c c c c @{}}
            \toprule
            \textbf{Model}                                                &
            \textbf{Prediction}\tnote{a}                                  &
            \textbf{Energy}\tnote{b}                                      &
            \textbf{Force}\tnote{b}                                       &
            \textbf{Stress}\tnote{b}                                      &
            \textbf{Params}\tnote{c}                                      &
            \textbf{Train time}\tnote{d}                                                                                           \\
            \midrule
            MACE-Omat-Small\tnote{e}~\cite{batatia2022mace,batatia2025crosslearning} & Gradient    & 17.9          & 85.9                & 3.5              & 8.2M       & --               \\
            MACE-Omat-Medium\tnote{e}                                      & Gradient    & 16.3          & 78.4                & 3.3              & 9.1M       & --               \\
            EquiformerV2-S~\cite{liao2023equiformerv2,barroso2024open}     & Direct      & 11.0          & 49.2                & \underline{2.4}  & 31M        & --               \\
            EquiformerV2-M                                                & Direct      & 10.0          & 44.8                & \textbf{2.3}     & 87M        & --               \\
            EquiformerV2-L                                                & Direct      & \underline{9.6} & 43.1              & \textbf{2.3}     & 154M       & --               \\
            EquiformerV3, $L_{\max}=4$~\cite{liao2026equiformerv3}         & Direct+Grad & 10.4          & 43.5                & 2.6              & 30M        & 4,782            \\
            EquiformerV3, $L_{\max}=6$                                    & Direct+Grad & 10.1          & \textbf{41.6}       & 2.5              & 49M        & 9,197            \\
            eSEN~\cite{fu2025esen}                                        & Gradient    & 10.7          & 47.3                & 2.6              & 30M        & --               \\
            UMA-S~\cite{wood2025umafamilyuniversalmodels}                 & Direct+Grad & 11.3          & 57.1                & 2.9              & 150M (6M)  & 46,080\tnote{f}  \\
            UMA-M                                                         & Direct+Grad & 10.0          & 47.3                & 2.7              & 1.4B (50M) & 129,024\tnote{f} \\
            UMA-L                                                         & Direct      & 9.7           & 43.5                & 2.5              & 700M       & 95,232\tnote{f}  \\
            \addlinespace[0.25em]
            \midrule
            \addlinespace[0.15em]
            DPA4-Nano                                                     & Gradient    & 18.7          & 95.8                & 3.4              & 0.48M      & 54               \\
            DPA4-Mini                                                     & Gradient    & 14.0          & 70.7                & 2.9              & 0.66M      & 80               \\
            DPA4-Neo                                                      & Gradient    & 12.3          & 58.7                & 2.7              & 1.1M       & 139              \\
            DPA4-Air                                                      & Gradient    & 10.6          & 51.0                & 2.6              & 5.1M       & 207              \\
            DPA4-Plus                                                     & Gradient    & 9.8           & 45.6                & 2.5              & 8.8M       & 368              \\
            DPA4-Pro                                                      & Gradient    & \textbf{9.4}  & \underline{42.7}    & \underline{2.4}  & 25.2M      & 1,474            \\
            \bottomrule
        \end{tabular*}
        \begin{tablenotes}[flushleft]
            \PaperTableNotesStyle
            \item[] The six DPA4 size classes are trained on OMat24 through the conservative energy-gradient objective and compared with published baselines.
            \item[] Boldface and underlining denote the best and second-best values of each evaluation metric; values that coincide at the reported precision carry the same mark.
            \item[a] Prediction denotes the force-prediction pathway: Gradient obtains forces as energy gradients, Direct predicts forces directly, and Direct+Grad denotes direct-force pre-training followed by gradient fine-tuning.
            \item[b] Values are mean absolute errors in meV per atom for energy, meV/\AA\ for forces and meV/\AA$^3$ for stress; lower values are better.
            \item[c] Params reports total parameters; parentheses give active parameters for UMA MoLE models.
            \item[d] Train time is the aggregate accelerator time, reported in H100-equivalent GPU-hours. DPA4 H20 runs are converted using the ratio of the theoretical FP32 peak throughputs of H20 and H100 GPUs. UMA runs on H200 GPUs are counted at parity, as the H200 and H100 share the same theoretical FP32 peak throughput.
            \item[e] No OMat24 validation metrics are published for MACE-Omat; the values listed here were obtained in this work by evaluating the released checkpoints through the ASE interface, and no train-time estimate is available.
            \item[f] The UMA train times include multitask pre-training and fine-tuning and therefore are not OMat24-only estimates.
        \end{tablenotes}
    \end{threeparttable}
\end{table}

OMat24~\cite{barroso2024open} is a large-scale inorganic materials dataset of 118 million DFT calculations, obtained by non-equilibrium sampling --- rattling, ab initio molecular dynamics and the relaxation of rattled structures --- starting from structures in the Alexandria database~\cite{schmidt2023alexandria}.
Its validation split has become a common reference point for single-dataset comparisons of recent equivariant architectures.
Because the dataset is much larger and more diverse than MPtrj, it separates architectures by their capacity to fit a broad potential-energy surface.
We therefore train the six DPA4 size classes on OMat24 with the architecture configurations of Supplementary Table~\ref{tab:sm_model_config} and report energy, force and stress MAEs on the validation set.

Table~\ref{tab:omat24} shows that DPA4-Pro attains the lowest energy MAE of all models compared, 9.4~meV/atom, with 25.2M parameters.
The next most accurate models, EquiformerV2-L~\cite{liao2023equiformerv2,barroso2024open} (9.6~meV/atom, 154M parameters) and UMA-L~\cite{wood2025umafamilyuniversalmodels} (9.7~meV/atom, 700M parameters), are 6.1$\times$ and 27.8$\times$ larger.
Its force MAE of 42.7~meV/\AA{} is the second lowest, behind EquiformerV3 with $L_{\max}=6$~\cite{liao2026equiformerv3} (41.6~meV/\AA{}, 49M parameters), and only the larger EquiformerV2 models reach a lower stress MAE (2.3 against 2.4~meV/\AA$^3$).
All DPA4 entries predict forces as energy gradients, whereas the strongest baselines rely on direct force prediction or on direct-force pre-training followed by gradient fine-tuning.
The training budgets also differ substantially: DPA4-Pro required 1,474 H100-equivalent GPU-hours, 6.2$\times$ below the 9,197 of EquiformerV3 with $L_{\max}=6$.

Across the six DPA4 size classes, increasing the parameter count from 0.48M to 25.2M and the training cost from 54 to 1,474 H100-equivalent GPU-hours monotonically reduces the energy MAE from 18.7 to 9.4~meV/atom, the force MAE from 95.8 to 42.7~meV/\AA{} and the stress MAE from 3.4 to 2.4~meV/\AA$^3$.
DPA4-Nano, DPA4-Mini and DPA4-Neo occupy the compact 0.48--1.1M-parameter regime at costs of only 54--139 H100-equivalent GPU-hours; within this regime, Mini and Neo reach energy MAEs of 14.0 and 12.3~meV/atom, below the 16.3~meV/atom of the larger 9.1M-parameter MACE-Omat-Medium.
DPA4-Air then matches the 30M-parameter eSEN baseline~\cite{fu2025esen} in energy (10.6 versus 10.7~meV/atom) with 5.9$\times$ fewer parameters, while DPA4-Plus reaches 9.8~meV/atom and 45.6~meV/\AA{} with 8.8M parameters, a lower energy MAE than every baseline except EquiformerV2-L and UMA-L, both more than an order of magnitude larger.
At 368 H100-equivalent GPU-hours, DPA4-Plus costs 4.0$\times$ less to train than DPA4-Pro for an energy MAE only 0.4~meV/atom higher, making it, as on Matbench Discovery, the most favorable accuracy--cost balance in the series.
Together with the Matbench Discovery results, this shows that the accuracy--parameter advantage of DPA4 persists when the training set is enlarged from MPtrj to OMat24.

\subsection{OMol25 molecular benchmark}
\label{sec:omol25}

\begin{table}[t]
    \PaperTableStyle
    \begin{threeparttable}
        \caption{Performance on the OMol25 OMol-0 out-of-distribution composition validation split.}
        \label{tab:omol25}
        \begin{tabular*}{\linewidth}{@{\extracolsep{\fill}} l c c c c c c @{}}
            \toprule
            \textbf{Model}                       &
            \textbf{Prediction}\tnote{a}         &
            \textbf{Cons.}\tnote{a}              &
            \textbf{Energy}\tnote{b}             &
            \textbf{Force}\tnote{b}              &
            \textbf{Params}\tnote{c}             &
            \textbf{Train time}\tnote{d}          \\
            \midrule
            MACE-OMol-L-0~\cite{batatia2022mace,levine2025open}                       & Gradient    & \cmark & 4.56             & 0.25              & 52.4M\tnote{e} & --               \\
            GemNet-OC-r6\tnote{f}~\cite{gasteiger2022gemnet,levine2025open}           & Direct      & \xmark & 1.04             & 0.17              & 39.1M      & --               \\
            GemNet-OC\tnote{f}~\cite{gasteiger2022gemnet,levine2025open}              & Direct      & \xmark & \textbf{0.78}    & 0.15              & 39.1M      & --               \\
            eSEN-sm-d.~\cite{fu2025esen,levine2025open}                              & Direct      & \xmark & 2.06             & 0.23              & 6.3M       & --               \\
            eSEN-sm-cons.~\cite{fu2025esen,levine2025open}                           & Gradient    & \cmark & 1.77             & 0.19              & 6.3M       & --               \\
            eSEN-md-d.~\cite{fu2025esen,levine2025open}                              & Direct      & \xmark & 1.15             & \textbf{0.11}    & 50.7M      & --               \\
            UMA-S-1.1~\cite{wood2025umafamilyuniversalmodels,levine2025open}          & Direct+Grad & \cmark & 1.45             & 0.22              & 150M (6M)  & 46,080\tnote{g}  \\
            UMA-M-1.1~\cite{wood2025umafamilyuniversalmodels,levine2025open}          & Direct+Grad & \cmark & 1.14             & 0.14              & 1.4B (50M) & 129,024\tnote{g} \\
            \addlinespace[0.25em]
            \midrule
            \addlinespace[0.15em]
            DPA4-Nano                                      & Gradient    & \cmark & 8.02             & 0.776             & 0.48M      & 186              \\
            DPA4-Mini                                      & Gradient    & \cmark & 4.97             & 0.502             & 0.66M      & 431              \\
            DPA4-Neo                                       & Gradient    & \cmark & 2.91             & 0.332             & 1.1M       & 819              \\
            DPA4-Air                                       & Gradient    & \cmark & 1.45             & 0.189             & 5.1M       & 1,549            \\
            DPA4-Plus                                      & Gradient    & \cmark & 1.20             & 0.146             & 8.8M       & 2,521            \\
            DPA4-Pro                                       & Gradient    & \cmark & \underline{0.92} & \underline{0.117} & 25.2M      & 8,829            \\
            \bottomrule
        \end{tabular*}
        \begin{tablenotes}[flushleft]
            \PaperTableNotesStyle
            \item[] A dash indicates that a value is unavailable.
            \item[] Boldface and underlining denote the best and second-best values within each MAE column.
            \item[a] Prediction denotes the force-prediction pathway: Gradient obtains forces as energy gradients, Direct predicts forces directly, and Direct+Grad denotes direct-force pre-training followed by gradient fine-tuning. Cons.\ marks whether the evaluated forces are conservative, i.e.\ energy gradients: \cmark{} for Gradient and Direct+Grad models, \xmark{} for Direct models.
            \item[b] Values are mean absolute errors (MAEs), reported in kcal/mol for energy and kcal/mol/\AA{} for forces. Values originally reported in meV and meV/\AA{} are converted using $1~\mathrm{meV}=0.0230605~\mathrm{kcal/mol}$ and $1~\mathrm{meV}/\text{\AA}=0.0230605~\mathrm{kcal/mol}/\text{\AA}$, respectively.
            \item[c] Params reports total parameters; parentheses give active parameters for UMA MoLE models.
            \item[d] Train time is the aggregate accelerator time, reported in H100-equivalent GPU-hours. DPA4 H20 runs are converted using the ratio of the theoretical FP32 peak throughputs of H20 and H100 GPUs. UMA runs on H200 GPUs are counted at parity, as the H200 and H100 share the same theoretical FP32 peak throughput.
            \item[e] The MACE-OMol-L-0 parameter count was computed from the released \texttt{MACE-omol-0-extra-large-1024.model} checkpoint distributed through the \href{https://github.com/ACEsuit/mace}{MACE repository}.
            \item[f] GemNet-OC-r6 uses a 6~\AA{} cutoff radius, whereas GemNet-OC uses the default 12~\AA{} cutoff radius.
            \item[g] The listed UMA train times are for the multitask UMA 1.0 runs and are not OMol25-only estimates. UMA-1.1 was obtained by further fine-tuning these models, but the cost of that step is not reported.
        \end{tablenotes}
    \end{threeparttable}
\end{table}

OMol25~\cite{levine2025open} is a molecular dataset of DFT single-point calculations at the $\omega$B97M-V/def2-TZVPD level, spanning biomolecules, metal complexes, electrolytes and recomputed community datasets, with charge and spin states as explicit inputs.
It is substantially broader in chemistry than SPICE-MACE-OFF and therefore probes transferability across molecular classes rather than accuracy within a single class.
The dataset has been released in successive versions; we train on the OMol-0 release, which the dataset paper describes as containing about 100 million DFT datapoints and which is also the release used for the UMA-1.1 models~\cite{levine2025open}, and we evaluate on its out-of-distribution composition validation split.

Across the six DPA4 size classes, increasing the parameter count from 0.48M to 25.2M monotonically decreases the energy and force MAEs from 8.02 to 0.92~kcal/mol and from 0.776 to 0.117~kcal/mol/\AA{}, respectively.
Across the series, the training budget increases from 186 to 8,829 H100-equivalent GPU-hours from Nano through Pro.
At the compact end, DPA4-Nano, DPA4-Mini and DPA4-Neo span 0.48--1.1M parameters, energy MAEs of 8.02--2.91~kcal/mol and training costs of 186--819 H100-equivalent GPU-hours.

DPA4-Air then matches the 1.45~kcal/mol energy MAE of UMA-S-1.1 at the reported precision and improves its force MAE from 0.22 to 0.189~kcal/mol/\AA{} with 5.1M parameters, 29$\times$ fewer than the 150M total parameters of UMA-S-1.1.
The strongest conservative baseline in Table~\ref{tab:omol25} is UMA-M-1.1, with energy and force MAEs of 1.14~kcal/mol and 0.14~kcal/mol/\AA{} at 1.4B total and 50M active parameters~\cite{wood2025umafamilyuniversalmodels,levine2025open}.
DPA4-Plus matches it at 1.20~kcal/mol and 0.146~kcal/mol/\AA{}, and DPA4-Pro surpasses it on both metrics at 0.92~kcal/mol and 0.117~kcal/mol/\AA{}, trained with 2,521 and 8,829 H100-equivalent GPU-hours, respectively.
UMA-M-1.0 required 129,024 GPU-hours for its two-stage training.
UMA-M-1.1 was obtained by further fine-tuning this model, but the cost of that step is not reported.
Non-conservative models predict energy and forces separately, free from the constraint that the forces are the gradients of the energy.
Among them, GemNet-OC attains the lowest energy MAE of the table at 0.78~kcal/mol~\cite{gasteiger2022gemnet,levine2025open}, and eSEN-md-d.\ attains the lowest force MAE at 0.11~kcal/mol/\AA{}~\cite{fu2025esen,levine2025open}.
Neither, however, surpasses DPA4-Pro on energy and force simultaneously.
The monotonic error reduction across the DPA4 size series extends the scaling trends observed on Matbench Discovery and OMat24 to both OMol25 energy and force.

\subsection{SPICE-MACE-OFF molecular benchmark}
\label{sec:small_mol}

SPICE-MACE-OFF~\cite{kovacs2023mace} is a small-molecule benchmark for transferable organic force fields.
The benchmark spans PubChem molecules, DES370K~\cite{donchev2021des370k} monomers and dimers, dipeptides, solvated amino acids, water clusters and larger QMugs-derived molecules~\cite{eastman2023spice,isert2022qmugs}.
Reference energies and forces are computed at the $\omega$B97M-D3(BJ)/def2-TZVPPD level~\cite{eastman2023spice,kovacs2023mace}.
We train the DPA4-Neo, DPA4-Air and DPA4-Plus variants on SPICE-MACE-OFF, using the same train/validation/test split as MACE-OFF~\cite{kovacs2023mace} and DPA3~\cite{zhang2025dpa3}, and report per-subset energy and force MAEs.
Following DPA3~\cite{zhang2025dpa3}, we also report a logarithmic weighted average MAE (LWAMAE) with equal weights, equivalent to the geometric mean of subset MAEs.

\begin{table}[tp]
    \PaperTableStyle
    \begin{threeparttable}
        \caption{SPICE-MACE-OFF small-molecule benchmark performance.}
        \label{table:spice}
        \begin{tabular*}{\linewidth}{@{\extracolsep{\fill}} l *{8}{c} @{}}
            \toprule
            \textbf{Dataset} &
            \makecell{\textbf{MACE}\\\textbf{(M)}} &
            \makecell{\textbf{MACE}\\\textbf{(L)}} &
            \makecell{\textbf{eSEN}\\\textbf{(3.2M)}} &
            \makecell{\textbf{eSEN}\\\textbf{(6.5M)}} &
            \makecell{\textbf{DPA3-}\\\textbf{L24}} &
            \makecell{\textbf{DPA4-}\\\textbf{Neo}} &
            \makecell{\textbf{DPA4-}\\\textbf{Air}} &
            \makecell{\textbf{DPA4-}\\\textbf{Plus}} \\
            \midrule
            \multicolumn{9}{@{}l}{\textit{Energy MAE (meV/atom)}} \\
            \addlinespace[0.15em]
            PubChem         & 0.91 & 0.88 & 0.22 & \underline{0.15} & 0.24 & 0.24 & \underline{0.13} & \textbf{0.12} \\
            DES370K M.      & 0.63 & 0.59 & 0.17 & \underline{0.13} & 0.18 & 0.17 & 0.17             & \textbf{0.12} \\
            DES370K D.      & 0.58 & 0.54 & 0.20 & \underline{0.15} & 0.23 & 0.18 & \textbf{0.11}    & \textbf{0.11} \\
            Dipeptides      & 0.52 & 0.42 & 0.10 & 0.07             & 0.13 & 0.10 & \underline{0.06} & \textbf{0.04} \\
            Sol.
            AA & 1.21 & 0.98 & 0.30 & 0.25 & 0.31 & 0.27 & \underline{0.13} & \textbf{0.09} \\ Water & 0.76 & 0.83 & 0.24 & \underline{0.15} & 0.32 & 0.30 & 0.16 & \textbf{0.11} \\ QMugs & 0.69 & 0.54 & 0.16 & 0.12 & 0.17 & 0.12 & \underline{0.07} & \textbf{0.05} \\ LWAMAE\tnote{a} & 0.73 & 0.65 & 0.19 & 0.14 & 0.22 & 0.18 & \underline{0.11} & \textbf{0.09} \\ \addlinespace[0.25em] \midrule \addlinespace[0.15em] \multicolumn{9}{@{}l}{\textit{Force MAE (meV/\AA)}} \\ \addlinespace[0.15em] PubChem & 20.57 & 14.75 & 6.10 & 4.21 & 8.47 & 5.92 & \underline{3.59} & \textbf{2.98} \\ DES370K M.
            & 9.36  & 6.58  & 1.85 & 1.24 & 3.15 & 2.13 & \underline{1.22} & \textbf{0.91} \\
            DES370K D.      & 9.02  & 6.62  & 2.77 & 2.12 & 3.19 & 2.16 & \underline{1.23} & \textbf{0.96} \\
            Dipeptides      & 14.27 & 10.19 & 3.04 & 2.00 & 4.81 & 3.10 & \underline{1.68} & \textbf{1.25} \\
            Sol. AA         & 23.26 & 19.43 & 5.76 & 3.68 & 8.77 & 6.69 & \underline{3.59} & \textbf{2.59} \\
            Water           & 15.27 & 13.57 & 3.88 & 2.50 & 6.89 & 4.46 & \underline{2.47} & \textbf{1.78} \\
            QMugs           & 23.58 & 16.93 & 5.70 & 3.78 & 8.66 & 5.36 & \underline{2.91} & \textbf{2.14} \\
            LWAMAE\tnote{a} & 15.42 & 11.66 & 3.83 & 2.58 & 5.78 & 3.89 & \underline{2.18} & \textbf{1.64} \\
            \midrule
            \textbf{Params}                  & 2.3M & 6.9M & 3.2M & 6.5M & 4.9M & 1.1M & 5.1M & 8.8M \\
            \textbf{Train time}\tnote{b}       & 10   & 14   & --   & --   & 288  & 7    & 7.8  & 9 \\
            \bottomrule
        \end{tabular*}
        \begin{tablenotes}[flushleft]
            \PaperTableNotesStyle
            \item[] The DPA4-Neo, DPA4-Air and DPA4-Plus variants are compared with MACE-OFF, eSEN and DPA3 baselines, together with parameter counts and train times.
            \item[] Reported values are test-set MAEs on the SPICE-MACE-OFF dataset~\cite{kovacs2023mace}; energy MAEs are in meV/atom and force MAEs are in meV/\AA.
            \item[] Boldface and underlining denote the best and second-best values for each dataset--target pair; values that coincide at the reported precision carry the same mark.
            \item[a] LWAMAE denotes the equal-weight geometric mean of subset MAEs, following the DPA3 comparison protocol~\cite{zhang2025dpa3}.
            \item[b] Train time is the aggregate accelerator time, reported in A100-equivalent GPU-days; a dash indicates that no public train-time estimate is available.
        \end{tablenotes}
    \end{threeparttable}
\end{table}

Table~\ref{table:spice} shows that DPA4 improves molecular energy and force accuracy across the chemically diverse SPICE-MACE-OFF subsets.
DPA4-Plus attains the lowest aggregate energy and force errors, with LWAMAEs of 0.09~meV/atom and 1.64~meV/\AA{}, respectively, and is best or tied for best on all 14 subset--target pairs at the reported precision.
Against the 6.5M-parameter eSEN baseline~\cite{fu2025esen}, this 8.8M-parameter model lowers the aggregate energy and force errors by 39\% and 36\%.
Against DPA3-L24~\cite{zhang2025dpa3}, the corresponding reductions are 61\% and 72\%, showing that the gain extends well beyond the inorganic-crystal benchmark.

The smaller variants preserve most of this accuracy at a fraction of the model size.
With 5.1M parameters, DPA4-Air reaches aggregate LWAMAEs of 0.11~meV/atom and 2.18~meV/\AA{}, which are 21\% and 16\% below the larger 6.5M-parameter eSEN baseline~\cite{fu2025esen} and 50\% and 62\% below DPA3-L24~\cite{zhang2025dpa3}.
It ranks first or second on 12 of the 14 subset--target pairs, and on the PubChem subset its force error is the lowest of all models compared.
DPA4-Neo extends the comparison to a still smaller scale: with 1.1M parameters it reaches LWAMAEs of 0.18~meV/atom and 3.89~meV/\AA{}, matching the 3.2M-parameter eSEN baseline~\cite{fu2025esen} with 2.9$\times$ fewer parameters and improving on DPA3-L24~\cite{zhang2025dpa3} by 18\% in energy and 33\% in force with 4.5$\times$ fewer parameters.
Beyond model size, these gains come at low training cost: DPA4-Neo, DPA4-Air and DPA4-Plus require only 7, 7.8 and 9 A100-equivalent GPU-days, respectively (Table~\ref{table:spice}), below the MACE baselines (10 and 14) and 32--41$\times$ smaller than the 288 GPU-days of DPA3-L24; no public training cost is reported for the eSEN baselines.

\subsection{Training and inference efficiency}
\label{sec:efficiency}

\begin{figure}[t]
    \centering
    \includegraphics[width=0.9\textwidth]{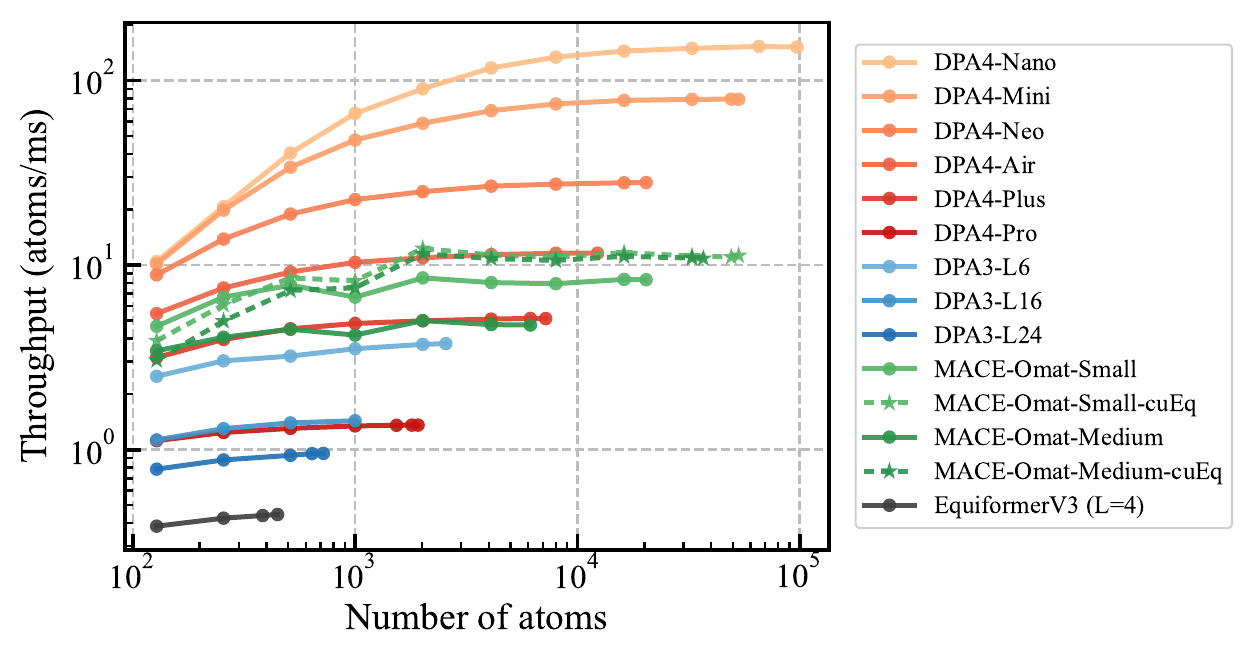}
    \caption{
        Inference throughput on diamond supercells at experimental density through the ASE calculator interface~\cite{ase-paper}.
        The six DPA4 size classes, from Nano to Pro, are compared with DPA3, MACE-Omat and EquiformerV3 baselines.
        Each point reports end-to-end throughput from the mean wall time of 100 evaluations of energy, forces and stress after 10 warm-up evaluations on a single NVIDIA H20 GPU.
        Diamond at the 6~\AA{} cutoff used by DPA4 carries 158 neighbors per atom, so the sweep varies the system size at a fixed and comparatively dense local environment.
        Each curve terminates at the largest supercell the corresponding model could evaluate within the memory of one GPU.
        cuEq denotes MACE inference with NVIDIA cuEquivariance-accelerated equivariant kernels~\cite{nvidia_cuequivariance}.
        Higher atom-normalized throughput indicates faster inference.
    }
    \label{fig:infer_diamond}
\end{figure}

The low training cost documented in the preceding sections rests on a systems contribution of this work: conservative energy-gradient training that runs entirely inside compiled GPU code.
The obstacle is the mixed second derivative required by force supervision --- the training gradient differentiates the energy-derived forces $-\partial E_{\Theta}/\partial\mathbf{R}$ once more with respect to the parameters --- which a standard compiled forward/backward stack does not expose inside the compiled region.
DPA4 resolves this by executing the energy-to-force derivative during symbolic tracing, so that the compiled graph itself contains the energy evaluation and its coordinate derivative, and the outer backward pass differentiates the force residual with respect to the parameters (Section~\ref{sec:compiled-training}).
Making this path robust requires a set of dedicated systems measures: a shape-stable compact edge representation under symbolic shapes, surgical removal of the autograd detach nodes that would sever the force-loss gradient, a compilation cache keyed by graph topology, and Inductor settings compatible with higher-order autograd (Section~\ref{sec:compiled-training}).
In controlled ablations, the compiled path under \texttt{torch.compile}~\cite{ansel2024pytorch2} with bf16 automatic mixed precision gives a 3.1$\times$ wall-clock training speedup and reduces peak training memory to about 40\% of the FP32 baseline (Table~\ref{tab:dpa4_compile_ablation}).
The gain is obtained without replacing energy-based force matching by a direct-force surrogate: the compiled path preserves the scalar energy-to-force relation of the uncompiled model.

Inference efficiency is evaluated through the ASE calculator interface~\cite{ase-paper}, which provides a common end-to-end route for single-point energy, force and stress evaluation across DPA4, DPA3, MACE and EquiformerV3 baselines.
The MACE baselines are evaluated both in their standard path and with NVIDIA cuEquivariance-accelerated equivariant kernels~\cite{nvidia_cuequivariance}.
The comparison sweeps diamond supercells at experimental density from a target size of 128 atoms upwards and therefore includes neighbor-list construction, model evaluation and calculator-interface overhead, while holding the local environment fixed at 158 neighbors per atom within the 6~\AA{} DPA4 cutoff.
All inference benchmarks are run on the same H20 hardware and software environment with the same adaptive size-search and timing protocol, as detailed in Supplementary Section~\ref{sec:sm-inference}.

Figure~\ref{fig:infer_diamond} shows that throughput rises with system size and approaches a saturated regime as the GPU becomes fully occupied.
Across the six DPA4 size classes, reducing the model scale from Pro to Nano increases saturated throughput monotonically from 1.36 to 153~atoms/ms and expands the largest evaluable system from 1{,}920 to 97{,}336 atoms.
DPA4-Nano therefore extends the throughput end of the family, reaching 1.93$\times$ the saturated throughput of DPA4-Mini and 112$\times$ that of DPA4-Pro.
For Fig.~\ref{fig:cps-throughput}, saturated throughput is taken as the largest atom-normalized throughput among all successful trials of the adaptive size search.
Because no Matbench Discovery CPS is reported for DPA4-Nano, the CPS--throughput comparison comprises the five DPA4 variants with available scores and ten compliant baselines.
Among these fifteen models, DPA4-Pro has the highest CPS, DPA4-Mini has the highest saturated throughput, and the three intermediate DPA4 variants are non-dominated between these endpoints.
The MACE-Omat and EquiformerV3 baselines of this sweep share their architecture and parameter count with the corresponding entries of Table~\ref{tab:omat24} (EquiformerV3 in its $L_{\max}=4$, 30M configuration), so accuracy and throughput can be read together.

Against MACE-Omat~\cite{batatia2022mace,batatia2025crosslearning}, the six DPA4 size classes expose a continuous accuracy--throughput trade-off rather than a uniform dominance relation.
At the throughput extreme, the 0.48M-parameter DPA4-Nano reaches 153~atoms/ms, 12.4$\times$ the saturated throughput of cuEquivariance-accelerated MACE-Omat-Small~\cite{nvidia_cuequivariance}, while its OMat24 energy and force MAEs are 4.5\% and 12\% higher, respectively.
The trade-off reverses already at DPA4-Mini: with 0.66M parameters, it lowers the OMat24 energy and force MAEs of the 8.2M-parameter MACE-Omat-Small by 22\% and 18\%, while running 9.3$\times$ faster than its standard implementation and 6.4$\times$ faster than its cuEquivariance-accelerated form.
DPA4-Neo, with 1.1M parameters, lowers both MAEs of the 9.1M-parameter MACE-Omat-Medium by 25\% and is still 2.4$\times$ faster than that baseline in its accelerated form.
DPA4-Air reaches 11.6~atoms/ms, matching the cuEquivariance-accelerated MACE-Omat-Medium result of 11.5~atoms/ms and remaining within 6\% of the accelerated MACE-Omat-Small result of 12.3~atoms/ms, at energy and force MAEs 35--41\% below theirs.

Against EquiformerV3~\cite{liao2026equiformerv3}, the same advantage appears in throughput and accessible system size.
DPA4-Pro is more accurate than the $L_{\max}=4$ configuration in all three targets (energy 9.4 against 10.4~meV/atom, force 42.7 against 43.5~meV/\AA{}, stress 2.4 against 2.6~meV/\AA$^3$) while running 3.0$\times$ faster and reaching 1{,}920 atoms against 448 within the memory of one GPU.
DPA4-Plus, with 3.4$\times$ fewer parameters, is 11.5$\times$ faster and reaches 7{,}200 atoms at a lower energy MAE (9.8 against 10.4~meV/atom) and a 4.8\% higher force MAE.
The DPA3-L6 and DPA3-L16 baselines~\cite{zhang2025dpa3} are 3.1$\times$ and 3.6$\times$ slower than DPA4-Air and DPA4-Plus, respectively.

\subsection{Mechanism ablations}
\label{sec:ablation}

To establish that the accuracy--cost gains originate from the proposed mechanisms rather than from confounding hyperparameter choices, we ablate the architecture cumulatively.
Starting from the complete DPA4 interaction block, the components are removed one after another, with each removal retained in all subsequent configurations; every configuration is retrained under an identical protocol and evaluated on a subsample of the WBM test set~\cite{WBM}.
Figure~\ref{fig:ablation} reports the outcome, with the underlying numbers in Supplementary Table~\ref{tab:sm_cumulative_ablation}; the complete sweeps, protocols and configurations are given in Supplementary Section~\ref{sec:sm-ablation}.

Removing all four components raises the energy and force MAEs by 27.9\% and 13.8\%, respectively, relative to the complete architecture (31.007 to 39.671~meV/atom and 39.169 to 44.581~meV/\AA{}).
To make this cumulative decomposition additive, every percentage quoted for the ladder is an absolute increment divided by the corresponding complete-architecture value, rather than a ratio between adjacent rungs.
Envelope-gated attention (A3) is the largest single contributor: switching it off accounts for 16.6 percentage points of the energy increase and 6.6 points of the force increase.
Because its weights are computed from rotationally invariant scalar channels, this gain reflects adaptive neighbor weighting rather than any relaxation of SO(3) equivariance.
The low-rank edge--node SO(2)-equivariant degree coupling (A1) accounts for 4.2 energy and 5.3 force percentage points, and thus matters more for forces than for energies, whereas the WBN inside the equivariant FFN (A4) accounts for 5.3 and 2.5 points and shows the opposite balance.
The second WBN, acting on the post-aggregation message--node path, i.e.\ the term $\mathcal{M}_i$ of Eq.~\eqref{eq:dpa4-message-node} (MN in Fig.~\ref{fig:ablation}), contributes 1.9 energy percentage points; its force contribution is $-0.5$ points because removing it marginally lowers the force MAE, so this ladder does not resolve a force benefit (Fig.~\ref{fig:ablation}b).

The aggregate inference overhead is modest: removing all four components reduces the parameter count from 2.05M to 1.19M but shortens the inference time by only 12.4\%, from 20.2 to 17.7~s (Fig.~\ref{fig:ablation}c).
Within the same additive decomposition, attention contributes 16.6 energy and 6.6 force percentage points while accounting for 2.0 points of the complete-model inference time.

\begin{figure}[t]
    \centering
    \includegraphics[width=\linewidth]{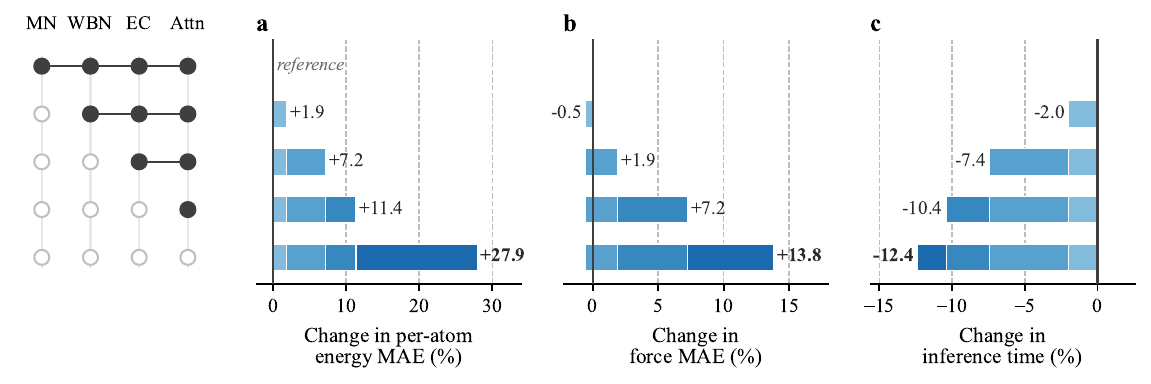}
    \caption{
        Cumulative ablation of the DPA4 architectural components.
        The matrix states which components each configuration retains: a filled marker denotes a component that is present and an open marker one that has been switched off, so that reading downwards accumulates the removals of the rows above.
        The row labels denote the four components: MN, the Wigner bilinear network (WBN) grid on the post-aggregation message--node path; WBN, the grid inside the equivariant feed-forward network (A4); EC, the low-rank edge--node SO(2)-equivariant degree coupling (A1); and Attn, the envelope-gated attention aggregation (A3).
        (a--c) Per-atom energy MAE, force MAE and inference time of each configuration relative to the complete architecture; every change is normalized by the corresponding complete-architecture value, and the segments therefore sum to the cumulative label at each rung.
        Inference time is the wall-clock time of a full test pass over the WBM subsample.
        Every configuration uses a single focus stream: the multi-focus design (A2) redistributes a fixed angular budget across parallel streams rather than adding a branch, and is therefore compared at matched SO(2) width in Supplementary Section~\ref{sec:sm-ablation}.
    }
    \label{fig:ablation}
\end{figure}

The four components of Fig.~\ref{fig:ablation} can each be disabled without changing the dimensions of the equivariant representation; this is what makes the cumulative ladder well defined.
The multi-focus design (A2) is different in kind: it separates expressivity from raw channel width by redistributing a fixed angular budget across parallel streams rather than by adding a branch, so its effect is defined only against a matched SO(2) width.
Every configuration in Fig.~\ref{fig:ablation} uses a single focus stream, so the contributions measured there are obtained without any assistance from A2, and the reallocation is assessed on its own.
At a fixed SO(2) dimension of 192, the 96-channel two-focus model matches the accuracy of the 192-channel single-focus model (energy and force MAEs of 26.994~meV/atom and 36.408~meV/\AA{} versus 27.286~meV/atom and 36.477~meV/\AA{}) while using 56\% fewer parameters, 23\% less training time and 34\% less inference time.
Widening the single stream beyond this point yields no further benefit: the 256-channel single-focus model carries about four times as many parameters as the 96-channel two-focus model, yet attains a 5.2\% higher energy MAE with a nearly identical force MAE (Supplementary Table~\ref{tab:dpa4_focus_ablation}).
Under a shared training recipe, a single wide equivariant stream is thus harder to optimize to comparable quality than several narrower focus streams.
Allocating the same angular budget to focus channels that compete over distinct edge-local motifs therefore provides a more parameter-efficient and more readily optimized route for scaling the SO(2) representation than widening a single stream.

The mechanism behind the A1 rung of Fig.~\ref{fig:ablation} is how the per-degree edge radial profiles modulate the node message in the local SO(2) frame.
The ablated variant applies the profiles degree-diagonally, the degree-$l$ profile scaling only the degree-$l$ channels, whereas the full form lets them parameterize a cross-degree kernel that mixes input and output degrees within each fixed $|m|$ stratum.
This kernel is a learnable function of the degree-indexed edge radial profiles, conserving $|m|$ as SO(2) equivariance requires (Section~\ref{sec:so2}), and is factorized across channels to keep the additional cost small.
A low rank suffices: in the dedicated rank sweep, even the rank-1 form recovers most of the gain over degree-diagonal scaling (Supplementary Table~\ref{tab:dpa4_radial_degree_ablation}), and the released variants use rank 1 or 2 (Supplementary Table~\ref{tab:sm_model_config}).

The quadrature rule on which the WBN (A4) is evaluated sets the numerical equivariance of the nonlinearity, and the comparison below isolates that choice from the design of the nonlinearity itself.
Table~\ref{tab:s2_full_equivariance} reports the equivariance error of the spherical-grid nonlinearity (Section~\ref{sec:s2-method}), applied to the complete coefficient layout with $l\le L$, under random SO(3) rotations.
Tensor-product latitude--longitude grids leave maximum fp64 residuals of $3.62\times10^{-7}$ -- $4.14\times10^{-6}$; for the higher angular orders, the fp64 residual is larger than the corresponding fp32 residual, showing that the error is set by the projection rule rather than by floating point round-off.
For the same maximum degree $L$, Lebedev quadrature reduces the fp64 residual to $2.31\times10^{-14}$ -- $7.99\times10^{-14}$ and decreases the grid size from 64--576 points to 26--170 points.
On the WBM ablation, this replacement slightly reduces the energy MAE while leaving the force MAE and wall-clock cost essentially unchanged in the FFN-only setting (Supplementary Table~\ref{tab:sm_quadrature_ablation}).
The main role of the Lebedev rule is therefore to remove a systematic numerical symmetry error from the nonlinear equivariant branch at lower quadrature size, independently of which nonlinearity is evaluated on it.

\begin{table}[t]
    \PaperTableStyle
    \begin{threeparttable}
        \caption{Equivariance of the spherical-grid nonlinearity under random SO(3) rotations.}
        \label{tab:s2_full_equivariance}
        \begin{tabular*}{\linewidth}{@{\extracolsep{\fill}}ccccccccc@{}}
            \toprule
            \textbf{$L$} &
            \multicolumn{4}{c}{\textbf{Product grid}} &
            \multicolumn{4}{c}{\textbf{Lebedev quadrature}} \\
            \cmidrule(lr){2-5}\cmidrule(lr){6-9}
            & \textbf{Rule} & \textbf{\#\,pts} & \textbf{fp64 error} & \textbf{fp32 error}
            & \textbf{$p$} & \textbf{\#\,pts} & \textbf{fp64 error} & \textbf{fp32 error} \\
            \midrule
            2 & $8\times8$   & 64  & $3.62\times10^{-7}$ & $4.77\times10^{-7}$ & 7  & 26  & $2.31\times10^{-14}$ & $2.38\times10^{-7}$ \\
            3 & $12\times12$ & 144 & $7.04\times10^{-7}$ & $6.86\times10^{-7}$ & 9  & 38  & $3.58\times10^{-14}$ & $3.58\times10^{-7}$ \\
            4 & $14\times14$ & 196 & $7.97\times10^{-7}$ & $1.55\times10^{-6}$ & 13 & 74  & $5.82\times10^{-14}$ & $6.56\times10^{-7}$ \\
            5 & $18\times18$ & 324 & $1.48\times10^{-6}$ & $1.49\times10^{-6}$ & 15 & 86  & $3.22\times10^{-14}$ & $6.56\times10^{-7}$ \\
            6 & $20\times20$ & 400 & $4.14\times10^{-6}$ & $2.27\times10^{-6}$ & 19 & 146 & $7.99\times10^{-14}$ & $8.35\times10^{-7}$ \\
            7 & $24\times24$ & 576 & $3.19\times10^{-6}$ & $2.03\times10^{-6}$ & 21 & 170 & $6.86\times10^{-14}$ & $8.79\times10^{-7}$ \\
            \bottomrule
        \end{tabular*}
        \begin{tablenotes}[flushleft]
            \PaperTableNotesStyle
            \item[] $L$ is the maximum spherical-harmonic degree. Product-grid rules are $(R_{\phi},R_{\theta})$ after the square-grid lift, and \#\,pts is the total number of grid points $R_{\phi}R_{\theta}$.
            \item[] For Lebedev quadrature, $p$ is the algebraic order of accuracy. fp64 and fp32 denote 64- and 32-bit floating-point arithmetic, respectively; errors are maximum absolute deviations between the two equivariance paths.
        \end{tablenotes}
    \end{threeparttable}
\end{table}

Beyond the mechanism ablations above, Supplementary Section~\ref{sec:sm-ablation} reports additional sweeps over compiled mixed-precision training, interaction depth, attention design variants, normalization placement and the learning-rate schedule.
The mechanism ablations confirm that the reported improvements arise from the targeted architectural designs A1--A4, while the additional sweeps establish stable design and training choices for the released DPA4 variants.


\section{Discussion}
\label{sec:discussion}

DPA4 shows that the accuracy--cost trade-off of equivariant interatomic potentials can be substantially improved by targeted architectural design.
Inside the EMFA SO(2) convolution, which avoids the full cost of SO(3) Clebsch--Gordan tensor products while remaining more expressive than prior SO(2) reductions, the edge--node degree coupling (A1), the multi-focus design (A2) and the envelope-gated attention (A3) each raise accuracy at low cost.
The Wigner bilinear network (A4) supplies the nonlinearity: it couples all angular degrees at once, in its full form approximates any continuous SO(3)-equivariant map, and, evaluated by exact Lebedev--uniform quadrature, is equivariant to machine precision.
Complementing the architecture, a compiler-friendly implementation makes the conservative energy-to-force path compatible with \texttt{torch.compile} and removes the systems overhead that usually limits expressive equivariant models.
The cumulative ablation confirms that every mechanism contributes (Section~\ref{sec:ablation}); together, these designs move DPA4 onto a new accuracy--cost frontier, in inference, training and parameter cost alike, across inorganic-crystal and organic-molecule benchmarks.

Across materials and molecular benchmarks, DPA4 reaches leading accuracy at a fraction of the cost of the strongest baselines.
On Matbench Discovery, DPA4-Pro leads every ranked metric of the compliant leaderboard, and all five variants evaluated on this benchmark lie on the Pareto frontier in CPS and saturated inference throughput among the fifteen compliant models tested (Fig.~\ref{fig:cps-throughput}).
They also define an accuracy--training-cost frontier that no baseline matches at less than four times its training budget (Table~\ref{tab:matbench}).
On OMat24, DPA4-Pro attains the lowest validation energy MAE of all models compared, with $6.1\times$ and $27.8\times$ fewer parameters than the two next most accurate models.
On the OMol25 composition-validation split, DPA4-Pro records lower total-energy and force MAEs than UMA-M-1.1, the strongest conservative baseline, using 25.2M versus 1.4B total parameters (50M active).
On SPICE-MACE-OFF, DPA4-Plus lowers the aggregate energy and force errors of the best published model by $39\%$ and $36\%$.
These results are obtained through the conservative energy-gradient path alone, without the denoising or direct-force pretraining stages on which the strongest baselines rely; the compiled implementation removes the double-backward overhead that motivated such two-stage protocols.
This makes compact, single-task energy-conservative potentials practical to train, ablate and redeploy in molecular-dynamics and structure-relaxation workflows.


DPA4 is established here as a backbone architecture in the single-task, per-dataset setting; its behavior under large-scale multitask pretraining on heterogeneous DFT data remains to be established.
An important next step is therefore to use DPA4 as a backbone for LAM pretraining and downstream adaptation while preserving its conservative energy-gradient training efficiency.
The relevant question is then not only whether a single pretrained model reaches a higher benchmark score, but whether accurate target-domain potentials can be generated, validated and refined repeatedly at low cost.
Such low-cost, repeated refinement is an ingredient for more automated, agentic potential development and, ultimately, for closing the loop between computation and experiment, so that simulation-driven model updates and laboratory feedback iteratively refine one another.


\section{Methods}
\label{sec:methods}

\subsection{Datasets}

For the Matbench Discovery benchmark, we train DPA4 on MPtrj, the Materials Project trajectory dataset introduced with CHGNet~\cite{deng2023chgnet,jain2013materials}.
MPtrj contains relaxation and static calculations for inorganic crystals across 89 elements, with energies, forces, stresses and magnetic moments computed at the GGA or GGA+$U$ level.
Generalization is evaluated on the Matbench Discovery benchmark, in which MPtrj-trained models relax the WBM candidate structures and are scored by their formation-energy predictions and derived convex-hull distances~\cite{riebesell2025matbench,WBM}.
The leaderboard additionally reports $\kappa$SRME, the symmetric relative mean error of the phonon thermal conductivity, which is sensitive to the smoothness and conservativeness of the learned potential~\cite{pota2024thermal}.
The architectural ablations instead train on MPtrj and evaluate on a fixed subsample of the WBM test set, a protocol shared by all ablations so that differences in accuracy and throughput reflect only the controlled model component.

OMat24~\cite{barroso2024open} contains 118 million structures labeled with energy, forces and cell stress.
Input structures were obtained by randomly sampling relaxed structures from the Alexandria PBE bulk dataset~\cite{schmidt2023alexandria} and driving them out of equilibrium by Boltzmann sampling of rattled structures, short ab initio molecular-dynamics trajectories and the relaxation of rattled structures.
Labels were computed with VASP at the PBE level, with Hubbard $U$ corrections for oxides and fluorides containing Co, Cr, Fe, Mn, Mo, Ni, V or W, following Materials Project defaults.
We train on the 100.8-million-structure training split and report energy, force and stress MAEs on the 5.3-million-structure validation split.

OMol25~\cite{levine2025open}, in its OMol-0 release, contains about 100 million DFT single-point calculations at the $\omega$B97M-V/def2-TZVPD level computed with ORCA.
It combines four domains --- biomolecules, metal complexes, electrolytes and recomputed community datasets --- and covers 83 elements, systems of up to 350 atoms, total charges from $-10$ to $+10$ and spin multiplicities from 1 to 11, so that charge and spin enter the model as explicit inputs.
Splits are constructed by molecular composition, so that the validation and test compositions are held out from training.
We train on its 101.7-million-structure training split and report energy and force MAEs on the out-of-distribution composition validation split.

SPICE-MACE-OFF is the organic-molecule dataset introduced for MACE-OFF~\cite{kovacs2023mace}.
It spans PubChem molecules, DES370K monomers and dimers~\cite{donchev2021des370k}, dipeptides, solvated amino acids, water clusters and larger QMugs-derived molecules~\cite{eastman2023spice,isert2022qmugs}.
Reference energies and forces are evaluated at the $\omega$B97M-D3(BJ)/def2-TZVPPD level with PSI4~\cite{najibi2018nonlocal,grimme2010consistent,grimme2011effect,weigend2005balanced,rappoport2010property,smith2020psi4}.
We use the same train/validation/test split as MACE-OFF~\cite{kovacs2023mace} and DPA3~\cite{zhang2025dpa3}, and report per-subset energy and force MAEs.

\subsection{DPA4 model architecture}

Consider a system of $N$ atoms with atomic numbers $Z=\{Z_i\}_{i=1}^{N}$ and positions $R=\{\mathbf{R}_i\}_{i=1}^{N}$.
DPA4 decomposes the potential energy
as the sum of a learned equivariant message-passing branch and an analytical
Ziegler--Biersack--Littmark (ZBL) short-range branch,
\begin{equation}
    E(Z,R)=E_{\Theta}^{\mathrm{NN}}(Z,R) + E^{\mathrm{ZBL}}(Z,R), \label{eq:dpa4-total-energy} \end{equation} where $\Theta$ collects all learnable parameters of the learned branch.
Forces and virials are obtained by differentiating the scalar energy in Eq.~\eqref{eq:dpa4-total-energy} with respect to atomic positions and the cell, so the learned and analytical branches together define a single conservative potential.

In the learned branch $E_{\Theta}^{\mathrm{NN}}$, the atomic species $Z$ and positions $R$ are encoded into per-atom irreducible representations in the real space $V_{\leq L}\otimes\mathbb{R}^{C}$ by a Geometry-Informed Embedding (GIE) stage (Sec.~\ref{sec:gie}), updated through $N_{\mathrm{layer}}$ stacked equivariant interaction blocks, and read out as a scalar energy by an atomic energy head (Sec.~\ref{sec:ener-head}).
Here $V_{\leq L}=\bigoplus_{l=0}^{L}V_l$ with $V_l$ the $(2l+1)$-dimensional irreducible representation of SO(3); the $l=0$ subspace is invariant and is used for scalar readout, while $l>0$ subspaces transform equivariantly and carry angular information.
Each interaction block is a residual stack of an EMFA SO(2) convolution (Sec.~\ref{sec:so2}) and an equivariant feed-forward network (FFN, Sec.~\ref{sec:s2-method}) interleaved with equivariant RMS norms.
The architectural designs A1--A4 of the Introduction operate inside the block: the low-rank edge--node SO(2)-equivariant degree coupling (A1), the multi-focus design with cross-focus competition (A2), and the envelope-gated attention (A3) act inside the EMFA SO(2) convolution, while the FFN is instantiated as a Wigner bilinear network (A4).
Native ZBL zone bridging couples the learned and analytical branches in the close-contact region (Sec.~\ref{sec:zbl_method}).

\subsubsection{Geometric inputs to the learned branch}
\label{sec:geom-inputs}

Throughout Sections~\ref{sec:geom-inputs}--\ref{sec:s2-method}, every geometric quantity entering the learned branch $E_{\Theta}^{\mathrm{NN}}$, namely radial basis functions, cutoff envelopes and the local-environment descriptor $\mathcal{D}_i$, is evaluated on a \emph{clamped distance} $\widetilde{r}(r_{ij})$ rather than the raw distance $r_{ij}=\|\mathbf{R}_j-\mathbf{R}_i\|$, so $r_{ij}$ should be read as $\widetilde{r}(r_{ij})$ in every downstream equation.
Quantities fixed by the bond direction alone, such as $Y_l^m(\widehat{\mathbf{r}}_{ij})$ and the edge-local frame, are unaffected, since the clamp acts on the radial magnitude only.
Edge messages sourced from atom $j$ are additionally weighted by a smooth \emph{source-freeze gate} $\eta_j\in[0,1]$ that vanishes whenever $j$ has a neighbor inside the inner zone.
Both $\widetilde{r}$ and $\eta_j$ are the technical core of \emph{native ZBL zone bridging} (Sec.~\ref{sec:zbl_method}), and are defined next.

\paragraph{Bridging window and septic Hermite polynomials.}
The \emph{bridging window} is an interval $[r_{\mathrm{in}},r_{\mathrm{out}}]$ with $0<r_{\mathrm{in}}<r_{\mathrm{out}}\le r_{\mathrm{c}}$ inside the neighbor cutoff $r_{\mathrm{c}}$.
Two septic Hermite polynomials
$h_{\mathrm{c}},h_{\mathrm{w}}:[0,1]\to[0,1]$ control the clamp and the gate,
respectively,
\begin{equation}
    h_{\mathrm{c}}(t)\coloneqq 20t^4-45t^5+36t^6-10t^7,
    \qquad
    h_{\mathrm{w}}(t)\coloneqq 35t^4-84t^5+70t^6-20t^7,
    \label{eq:dpa4-hermite}
\end{equation}
each being the unique interpolant fixed by eight boundary conditions.
For $h_{\mathrm{c}}$, it is required that $h_{\mathrm{c}}(0)=0$, $h_{\mathrm{c}}(1)=1$,
$h_{\mathrm{c}}'(0)=0$, $h_{\mathrm{c}}'(1)=1$,
$h_{\mathrm{c}}''(0)=h_{\mathrm{c}}''(1)=h_{\mathrm{c}}'''(0)=h_{\mathrm{c}}'''(1)=0$, 
while for
$h_{\mathrm{w}}$
it is required that $h_{\mathrm{w}}(0)=0$, $h_{\mathrm{w}}(1)=1$,
$h_{\mathrm{w}}^{(k)}(0)=h_{\mathrm{w}}^{(k)}(1)=0$ for $k=1,2,3$.

\paragraph{Clamped distance map and bridging amplitude.}
With $t(r)\coloneqq(r-r_{\mathrm{in}})/(r_{\mathrm{out}}-r_{\mathrm{in}})$,
the clamped distance map and the bridging amplitude are
\begin{equation}
    \widetilde{r}(r) =
    \begin{cases}
        r_{\mathrm{in}}, & r\le r_{\mathrm{in}},               \\
        r_{\mathrm{in}}+(r_{\mathrm{out}}-r_{\mathrm{in}})\,h_{\mathrm{c}}\!\bigl(t(r)\bigr),
                         & r_{\mathrm{in}}<r<r_{\mathrm{out}}, \\
        r,               & r\ge r_{\mathrm{out}},
    \end{cases}
    \qquad
    w(r) =
    \begin{cases}
        0,                                & r\le r_{\mathrm{in}},               \\
        h_{\mathrm{w}}\!\bigl(t(r)\bigr), & r_{\mathrm{in}}<r<r_{\mathrm{out}}, \\
        1,                                & r\ge r_{\mathrm{out}}.
    \end{cases}
    \label{eq:dpa4-clamp}
\end{equation}
Both $\widetilde{r}$ and $w$ are $C^3$ on $\mathbb{R}_{>0}$ by the Hermite boundary conditions.
The associated \emph{clamped displacement} is
\begin{equation}
    \widetilde{\mathbf{r}}_{ij} = \widetilde{r}(r_{ij})\,\widehat{\mathbf{r}}_{ij},
    \qquad \|\widetilde{\mathbf{r}}_{ij}\|=\widetilde{r}(r_{ij}),
    \label{eq:dpa4-rvec-tilde}
\end{equation}
which carries the clamped magnitude along the exact bond direction.

\paragraph{Source-freeze gate.}
For each atom $j$, let $\mathcal{N}(j)\coloneqq \{i:r_{ji}<r_{\mathrm{c}},\,i\neq j\}$ denote its set of neighbors within the cutoff.
The source-freeze gate is
\begin{equation}
    \eta_j \;\coloneqq\; \prod_{i\in\mathcal{N}(j)} w(r_{ji})
    \;\in\;[0,1].
    \label{eq:dpa4-source-freeze}
\end{equation}
As a finite product of $C^3$ functions, $\eta_j$ is $C^3$ in the atomic positions wherever no two atoms coincide.
It vanishes whenever some neighbor of $j$ lies in the inner zone, $r_{ji}\le r_{\mathrm{in}}$.

\paragraph{Behavior inside and outside the bridging window.}
Two properties of Eqs.~\eqref{eq:dpa4-clamp}--\eqref{eq:dpa4-source-freeze} are used repeatedly in the rest of Section~\ref{sec:methods}.
First, on the inner zone $r_{ij}\le r_{\mathrm{in}}$ every geometric quantity of the learned branch is constant in $r_{ij}$ (since $\widetilde r$ is constant there) and every message sourced from an atom $j$ whose own neighborhood penetrates the inner zone is silenced (since $\eta_j=0$), so the learned branch contributes zero gradient there and the short-range repulsion is handled exclusively by $E^{\mathrm{ZBL}}$.
Second, once $r_{ij}\ge r_{\mathrm{out}}$ and every neighbor of $j$ is at least $r_{\mathrm{out}}$ away, the clamp is the identity and the gate is one, so the learned branch sees the true geometry.

\subsubsection{Geometry-Informed Embedding (GIE)}
\label{sec:gie}

The initial feature $\mathbf{h}_i^{(0)}\in V_{\leq L}\otimes\mathbb{R}^{C}$ depends on both \emph{chemistry} (atomic species $Z_i$ and the species of neighbors) and \emph{geometry} (the relative positions $\{\mathbf{r}_{ij}\}$ inside the cutoff).
The two enter the slices differently: the $l=0$ slice starts from a chemistry-only embedding that geometry then modulates, whereas the $l\geq 1$ slices are constructed from the neighbor directions.

\paragraph{$l=0$ slice.}
The chemistry-only baseline is a learnable element embedding $\mathbf{T}\in\mathbb{R}^{n_t\times C}$, giving $\mathbf{h}_{i,l=0,c}^{(0)}=\mathbf{T}_{Z_i,c}$.
To inject local geometry into the scalar slice without breaking SO(3) invariance, DPA4 builds a compact rotation-invariant local-environment descriptor $\mathcal{D}_i$ in the spirit of the smooth-edition Deep Potential descriptor~\cite{zhang2018end} as follows.
The per-edge four-vector is
\begin{equation}
    \mathbf{u}_{ij,0} = \frac{s_5(r_{ij})}{r_{ij}},
    \qquad
    \mathbf{u}_{ij,k} = \mathbf{u}_{ij,0}\,\widehat{\mathbf{r}}_{ij,k},\quad k=1,2,3,
    \label{eq:dpa4-env-vector}
\end{equation}
whose first component is a smooth invariant and whose last three components
together transform as an SO(3) vector. Here $s_5$ is one of a family of $C^3$
cutoff envelopes $s_p$ that decay from $s_p(0)=1$ and vanish together with
their first three derivatives at $r_{\mathrm{c}}$; their closed forms are given
in Supplementary Section~\ref{sec:sm-radial}.
DPA4 uses $s_5$ for edge weighting and in the degree normalization, and $s_7$ in the radial basis defined below.
With a separate
radial-species map
$\mathbf{g}:\mathbb{R}_{\ge 0}\times\{1,\dots,n_t\}^2\to\mathbb{R}^{C_{\mathrm{env}}}$,
the per-atom matrix is
\begin{equation}
    A_i = n_i\sum_{j:r_{ij}<r_{\mathrm{c}}} \eta_j\,
    \mathbf{u}_{ij}\otimes \mathbf{g}(r_{ij};Z_i,Z_j)
    \in\mathbb{R}^{4\times C_{\mathrm{env}}},
    \label{eq:dpa4-env-matrix}
\end{equation}
where $n_i=(d_i+\varepsilon)^{-1/2}$ with
$d_i=\sum_{j:r_{ij}<r_{\mathrm{c}}} s_5(r_{ij})^2$ is a smooth degree
normalization (squaring $s_5$ makes $d_i$ inherit $C^6$ regularity; the
regularizer $\varepsilon>0$ prevents singularity for isolated atoms), and
$\eta_j$ is the source-freeze gate defined in
Eq.~\eqref{eq:dpa4-source-freeze}. Contracting the first axis of $A_i$ yields
a rotation-invariant Gram-style descriptor,
\begin{equation}
    \mathcal{D}_i = A_i^{\top}
    A_i^{(:,1{:}K_{\mathrm{env}})} \in\mathbb{R}^{C_{\mathrm{env}}\times K_{\mathrm{env}}}, \label{eq:dpa4-env-descriptor} \end{equation} where the truncation to $K_{\mathrm{env}}$ columns controls cost.
Invariance of $\mathcal{D}_i$ follows because contracting the first axis of $A_i$ pairs the scalar component of $\mathbf{u}_{ij}$ with itself and the three spatial components with one another, and both pairings are rotation invariant.
The descriptor then conditions the scalar features through a
Feature-wise Linear Modulation (FiLM) step~\cite{perez2018film},
\begin{equation}
    \mathbf{h}_{i,l=0}^{(0)}
    \leftarrow
    \boldsymbol{\gamma}_i \odot \mathbf{h}_{i,l=0}^{(0)} + \boldsymbol{\beta}_i,
    \label{eq:dpa4-film}
\end{equation}
with per-channel scale and shift
\begin{equation}
    \boldsymbol{\gamma}_i = \mathbf{1} + e^{\lambda_\alpha}\tanh\!\bigl(N_0(W_\alpha\,\mathrm{vec}\,\mathcal{D}_i)\bigr),
    \qquad
    \boldsymbol{\beta}_i = e^{\lambda_\beta}\tanh\!\bigl(N_0(W_\beta\,\mathrm{vec}\,\mathcal{D}_i)\bigr),
    \label{eq:dpa4-film-gb}
\end{equation}
where $W_\alpha,W_\beta\in\mathbb{R}^{C\times C_{\mathrm{env}}
        K_{\mathrm{env}}}$ are learnable projections and $N_0$ is a scalar RMS normalizer.
The bounded $\tanh$ outputs are scaled by learnable log-strengths $\lambda_\alpha,\lambda_\beta\in\mathbb{R}$, initialized at $\log(0.01)$, so the modulation begins close to the identity ($\boldsymbol{\gamma}_i\approx\mathbf{1}$, $\boldsymbol{\beta}_i\approx\mathbf{0}$) and the species embedding dominates at the start of training.
Because $\mathcal{D}_i$ is SO(3)-invariant and FiLM acts diagonally in the channel index, the modulation preserves the invariance of the $l=0$ slice.

\paragraph{$l\geq 1$ slices.}
For higher degrees no chemistry-only baseline exists: equivariant features must carry directional information from the start.
DPA4 obtains them by projecting
each neighbor direction onto real spherical harmonics and weighting the
projection by a radial-species profile,
\begin{equation}
    \bigl[\mathbf{h}_{i}^{(0)}\bigr]^{(l)}_{m,c}
    =
    n_i\sum_{j:r_{ij}<r_{\mathrm{c}}}
    \eta_j\,Y_l^m(\widehat{\mathbf{r}}_{ij})\,
    \rho_{ij,l,c},\qquad l\geq 1.
    \label{eq:dpa4-gie}
\end{equation}
Here $n_i$ and $\eta_j$ are the same smooth degree normalization and source-freeze gate as in Eq.~\eqref{eq:dpa4-env-matrix}.
The radial-species profile $\rho_{ij,l,c}$
mixes the per-pair chemistry into a smooth function of distance,
\begin{equation}
    \rho_{ij,l,c}
    =
    \bigl[\mathrm{MLP}^{\mathrm{rad}}\!\bigl(\boldsymbol{\phi}(r_{ij})\bigr)\bigr]_{l,c}
    + \bigl[\mathbf{T}_{\mathrm{edge}}(Z_i,Z_j)\bigr]_{c},
    \label{eq:dpa4-rho}
\end{equation}
where $\mathrm{MLP}^{\mathrm{rad}}:\mathbb{R}^{n_r}\to\mathbb{R}^{(L+1)\times C}$ is
bias-free with SiLU activations, $\mathbf{T}_{\mathrm{edge}}(Z_i,Z_j)\in\mathbb{R}^{C}$ is a
per-pair species embedding broadcast across the $(L+1)$ degree slots, and
$\boldsymbol{\phi}(r_{ij})=(\phi_1(r_{ij}),\dots,\phi_{n_r}(r_{ij}))
    \in\mathbb{R}^{n_r}$ is the Bessel radial basis
$\phi_n(r)=\sin(\omega_n r)\,s_7(r)/r$, $n=1,\dots,n_r$, with learnable
frequencies initialized at $\omega_n=n\pi/r_{\mathrm{c}}$ and $s_7$ the cutoff
envelope of Supplementary Section~\ref{sec:sm-radial}.
Equation~\eqref{eq:dpa4-gie} is SO(3)-equivariant because every prefactor is rotation-invariant and the entire angular content is carried by the degree-$l$ spherical harmonic $Y_l^m(\widehat{\mathbf{r}}_{ij})$, which transforms as $D^l(R)$ under a rotation $R$.

\subsubsection{EMFA SO(2) convolution}
\label{sec:so2}

Each interaction block applies the EMFA (Edge--node degree coupling, Multi-Focus, Attention) SO(2) convolution $\mathcal{C}_{\Theta}$, an SO(3)-equivariant convolution that takes the per-atom node features $\mathbf{h}_j\in V_{\le L}\otimes\mathbb{R}^{C}$, $j=1,\dots,N$, and returns a per-atom update $(\mathcal{C}_{\Theta}\mathbf{h})_i\in V_{\le L}\otimes\mathbb{R}^{C}$ obtained by aggregating information from the neighbors of atom $i$.
The operator is built in six stages: (i) transport each source node feature into a per-edge SO(2) gauge aligning the bond direction with a fixed reference axis; (ii) construct the in-frame message through a \emph{low-rank edge--node SO(2)-equivariant degree coupling} (A1); (iii) introduce message nonlinearity through a \emph{multi-focus design} (A2); (iv) lift the in-frame message back to the global frame; (v) aggregate neighbor messages with envelope-gated attention (A3) modulated by a destination-side output gate; and (vi) project the result back to representation width through a channel post-mixer.
The closed form of $\mathcal{C}_{\Theta}$ collecting all six stages is given as Eq.~\eqref{eq:dpa4-conv} at the end of the subsection.

\paragraph{Edge-local frame.}
Full SO(3)-equivariant tensor products require Clebsch--Gordan expansions whose cost grows steeply with angular order~\cite{batatia2022mace,liao2023equiformerv2}.
DPA4 instead reduces SO(3)
convolutions to SO(2) by rotating each directed edge $(i,j)$ into a gauge that
aligns its bond direction with the reference axis,
$R_{ij}\widehat{\mathbf{r}}_{ij}=(0,0,1)^{\top}$~\cite{passaro2023reducing}.
In this frame the residual symmetry is the abelian group SO(2), so angular orders $m$ decouple into independent strata.
DPA4 retains coefficients with $|m|\leq M\leq L$ inside the convolution and applies a degree-dependent lift factor $\Upsilon_M$ (Eq.~\eqref{eq:sm-rescale}) after rotating back to compensate the norm loss from truncation.

\paragraph{Per-edge equivariant message.}
For each directed edge $(i,j)$, the source node feature is pre-mixed to hidden
width and transported into the edge-local frame,
\begin{equation}
    \mathbf{x}_{ij}=P_M\,D(R_{ij})\,\mathbf{h}_j',
    \qquad
    \mathbf{h}_j'\;\coloneqq\;L^{\mathrm{pre}}_{\mathrm{deg}}\mathbf{h}_j
    \in V_{\le L}\otimes\mathbb{R}^{H},
    \label{eq:dpa4-local-projection}
\end{equation}
with $\mathbf{x}_{ij}\in\mathbb{R}^{D_M\times H}$, where
$L^{\mathrm{pre}}_{\mathrm{deg}}$ is a degree-wise channel pre-mixer
(Eq.~\eqref{eq:sm-degree-linear}) that lifts the representation width from
$C$ to a hidden width $H=F\,C_f$ (with focus count $F$ and per-focus width
$C_f$), $D(R_{ij})$ is block-diagonal in $l$ (Eq.~\eqref{eq:sm-direct-sum}), and
$P_M$ selects the retained $m$-strata to dimension $D_M$
(Eq.~\eqref{eq:sm-trunc-index}).
The radial-species feature $\rho_{ij}$ of Eq.~\eqref{eq:dpa4-rho} depends only on the pair distance and the species pair, so it is computed once and reused by every interaction block.
Each block lifts it to hidden width by its own degree-wise channel map applied independently at each degree $l$ (Eq.~\eqref{eq:sm-degree-linear}),
\begin{equation}
    \widetilde{\rho}_{ij} = L^{\mathrm{rad}}_{\mathrm{lift}}\,\rho_{ij}
    \in\mathbb{R}^{(L+1)\times H},
    \qquad
    L^{\mathrm{rad}}_{\mathrm{lift}}:\mathbb{R}^{C}\to\mathbb{R}^{H},
    \label{eq:dpa4-rho-lift}
\end{equation}
giving the edge-side operand at the same width as $\mathbf{x}_{ij}$.

\paragraph{Low-rank edge--node SO(2)-equivariant degree coupling.}
This stage realizes the architectural design A1.
In the local frame, the edge angular feature $Y_l(\widehat{\mathbf{r}}_{ij})$ collapses to its $m=0$ component for every degree $l$, so $\widetilde{\rho}_{ij}$ of Eq.~\eqref{eq:dpa4-rho-lift} serves as the radial-modulated $m=0$ slice of the per-degree edge SO(2) irreps.
The edge--node degree coupling lets the edge-side irreps parameterize a learnable linear map acting on the node-side irreps $\mathbf{x}_{ij}$ that, at each fixed $|m|$-stratum, mixes the different angular degrees $l$ without coupling different $|m|$,
\begin{equation}
    \mathbf{x}_{ij,l,m,c}\;\leftarrow\;
    \sum_{l'\geq |m|}
    \mathcal{K}_{l,l',|m|,c}(\widetilde{\rho}_{ij})\,\mathbf{x}_{ij,l',m,c},
    \label{eq:dpa4-radial-coupling}
\end{equation}
where each kernel entry $\mathcal{K}_{l,l',|m|,c}(\widetilde{\rho}_{ij})\in \mathbb{R}$ is a learnable linear functional of $\widetilde{\rho}_{ij}$.
Because $\mathcal{K}$ depends only on rotation-invariant radial-species information and never mixes different $|m|$, the $(-m,+m)$ pair continues to transform as a single two-dimensional real SO(2) representation.
Under the edge-local reduction, the Clebsch--Gordan tensor product of a standard SO(3)-equivariant convolution is exactly an SO(2) map that mixes angular degrees within each fixed $|m|$ stratum (Proposition~3.1 of ref.~\cite{passaro2023reducing}), and Eq.~\eqref{eq:dpa4-radial-coupling} is a learnable map of that form.
Kernels constrained to be diagonal in $l$, as when the degree-$l$ edge feature multiplies only the degree-$l$ node feature, form a strict subspace of those maps: the cross-degree entries $l\neq l'$ are therefore not an enhancement of the reduction but what makes it faithful to the product it replaces.

The per-channel kernel of Eq.~\eqref{eq:dpa4-radial-coupling} improves accuracy over the degree-diagonal baseline, but carries $H$ independent copies of the degree-pair coefficients, raising the training and inference costs by factors of $3.36$ and $2.22$, respectively (Supplementary Table~\ref{tab:dpa4_radial_degree_ablation}).
$\mathcal{K}$ is therefore parameterized by a low-rank factorization across the channel index $c$,
\begin{equation}
    \mathcal{K}_{l,l',|m|,c}(\widetilde{\rho}_{ij})
    =
    \sum_{r=1}^{R}
    \bigl[K_{l,l',|m|}(\widetilde{\rho}_{ij})\bigr]_{r}\,B_{r,c},
    \label{eq:dpa4-K-lowrank}
\end{equation}
with rank $R\leq H$, where $K_{l,l',|m|}:\mathbb{R}^{(L+1)\times H}\to\mathbb{R}^{R}$ is a learnable linear map and $B\in\mathbb{R}^{R\times H}$ is a learnable channel basis.
The diagonal special case $\mathcal{K}_{l,l',|m|,c}(\widetilde{\rho}_{ij})=\widetilde{\rho}_{ij,l,c}\, \delta_{l,l'}$ reduces to per-degree scalar radial modulation.
The rank sweep favors small ranks on both counts: the rank-1 to rank-4 kernels attain lower energy and force MAEs than the dense one at a fraction of its cost (Supplementary Table~\ref{tab:dpa4_radial_degree_ablation}).

\paragraph{Multi-focus design for message nonlinearity.}
The in-frame edge feature is processed by two distinct nonlinear mechanisms that together realize the architectural design A2.
First, the hidden width
factorizes as $\mathbb{R}^{H}=\mathbb{R}^{F}\otimes \mathbb{R}^{C_f}$, so
$\mathbf{x}_{ij}\in\mathbb{R}^{D_M\times F\times C_f}$ splits into $F$
\emph{focus streams}, on each of which a multi-layer SO(2) stack acts in
parallel as a composition of $S$ residual layers,
\begin{equation}
    \mathbf{x}_{ij}\;\leftarrow\;
    \mathbf{x}_{ij}+\Lambda_s\odot
    \Gamma_s\!\bigl(L^{\mathrm{SO2}}_s\,N_s(\mathbf{x}_{ij})\bigr),
    \qquad s=1,\dots,S,
    \label{eq:dpa4-so2-layer}
\end{equation}
where $N_s$ is an equivariant RMS norm
(Eq.~\eqref{eq:sm-rmsnorm}), $L^{\mathrm{SO2}}_s$ is an
edge-independent SO(2)-equivariant linear map
(Eqs.~\eqref{eq:sm-so2-m0}, \eqref{eq:sm-so2-mn}; unrestricted on $m=0$,
the real form of complex multiplication on each $|m|>0$ subspace) supplying
cross-$l$ mixing at fixed $|m|$ that complements the edge-dependent
cross-$l$ mixing realized by $\mathcal{K}$ in
Eq.~\eqref{eq:dpa4-radial-coupling}, $\Gamma_s$ is a gated activation acting on
the scalar slice (Eq.~\eqref{eq:sm-gated}), and
$\Lambda_s\in\mathbb{R}^{F\times C_f}$ is a learnable
per-(focus,\,channel) residual scale initialized at $10^{-3}$.
The gated activations $\{\Gamma_s\}$ provide the \emph{first nonlinearity} in the message-construction pipeline.
The learnable weights of the stack are $L^{\mathrm{SO2}}_s$, $N_s$, $\Gamma_s$ and $\Lambda_s$ for $s=1,\dots,S$.

The \emph{second nonlinearity} is a cross-focus competition that depends on the edge's own SO(2)-invariant $l=0$ content through a softmax, turning the multi-focus split into a learnable nonlinear gating mechanism on top of the per-focus SO(2) stack.
The competition is computed from $\mathbf{s}_{ij}\in\mathbb{R}^{F\times C_f}$, the $(l,m)=(0,0)$ component of $\mathbf{x}_{ij}$ at the entry of the stack (Eq.~\eqref{eq:dpa4-so2-layer}), normalized by a focus-wise scalar RMS norm $N_0:\mathbb{R}^{C_f}\to\mathbb{R}^{C_f}$ applied independently to each of the $F$ stream rows.
The streams are reweighted by a per-focus competition weight,
\begin{equation}
    \begin{aligned}
        \mathbf{x}_{ij} & \;\leftarrow\; \alpha_{ij}\odot\mathbf{x}_{ij},                                     \\[2pt]
        \alpha_{ij,f}   & = (1-\epsilon)\,\frac{\exp\!\bigl(\tau^{-1}\sum_{c}
        W^{\mathrm{cf}}_{c,f}\, N_0(\mathbf{s}_{ij})_{f,c}\bigr)} {\sum_{f'}\exp\!
            \bigl(\tau^{-1}\sum_{c}
        W^{\mathrm{cf}}_{c,f'}\, N_0(\mathbf{s}_{ij})_{f',c}\bigr)} +\frac{\epsilon}{F},
    \end{aligned}
    \label{eq:dpa4-focus-compete}
\end{equation}
where $W^{\mathrm{cf}}\in\mathbb{R}^{C_f\times F}$ is a learnable channel-to-focus scoring matrix, $\tau>0$ is a softmax temperature that sharpens ($\tau\to 0$) or flattens ($\tau\to\infty$) the competition between streams, $\epsilon\in[0,1)$ is a label-smoothing strength that mixes the softmax with the uniform distribution $1/F$ over focuses to prevent any single stream from being driven to zero focus weight, and $\alpha_{ij}\in\mathbb{R}^{F}$ is broadcast across the $(l,m)$ and $c$ axes.
Equivariance is preserved because $\alpha_{ij,f}$ is constructed from an SO(2)-invariant $l=0$ slice.

\paragraph{Lift back to the global frame.}
The equivariant edge message in the global frame is recovered by inverting the
local gauge,
\begin{equation}
    \mathbf{m}_{ij}
    =
    \Upsilon_M\,D(R_{ij})^{\top}
    P_M^{\top}\,\mathbf{x}_{ij} \in V_{\le L}\otimes\mathbb{R}^{H}, \label{eq:dpa4-edge-message} \end{equation} with $D(R_{ij})^{\top}$ the inverse gauge rotation, $P_M^{\top}$ the re-embedding of the truncated $m$-strata back into the full $(L+1)^2$ layout, and $\Upsilon_M$ the truncation rescale of Eq.~\eqref{eq:sm-rescale}.
At this point $\mathbf{m}_{ij}$ still carries the hidden width $H$; channel post-mixing back to representation width is deferred until after neighbor aggregation.

\paragraph{Envelope-gated attention.}
\label{sec:attn}
The aggregation of $\mathbf{m}_{ij}$ over neighbors uses an envelope-gated attention weight $w_{ij}^{(f,a)}$: a destination-wise normalized softmax over neighbors whose logits are scalar functions of the invariant $l=0$ destination and source features plus a radial bias.
Each focus $f\in\{1,\dots,F\}$ is split into $H_a$ heads of width $d_a=C_f/H_a$, indexed by $a\in\{1,\dots,H_a\}$.
The pre-mixed hidden-width node feature $\mathbf{h}_n'$ of Eq.~\eqref{eq:dpa4-local-projection} has an $l=0$ slice $\mathbf{h}_n'|_{l=0}\in\mathbb{R}^{H}$, which the factorization $H=F\cdot H_a\cdot d_a$ reshapes into $\mathbb{R}^{F\times H_a\times d_a}$.
A focus-wise scalar RMS norm $N_0$ applied in this layout gives
$N_0\bigl(\mathbf{h}_n'|_{l=0}\bigr)\in\mathbb{R}^{F\times H_a\times d_a}$,
which is read either with the head axis resolved or, where the head axis is not
needed, in the focus layout $\mathbb{R}^{F\times C_f}$ with $C_f=H_a d_a$.
From it the per-node, per-(focus, head) queries and keys are formed,
\begin{equation}
    \mathbf{q}_i^{(f,a)}
    = Q^{(f)}\,N_0\bigl(\mathbf{h}_i'|_{l=0}\bigr)_{f,a,:},\qquad
    \mathbf{k}_j^{(f,a)}
    = K^{(f)}\,N_0\bigl(\mathbf{h}_j'|_{l=0}\bigr)_{f,a,:}
    \in\mathbb{R}^{d_a},
    \label{eq:dpa4-attn-qk}
\end{equation}
with learnable per-focus query/key matrices $Q^{(f)},K^{(f)}\in
    \mathbb{R}^{d_a\times d_a}$. The attention logit for edge $(i,j)$ at focus
$f$, head $a$ combines a scaled dot product with a radial bias linear in the
$l=0$ lifted radial-species feature $\widetilde{\rho}_{ij,0}$ of Eq.~\eqref{eq:dpa4-rho-lift},
\begin{equation}
    \ell_{ij}^{(f,a)}
    =
    \frac{\langle\mathbf{q}_i^{(f,a)},\mathbf{k}_j^{(f,a)}\rangle}{\sqrt{d_a}}
    + \sum_{c=1}^{C_f}
    W^{\mathrm{rb}}_{c,f,a}\,\widetilde{\rho}_{ij,0,f,c}, \label{eq:dpa4-attn-logit} \end{equation} where $W^{\mathrm{rb}}\in\mathbb{R}^{C_f\times F\times H_a}$ is a learnable radial-bias tensor and the channel axis of $\widetilde{\rho}_{ij,0}$ is split as $H=F\cdot C_f$, so that each focus reads its own $C_f$ channels.
The attention weight is
\begin{equation}
    w_{ij}^{(f,a)}
    =
    \frac{s_5(r_{ij})^2\,\eta_j\,\exp\!\bigl(\ell_{ij}^{(f,a)}\bigr)}
    {\operatorname{softplus}(\zeta_{f,a})
        +\sum_{k:r_{ik}<r_{\mathrm{c}}}
        s_5(r_{ik})^2\,\eta_k\,\exp\!\bigl(\ell_{ik}^{(f,a)}\bigr)},
    \label{eq:dpa4-attention}
\end{equation}
with the $C^3$ envelope $s_5$ of Supplementary Section~\ref{sec:sm-radial}, the
source-freeze gate $\eta_j$ of Eq.~\eqref{eq:dpa4-source-freeze}, and a
learnable null-logit $\zeta_{f,a}\in\mathbb{R}$.
Two mechanisms make $w_{ij}^{(f,a)}$ smooth at the cutoff.
The numerator factor $s_5(r_{ij})^2$ drives the weight $C^3$-smoothly to zero as $r_{ij}\to r_{\mathrm{c}}$, and the $\operatorname{softplus}(\zeta_{f,a})$ term keeps the denominator strictly positive even when every incident edge of atom $i$ is silenced ($s_5\to 0$ or $\eta\to 0$), removing $0/0$ indeterminacies.
The weight $w_{ij}^{(f,a)}$ is SO(3)-invariant by construction, so its use as a per-(focus, head) reweighting preserves equivariance.

Reshaping the channel axis of the edge message $\mathbf m_{ij}\in V_{\le
            L}\otimes\mathbb{R}^{H}$ along the focus/head factorization $H=F\cdot H_a\cdot
    d_a$ yields per-(focus,\,head) slices $\mathbf m_{ij}^{(f,a)}\in V_{\le
            L}\otimes\mathbb{R}^{d_a}$, and aggregation under the attention weights gives
\begin{equation}
    \mathbf A_i^{(f,a)}
    =\sum_{j:r_{ij}<r_{\mathrm{c}}}w_{ij}^{(f,a)}\,\mathbf m_{ij}^{(f,a)}
    \in V_{\le L}\otimes\mathbb{R}^{d_a}.
    \label{eq:dpa4-attn-agg}
\end{equation}
A destination-side scalar output gate then modulates each $(f,a)$ slice
multiplicatively,
\begin{equation}
    \widetilde{\mathbf A}_i^{(f,a)}
    = G_i^{(f,a)}\,\mathbf A_i^{(f,a)},
    \qquad
    G_i^{(f,a)}
    = \sigma\!\Bigl(
    \textstyle\sum_{c=1}^{C_f}
    W^{\mathrm{og}}_{c,f,a}\, N_0\bigl(\mathbf h_i'|_{l=0}\bigr)_{f,c} \Bigr)\in(0,1), \label{eq:dpa4-attn-gate} \end{equation} where $\sigma(t)=(1+e^{-t})^{-1}$ is the logistic sigmoid and $W^{\mathrm{og}}\in\mathbb{R}^{C_f\times F\times H_a}$ is a learnable output-gate tensor.

\paragraph{The convolution in closed form.}
Concatenating the gated slices $\widetilde{\mathbf A}_i^{(f,a)}$ over all $(f,a)$ pairs, the inverse of the focus/head split above, restores a single hidden-width tensor $\widetilde{\mathbf A}_i=\operatorname{concat}_{(f,a)}\widetilde{\mathbf A}_i^{(f,a)}\in V_{\le L}\otimes\mathbb{R}^{H}$.
The aggregated message is coupled to the destination node feature by a two-operand Wigner bilinear product $\mathcal{P}$ (Section~\ref{sec:s2-method}),
\begin{equation}
    \mathcal{M}_i
    = \mathcal{P}\bigl(\widetilde{\mathbf A}_i,\mathbf{h}_i\bigr)
    \in V_{\le L}\otimes\mathbb{R}^{H},
    \label{eq:dpa4-message-node}
\end{equation}
a residual branch initialized at a small scale.
Adding it and projecting back to representation width with a degree-wise channel post-mixer $L^{\mathrm{post}}_{\mathrm{deg}}:V_{\le L}\otimes\mathbb{R}^{H}\to V_{\le L}\otimes\mathbb{R}^{C}$ (Eq.~\eqref{eq:sm-degree-linear}) gives the EMFA SO(2) convolution at atom $i$,
\begin{equation}
    \bigl(\mathcal{C}_{\Theta}\mathbf{h}\bigr)_i
    =
    L^{\mathrm{post}}_{\mathrm{deg}}\bigl[\widetilde{\mathbf A}_i+\mathcal{M}_i\bigr]
    \in V_{\le L}\otimes\mathbb{R}^{C}.
    \label{eq:dpa4-conv}
\end{equation}
The architectural designs A1 and A2 are absorbed into the per-edge message $\mathbf{m}_{ij}$, while A3 appears explicitly through the attention weight $w_{ij}^{(f,a)}$ and the destination-side output gate $G_i^{(f,a)}$.
Equivariance follows because $w_{ij}^{(f,a)}$ and $G_i^{(f,a)}$ are SO(3)-invariant scalars and every other operation is either degree-wise, acts inside an SO(2)-equivariant local frame, or is the equivariant bilinear product of Section~\ref{sec:s2-method}.
The post-mixer is zero-initialized so $\mathcal{C}_{\Theta}\equiv 0$ at the start of training.

\subsubsection{Equivariant feed-forward network}
\label{sec:s2-method}

After each SO(2) convolution, DPA4 applies an equivariant feed-forward network (FFN) that acts on every atom independently and carries a residual connection.
Its nonlinearity cannot act component by component: a point-wise map applied to the components of a degree-$l>0$ feature destroys equivariance, so an equivariant network must obtain its nonlinearity by coupling degrees to one another.
DPA4 realizes the architectural design A4 through a \emph{Wigner bilinear network} (WBN): a stack of \emph{Wigner bilinear blocks} (WBBs) between equivariant linear maps.
Each block expands the node feature into a pair of scalar functions on the rotation group $SO(3)$, multiplies them point by point, and projects the product back onto irreducible coefficients.
That cycle realizes the Clebsch--Gordan tensor product without ever forming it: a single product couples all degrees at once, with no enumeration of coupling paths and no stored Clebsch--Gordan tables.
The same block, fed from two different inputs rather than one, also supplies the message--node coupling $\mathcal{P}$ of Eq.~\eqref{eq:dpa4-message-node}.

\paragraph{Construction.}
A WBN is the composition
\begin{equation}
    \mathcal{N}_{\mathrm{WBN}}
    = L_{\mathrm{read}}\circ\mathcal{B}_{K}\circ\cdots\circ\mathcal{B}_{1}\circ L_{\mathrm{emb}} ,
    \label{eq:dpa4-wbn}
\end{equation}
where $L_{\mathrm{emb}}$ and $L_{\mathrm{read}}$ are equivariant linear maps adapting the input and output spaces to a hidden space $V_{\le L}\otimes\mathbb{R}^{C}$, and every block $\mathcal{B}_{\kappa}:V_{\le L}\otimes\mathbb{R}^{C}\to V_{\le L}\otimes\mathbb{R}^{C}$ is a WBB.

The block rests on the structure of band-limited functions on $SO(3)$.
Such a function is determined by the coefficients of its expansion in Wigner matrices $D^{(l)}_{mk}$, on which a rotation acts through the row index $m$ alone.
The column index $k$, running over $|k|\le l$, therefore enumerates the $2l+1$ copies of degree $l$ that the function carries; the $k$ axis is stored at width $2L+1$, with the entries $|k|>l$ absent.

The block uses two degree-wise channel-mixing maps (Eq.~\eqref{eq:sm-degree-linear}) at FFN hidden width $H_{\mathrm{FFN}}$: $L^{\mathrm{ch}}_{\mathrm{in}}:V_{\le L}\otimes\mathbb{R}^{C}\to V_{\le L}\otimes\mathbb{R}^{2(2L+1)H_{\mathrm{FFN}}}$ and $L^{\mathrm{ch}}_{\mathrm{out}}:V_{\le L}\otimes\mathbb{R}^{(2L+1)H_{\mathrm{FFN}}}\to V_{\le L}\otimes\mathbb{R}^{C}$.
The first acts degree by degree on the node feature, its $2(2L+1)H_{\mathrm{FFN}}$ output channels read as a side $\mathrm{X}\in\{\mathrm{L},\mathrm{R}\}$, a frame $k$, and a hidden channel $c$,
\begin{equation}
    u^{\mathrm{X},(l)}_{m,k,c}
    = \bigl[L^{\mathrm{ch}}_{\mathrm{in}}(\mathbf{h}_i)\bigr]^{(l)}_{m,(\mathrm{X},k,c)}
    = \sum_{c'=1}^{C} W^{(l)}_{(\mathrm{X},k,c),\,c'}\,
      \mathbf{h}^{(l)}_{i,m,c'},
    \qquad |k|\le l,
    \label{eq:dpa4-wbb-lift}
\end{equation}
with $W^{(l)}$ the learnable matrix of $L^{\mathrm{ch}}_{\mathrm{in}}$ at degree $l$.
The index $m$ passes through untouched, so every frame at degree $l$ is a linear image of the same input coefficients: the input carries no frame axis, and the frame multiplicity is created here by the channel map alone.
The two arrays are then expanded into scalar functions on the group,
\begin{equation}
    f^{\mathrm{X}}_{c}(g)
    = \sum_{l=0}^{L}\ \sum_{m=-l}^{l}\ \sum_{k=-l}^{l}
      u^{\mathrm{X},(l)}_{m,k,c}\,D^{(l)}_{mk}(g),
    \qquad \mathrm{X}\in\{\mathrm{L},\mathrm{R}\},
    \label{eq:dpa4-wbb-synth}
\end{equation}
these functions are multiplied point by point, and Wigner orthogonality returns the product to coefficients,
\begin{equation}
    z^{(l)}_{m,k,c}
    = (2l+1)\int_{SO(3)} D^{(l)}_{mk}(g)\,f^{\mathrm{L}}_{c}(g)\,f^{\mathrm{R}}_{c}(g)\,dg ,
    \label{eq:dpa4-wbb-analysis}
\end{equation}
with $dg$ the normalized Haar measure and $D^{(l)}$ taken in the real basis carried by the features, so that no complex conjugation appears.
Only channel-diagonal products $f^{\mathrm{L}}_{c}f^{\mathrm{R}}_{c}$ are formed; the mixing that a product over all left--right channel pairs would supply is instead carried by the linear maps surrounding it, at a cost linear rather than quadratic in $H_{\mathrm{FFN}}$.

An invariant branch acts in parallel on the pair of $l=k=0$ slices of $u^{\mathrm{L}}$ and $u^{\mathrm{R}}$, written $\mathbf{u}^{0}_i\in\mathbb{R}^{2H_{\mathrm{FFN}}}$: it gates the coefficients and writes an additive term into the invariant slot,
\begin{equation}
    \widetilde{z}^{\,(l)}_{m,k,c}
    = \sigma\!\bigl(W_{\mathrm{g}}\mathbf{u}^{0}_i\bigr)_{c}\sum_{c'}Q_{cc'}z^{(l)}_{m,k,c'}
    + \delta_{l0}\delta_{m0}\delta_{k0}\,\mathrm{SwiGLU}(\mathbf{u}^{0}_i)_{c},
    \label{eq:dpa4-wbb-scalar}
\end{equation}
where $\sigma$ is the logistic function, the learnable $Q\in\mathbb{R}^{H_{\mathrm{FFN}}\times H_{\mathrm{FFN}}}$ is shared by all coefficients, and $W_{\mathrm{g}}\in\mathbb{R}^{H_{\mathrm{FFN}}\times 2H_{\mathrm{FFN}}}$ produces the gate.
The gated linear unit is $\mathrm{SwiGLU}(\mathbf{u})\coloneqq\mathbf{u}_{\mathrm{v}}\odot\mathrm{SiLU}(\mathbf{u}_{\mathrm{g}})$, with $\mathbf{u}_{\mathrm{g}},\mathbf{u}_{\mathrm{v}}\in\mathbb{R}^{H_{\mathrm{FFN}}}$ the two halves of its argument along the channel axis.
The block is completed by the output map and the residual connection; writing $\mathcal{P}(\mathbf{h},\mathbf{h}')$ for the map carrying the two operands to $L^{\mathrm{ch}}_{\mathrm{out}}[\widetilde{z}]$, which $L^{\mathrm{ch}}_{\mathrm{in}}$ cuts from a single input,
\begin{equation}
    \mathcal{B}(\mathbf{h}_i)
    = \mathbf{h}_i + L^{\mathrm{ch}}_{\mathrm{out}}\bigl[\widetilde{z}\bigr]
    = \mathbf{h}_i + \mathcal{P}(\mathbf{h}_i,\mathbf{h}_i),
    \label{eq:dpa4-ffn}
\end{equation}
with $L^{\mathrm{ch}}_{\mathrm{out}}$ zero-initialized, so that the block begins at the identity.
Feeding the two operands of $\mathcal{P}$ from different inputs instead, each lifted to the frame axis by its own channel map, gives the form used by the message--node coupling of Eq.~\eqref{eq:dpa4-message-node}.

\paragraph{Quadrature on the rotation group.}
The integral in Eq.~\eqref{eq:dpa4-wbb-analysis} is evaluated by a product rule, built on the factorization of a rotation into the point it assigns to the pole and a residual spin: every $g\in SO(3)$ can be written uniquely as $g=R_{\mathbf{q}}R_{z}(\varphi)$, where $\mathbf{q}=g\hat{\mathbf{z}}\in\Sph$, $R_{\mathbf{q}}$ is a fixed reference rotation carrying $\hat{\mathbf{z}}$ to $\mathbf{q}$, and $R_{z}(\varphi)$ is the rotation by $\varphi\in[0,2\pi)$ about $\hat{\mathbf{z}}$.
A quadrature on the group therefore combines a rule on the sphere with a rule on the angle.
A Lebedev rule of algebraic order $p$, with points $\{\mathbf{q}_a\}_{a=1}^{A}\subset\Sph$ and positive weights $\{w_a\}$ normalized to $\sum_a w_a=1$, supplies the spherical factor, and the $n_\varphi$ equally spaced angles $\varphi_\gamma=2\pi\gamma/n_\varphi$, $\gamma=0,\dots,n_\varphi-1$, supply the angular one, giving the $G=A\,n_\varphi$ nodes $g_{a\gamma}=R_{\mathbf{q}_a}R_{z}(\varphi_\gamma)$ with Haar weights $w_a/n_\varphi$.
On this rule the block evaluates its operands at the sample points and projects the product back to coefficients,
\begin{equation}
    z^{(l)}_{m,k,c}=(2l+1)\sum_{a,\gamma}\frac{w_a}{n_\varphi}
      D^{(l)}_{mk}(g_{a\gamma})\,f^{\mathrm{L}}_c(g_{a\gamma})\,f^{\mathrm{R}}_c(g_{a\gamma}),
    \label{eq:dpa4-wbb-quad}
\end{equation}
the operand values $f^{\mathrm{X}}_c(g_{a\gamma})$ being Eq.~\eqref{eq:dpa4-wbb-synth} read off at the nodes.
Both directions are matrices in the Wigner values at the nodes, depend only on the precomputed rule, and are formed once and cached as buffers.
Whenever the rule resolves the degrees present in the product, the two maps are mutually inverse on the band-limited space, which is the discrete counterpart of Wigner orthogonality; that condition is made precise next.

\paragraph{Exactness and equivariance.}
A rotation acts on the expanded functions by translating their argument, $f^{\mathrm{X}}_c(g)\mapsto f^{\mathrm{X}}_c(R^{-1}g)$; the point-wise product commutes with that action, and the projection of Eq.~\eqref{eq:dpa4-wbb-analysis} returns it to a rotation of the coefficients, so the exact block is equivariant.
The discrete block loses nothing, because the summand of Eq.~\eqref{eq:dpa4-wbb-quad} is a product of three Wigner entries of degree at most $L$, hence a polynomial of degree at most $3L$ in the sampled axis and, at fixed axis, a trigonometric polynomial of frequency at most $3L$ in $\varphi$.
A Lebedev rule of order $p\ge 3L$ crossed with $n_\varphi\ge 3L+1$ angles integrates every such term exactly, so the sampled block coincides with the exact one and the FFN is equivariant to floating-point precision.
Bounding those degrees requires the single bilinear product; a general point-wise nonlinearity would spread spectral weight beyond any band limit and leave the truncation only approximate.

\paragraph{Expressive power.}
The expressive power of an equivariant nonlinearity is measured by which equivariant maps a network built from it can represent.
The WBN is \emph{complete} in this sense: given sufficient band limit, width and depth, it approximates any continuous SO(3)-equivariant map between finite direct sums of irreps, uniformly on compact sets, and realizes polynomial equivariant maps exactly (Supplementary Theorem~\ref{thm:sm-wbn-complete}).
Working on the rotation group rather than the sphere is what makes this possible: with the full frame range $|k|\le l$ retained, the point-wise product realizes every Clebsch--Gordan coupling $l_1\otimes l_2\to l$ permitted by $|l_1-l_2|\le l\le l_1+l_2$.
A concise, self-contained proof is given in Supplementary Section~\ref{sec:sm-wbn-complete}.

\paragraph{Instantiation in DPA4.}
DPA4 instantiates the WBN economically.
Each FFN is itself a WBN acting on the node feature space: input, hidden and output spaces coincide, so neither the embedding $L_{\mathrm{emb}}$ nor the readout $L_{\mathrm{read}}$ is needed (Supplementary Remark~\ref{rem:sm-wbn-adapters}), and the blocks $\mathcal{B}_1,\dots,\mathcal{B}_K$ of Eq.~\eqref{eq:dpa4-wbn} are stacked directly, with $K=1$ in every released variant.
The frame range is truncated as well: only frames $|k|\le k_{\max}$ are retained, with $k_{\max}=1$ throughout this work.
Every formula above carries over unchanged, with the frame width $2L+1$ replaced by $2k_{\max}+1$; because the retained functions carry $\varphi$-frequencies at most $k_{\max}$, the angular part of the rule shrinks to $n_\varphi=3k_{\max}+1$ points while remaining exact.
For the reference configuration ($C=64$, $L=4$, $H_{\mathrm{FFN}}=192$) the quadrature is a Lebedev rule of order $13$ ($74$ points, satisfying $p\ge 3L$) crossed with $n_\varphi=4$ angles, i.e.\ $296$ group samples per atom.
Truncation gives up the hypothesis of the completeness result, which requires all frames, and is adopted for economy: $k_{\max}=1$ already restores the couplings that a sphere-restricted nonlinearity cannot form.

\paragraph{Relation to spherical-grid nonlinearities.}
At $k_{\max}=0$ only the entries $D^{(l)}_{m0}$ survive, which are, up to normalization, the spherical harmonics $Y_l^m$: the expanded functions are constant along $\varphi$ and descend to the quotient $\Sph=SO(3)/SO(2)$, so the WBB reduces to a spherical-grid bilinear nonlinearity of the kind used by Equiformer-family models~\cite{liao2023equiformerv2,liao2026equiformerv3}.
The restriction to the sphere is not free.
Re-expanding a product of spherical harmonics is governed by the Gaunt coefficients, which vanish whenever $l_1+l_2+l$ is odd.
The parity-odd couplings are therefore lost, among them the cross product $1\otimes 1\to 1$ and with it any construction of the signed volume $\det[\mathbf{r}_1,\mathbf{r}_2,\mathbf{r}_3]$; no composition of sphere-based bilinear layers can recover them.
Retaining $k=\pm 1$ reopens these channels at modest extra cost, which is what motivates the default $k_{\max}=1$.
The Lebedev axis rule attains a given algebraic order with substantially fewer points than the latitude--longitude tensor grids used by EquiformerV2 and EquiformerV3.
The measured equivariance error of the resulting FFN is at machine precision (Table~\ref{tab:s2_full_equivariance}; Supplementary Table~\ref{tab:sm_s2_truncated_equivariance}).

\subsubsection{Atomic energy head}
\label{sec:ener-head}

The scalar energy of the learned branch is read out in two steps, the first reducing the final node feature to an invariant vector $\mathbf{s}_i\in\mathbb{R}^{C}$ in one of two modes.
In the equivariant mode the Wigner bilinear block $\mathcal{B}$ of Eq.~\eqref{eq:dpa4-ffn} is applied to the final node feature and the $l=0$ slice of the result is extracted,
\begin{equation}
    \mathbf{s}_i = \bigl[\mathcal{B}\bigl(\mathbf{h}_i^{(N_{\mathrm{layer}})}\bigr)\bigr]^{(l=0)}
    \in\mathbb{R}^{C},
    \label{eq:dpa4-readout}
\end{equation}
so that the angular information carried by the $l>0$ coefficients is coupled into the invariant slot before the slice is taken; the residual connection of Eq.~\eqref{eq:dpa4-ffn} is part of the block.
In the scalar mode the $l>0$ coefficients are discarded first and the same block is applied at $L=0$, where only its invariant branch survives.
DPA4-Mini, DPA4-Neo, DPA4-Air and DPA4-Plus use the equivariant mode and DPA4-Pro the scalar mode (Supplementary Table~\ref{tab:sm_model_config}); every variant uses a single read-out block.
A gated multilayer perceptron then maps $\mathbf{s}_i$ to the atomic energy: each hidden layer projects to twice its width and splits the output into a value and a gate, $\mathbf{x}\leftarrow\mathbf{v}\odot\mathrm{SiLU}(\mathbf{g})$.
A bias-free linear map returns a single number, to which a per-species constant $b_{Z_i}$ is added,
\begin{equation}
    E^{\mathrm{NN}}_{\Theta,i}=\mathbf{w}^{\top}\mathrm{MLP}(\mathbf{s}_i)+b_{Z_i},
    \qquad
    E^{\mathrm{NN}}_{\Theta}=\sum_{i=1}^{N}E^{\mathrm{NN}}_{\Theta,i}.
    \label{eq:dpa4-ener-head}
\end{equation}
The released variants use a single hidden layer of width $32\lceil 8C/96\rceil$ ($192$ for the $C=64$ variants).
The network is shared by all species: chemistry enters through the embeddings of Section~\ref{sec:gie} and through $b_{Z_i}$, which is fixed to the per-species mean energy of the training set rather than learned.

\subsubsection{Native ZBL zone bridging}
\label{sec:zbl_method}

The analytical branch is a pairwise sum of the Ziegler--Biersack--Littmark screened Coulomb potential~\cite{ziegler1985},
\begin{equation}
    E^{\mathrm{ZBL}}(Z,R) = \tfrac{1}{2}\sum_{i\neq j}E_{ij}^{\mathrm{ZBL}}(r_{ij}),
    \qquad
    E_{ij}^{\mathrm{ZBL}}(r)= \frac{k_{\mathrm{e}}Z_iZ_j}{r}\,\Phi\!\left(\frac{r}{a_{ij}}\right),
    \qquad
    a_{ij}=\frac{0.88534\,a_0}{Z_i^{0.23}+Z_j^{0.23}},
    \label{eq:dpa4-zbl}
\end{equation}
with $k_{\mathrm{e}}$ the Coulomb constant and $a_0$ the Bohr radius.
The universal screening function takes the standard four-exponential form
\begin{equation}
    \Phi(x)=0.18175\,e^{-3.1998\,x}
    +0.50986\,e^{-0.94229\,x}
    +0.28022\,e^{-0.4029\,x}
    +0.028171\,e^{-0.20162\,x}.
    \label{eq:dpa4-zbl-phi}
\end{equation}
$E^{\mathrm{ZBL}}$ is evaluated on the raw pair distances $r_{ij}$, in contrast with the learned branch $E_{\Theta}^{\mathrm{NN}}$, which sees only the clamped distance $\widetilde{r}(r_{ij})$ (Sec.~\ref{sec:geom-inputs}).
Native ZBL zone bridging couples the two branches inside the energy model rather than as a post-hoc energy-level splice.

This distinction removes a force artifact that is intrinsic to conventional energy-level splicing.
In such a splice a coordinate-dependent switching function $\lambda_i(R)$ blends an analytical ZBL branch with a learned atom-wise energy, as in DP-ZBL-type pair corrections~\cite{wang2019dpzbl}, giving the energy
\begin{equation}
    E^{\mathrm{splice}}(Z,R) = \sum_i \left[ \lambda_i(R)E_i^{\mathrm{ZBL}}(Z,R) +\bigl(1-\lambda_i(R)\bigr)E_{\Theta,i}^{\mathrm{NN}}(Z,R) \right].
    \label{eq:zbl-splice-energy}
\end{equation}
Differentiating with respect to $\mathbf{R}_k$ gives the force
\begin{equation}
    \mathbf{F}_k^{\mathrm{splice}}
    =
    \mathbf{F}_{k}^{\mathrm{weighted}}
    -
    \sum_i
    \frac{\partial \lambda_i}{\partial \mathbf{R}_k}
    \left(
    E_i^{\mathrm{ZBL}}-E_{\Theta,i}^{\mathrm{NN}}
    \right),
    \label{eq:zbl-switching-force}
\end{equation}
where $\mathbf{F}_{k}^{\mathrm{weighted}}$ is the $\lambda$-weighted sum of the two branch forces.
The second term is a switching force proportional to the branch energy mismatch in the splice window.
It has no independent physical counterpart and vanishes only when the switching weight is constant or the two branches are exactly energy-matched throughout the switching region.

Native ZBL zone bridging carries no analogue of the switching force of Eq.~\eqref{eq:zbl-switching-force}: in Eq.~\eqref{eq:dpa4-total-energy} no coordinate-dependent switching weight multiplies the energy difference between the branches.
The clamped distance and the source-freeze gate (Sec.~\ref{sec:geom-inputs}) make $E_{\Theta}^{\mathrm{NN}}$ independent of $r_{jk}$ for every pair with $r_{jk}\le r_{\mathrm{in}}$.
Such a pair therefore contributes exactly the ZBL pair force, while many-body contributions involving non-frozen neighbors are unaffected.

\subsubsection{Symmetry guarantees}
\label{sec:symmetry}

The total energy $E(Z,R)$ in Eq.~\eqref{eq:dpa4-total-energy} is invariant under translations, under permutations of atoms of the same chemical species and under global rotations, is $C^3$-smooth in $R$ away from coincident atoms, and gives conservative forces.
Translation invariance follows because every geometric input is built from relative displacements.
Permutation invariance follows because all edge and node operators are globally shared and depend on species only through learned embeddings of $Z_i$, and neighbor aggregation is by summation or softmax-weighted summation.
Rotational invariance of $E$ follows because the Geometry-Informed Embedding (Sec.~\ref{sec:gie}), the EMFA SO(2) convolution (Sec.~\ref{sec:so2}), the equivariant feed-forward network (Sec.~\ref{sec:s2-method}) and the equivariant RMS norm are all SO(3)-equivariant on $V_{\leq L}\otimes\mathbb{R}^{C}$, while the atomic energy head reads out only the $l=0$ invariant slice.
The analytical ZBL branch is invariant as well, depending only on the scalar distances $r_{ij}$.
Smoothness follows from the $C^3$ cutoff envelope $s_5$, the $C^3$ clamped distance $\widetilde{r}$ and source-freeze gate $\eta_j$ (Sec.~\ref{sec:geom-inputs}), and the softplus stabilizer in the attention denominator (Eq.~\eqref{eq:dpa4-attention}).
Forces and virials are obtained by automatic differentiation of the single scalar energy in Eq.~\eqref{eq:dpa4-total-energy}, so the force field is conservative.

\subsection{Training}

DPA4 is trained as a conservative potential, so the loss constrains the energy and its coordinate and cell derivatives jointly.
For a mini-batch of configurations $b=1,\ldots,B$, with $N_b$ atoms in configuration $b$, the training objective is
\begin{equation}
    \mathcal{L}
    = \lambda_E\,\frac{1}{B}\sum_{b=1}^{B}\frac{|E_{\Theta,b}-E_b|}{N_b}
    + \lambda_F\,\frac{1}{\sum_{b}N_b}\sum_{b=1}^{B}\sum_{i=1}^{N_b}\bigl\|\mathbf{F}_{\Theta,bi}-\mathbf{F}_{bi}\bigr\|_2
    + \lambda_\Xi\,\frac{1}{B}\sum_{b=1}^{B}\frac{\bigl\|\Xi_{\Theta,b}-\Xi_b\bigr\|_1}{9N_b}.
    \label{eq:dpa4-loss}
\end{equation}
Here $E_{\Theta,b}$, $\mathbf{F}_{\Theta,bi}$ and $\Xi_{\Theta,b}$ denote the DPA4 energy, atomic forces and virial tensor, while $E_b$, $\mathbf{F}_{bi}$ and $\Xi_b$ denote the reference DFT labels.
The force term averages the Euclidean norm of each atomic force-vector residual over all atoms of the batch; the energy term is normalized per atom and the virial term per atom and per tensor component.
The weights $\lambda_E$, $\lambda_F$ and $\lambda_\Xi$ are the loss weights $(E,F,V)$ listed for each benchmark in Supplementary Tables~\ref{tab:sm_main_ablation_config}--\ref{tab:sm_matbench_config}.

All benchmark models train in bf16 mixed precision, with FP32 geometric reductions and TF32 matrix products where available, under a warmup--stable--decay learning-rate schedule~\cite{wen2024wsd} (Supplementary Section~\ref{sec:sm-wsd-schedule}) and the HybridMuon optimizer~\cite{jordan2024muon,liu2025muon} (Supplementary Section~\ref{sec:sm-hybrid-muon}).
HybridMuon routes matrix-valued hidden transformations to Muon updates and scalar, normalization or auxiliary parameters to Adam-family updates.
For degree-wise equivariant linear maps, slice-mode Muon applies an independent matrix update to each degree-$l$ channel block, preserving the representation-block structure instead of flattening all degrees into one matrix.
The Muon path uses match-RMS scaling to keep its update magnitude on the same learning-rate scale as the Adam-family path.
It also uses Magma-lite alignment damping~\cite{joo2026magma} to attenuate Muon blocks whose current gradients are poorly aligned with their momentum (Supplementary Section~\ref{sec:sm-magma-lite}).
The remaining model- and dataset-specific hyperparameters are reported in the same tables.

\subsection{Compiled conservative energy-gradient training}
\label{sec:compiled-training}

We compile the conservative energy-gradient path: compilation changes how the training step is executed, not what it computes, and forces remain the negative gradient of the predicted energy.
The obstacle is that the force is itself a coordinate derivative of the energy, so training on forces requires a double backward pass,
\begin{equation}
    \mathbf{F}_{\Theta}
    = -\frac{\partial E_{\Theta}}{\partial \mathbf{R}},
    \qquad
    \frac{\partial\mathcal{L}}{\partial\Theta}
    \supset
    \frac{\partial^2E_{\Theta}}{\partial\mathbf{R}\,\partial\Theta},
    \label{eq:dpa4-double-backward}
\end{equation}
leaving a coordinate--parameter mixed derivative in the training gradient.
A standard compiled forward/backward stack captures the forward graph and its first reverse-mode derivative, but does not expose a nested coordinate derivative inside the compiled region.
DPA4 executes the energy-to-force derivative during symbolic tracing and lowers the resulting tensor graph with PyTorch Inductor~\cite{ansel2024pytorch2}.
The compiled graph then contains the energy evaluation and its coordinate derivative, while the outer backward pass differentiates the force residual with respect to the model parameters.

\paragraph{Tracing the conservative graph.}
The compiled function wraps the energy computation as a tensor-only map: from coordinates, atom types, neighbor indices and optional conditioning tensors to energies, forces and virials.
Before entering the traced function, the coordinates are rebound to a fresh leaf tensor.
This restart removes any upstream graph carried by data loading or neighbor construction while preserving a well-defined coordinate endpoint for $\partial E_{\Theta}/\partial\mathbf{R}$.
During training, the force construction keeps the coordinate-derivative graph alive, which is what lets the outer loss backward pass reach the parameters through the force residual.

\paragraph{Symbolic tracing and higher-order differentiability.}
We trace the conservative graph with symbolic shapes and real tensor inputs.
Real inputs are needed because compact edge construction contains data-dependent operations; once that control flow is resolved, the runtime dimensions are represented symbolically.
The trace pads a representative batch so that the frame, total-atom and local-atom counts take distinct primes no smaller than five, none of them equal to any dimension of the model.
Primes at least five avoid the extents the graph treats specially: one for broadcasting, and the two, three and nine appearing as charge--spin, Cartesian and virial literals in the model code.
Distinctness stops the tracer from unifying two runtime dimensions into a single symbol.
The SiLU backward operation is decomposed into elementary pointwise operations before tracing, so the compiler receives an explicit first-derivative graph when the outer backward pass differentiates it again.

\paragraph{Preserving the force-loss gradient.}
When the coordinate derivative is traced in training mode, autograd wraps every saved forward activation in a chain of two detach nodes.
In the traced graph these nodes become ordinary tensor operations and sever the path from the force loss back to the parameters.
We remove a detach node when its input is another detach node, or when all of its users are detach nodes.
Those two patterns identify an autograd-inserted chain; every other detach operation was written deliberately by the model code and is kept.
The edited graph is then rebuilt into a fresh graph before compilation, so all compiler passes see a consistent node topology.

\paragraph{Dynamic edge representation and cache structure.}
Edge vectors are formed by differentiable indexed selection from the coordinate tensor, so the coordinate-derivative path runs through them.
Because the edge count is data-dependent and therefore symbolic, a single masked sentinel edge is appended to every batch; it guarantees a nonempty edge tensor while contributing exactly zero to downstream reductions.
Compiled callables are cached under a key recording whether the graph is a training or an evaluation graph, whether atomic virials are produced and whether optional auxiliary inputs are present.
This multi-slot cache avoids recompilation when training is periodically interrupted by validation, while keeping distinct output signatures in separate compiled graphs.

\paragraph{Inductor configuration.}
The traced graph is lowered with dynamic-shape compilation and deterministic settings rather than autotuning, with shape padding enabled for fluctuating symbolic dimensions.
We disable the Triton features that misbehave on the combination of a data-dependent edge count and a second-order autograd graph, most visibly the oversized fused reduction kernels that fail to compile on some GPU and Triton combinations.
Every generated kernel is constrained to a one-dimensional launch grid, because the default tiling can place the data-dependent edge or node axis on a launch dimension whose extent is limited to $65535$.
The evaluation graph additionally disables the peak-memory reordering pass, whose cost model needs size hints that the traced placeholders do not carry.

This design keeps conservative training inside compiled GPU code without replacing conservative force matching with a direct-force surrogate; its measured effect on training time, memory and accuracy is reported in Section~\ref{sec:efficiency}.


\section{Acknowledgments}

We gratefully acknowledge the support received for this work.
The work of Linfeng Zhang is supported by the Advanced Materials-National Science and Technology Major Project, China (No. 2024ZD0606900).
The work of Tiancheng Li and Jianming Xue is supported by the National Natural Science Foundation of China (Grant No.~12135002).
The work of Han Wang is supported by the National Natural Science Foundation of China (grant nos. 12525113 and 12561160120) and the National Key R\&D Program of China (grant no. 2022YFA1004300).

\section{Data availability}

The DPA4 training and inference codes are available in the DeePMD-kit repository (\url{https://github.com/deepmodeling/deepmd-kit}) from version 3.2.0.

\bibliographystyle{naturemag-doi}
\bibliography{ref}

\clearpage

\renewcommand{\thefigure}{S-\arabic{figure}}
\setcounter{figure}{0}
\renewcommand{\thetable}{S-\arabic{table}}
\setcounter{table}{0}
\renewcommand{\thesection}{S-\arabic{section}}
\setcounter{section}{0}
\makeatletter
\@ifundefined{theHfigure}
{\newcommand{\theHfigure}{S-\arabic{figure}}}
{\renewcommand{\theHfigure}{S-\arabic{figure}}}
\@ifundefined{theHtable}
{\newcommand{\theHtable}{S-\arabic{table}}}
{\renewcommand{\theHtable}{S-\arabic{table}}}
\@ifundefined{theHsection}
{\newcommand{\theHsection}{S-\arabic{section}}}
{\renewcommand{\theHsection}{S-\arabic{section}}}
\makeatother

\FloatBarrier
\begin{center}
    {\LARGE\bfseries Supplementary Information for

        \textit{Pushing the accuracy--cost frontier of machine-learning
            interatomic potentials} \par}
    \vspace{1em}
\end{center}

\addtocontents{toc}{\protect\setcounter{tocdepth}{2}}
\tableofcontents
\clearpage

\section{Mathematical formulation, operators and proofs}
\label{sec:sm-model-formulation}
\label{sec:sm-methods}

This section collects the mathematical details that supplement the architecture description in Section~\ref{sec:methods}: the formal SO(3) representation notation, the edge-local frame and SO(2) decomposition, the closed-form cutoff envelope coefficients, the truncated SO(2) layout and rescaled lift, the SO(2)-equivariant operator algebra, and the cutoff-consistent first-layer bias correction used by the ablations in later supplementary sections.
It closes with two results on the equivariant nonlinearity: the numerical equivariance of the truncated local layouts used inside the EMFA SO(2) convolution, and a proof that the Wigner bilinear network is complete in its full-frame form.

\subsection{Notation for SO(3) representations}
\label{sec:sm-representation}

For each integer $l\ge 0$, let $V_l\cong\mathbb{R}^{2l+1}$ be the real $(2l+1)$-dimensional irreducible representation of SO(3), spanned by the real spherical harmonics $\{Y_l^m\}_{m=-l}^{l}$.
The action of a rotation
$Q\in\mathrm{SO}(3)$ on $V_l$ is given by the real Wigner $D$-matrix
$D^l(Q)\in\mathrm{O}(2l+1)$, the orthogonal $(2l+1)\times(2l+1)$ matrix
that describes how the real spherical harmonics transform under rotation,
\begin{equation}
    Y_l^m(Q^{-1}\widehat{\mathbf{r}}) =\sum_{m'=-l}^{l}D^l(Q)_{m,m'}\,Y_l^{m'}(\widehat{\mathbf{r}}).
    \label{eq:sm-wigner-D}
\end{equation}
The map $D^l:\mathrm{SO}(3)\to\mathrm{O}(2l+1)$ is a continuous group homomorphism (i.e.\ $D^l(QQ')=D^l(Q)D^l(Q')$), and the $V_l$ are \emph{irreducible}: $V_l$ admits no SO(3)-invariant subspace other than $\{0\}$ and $V_l$ itself.
The trivial case $l=0$ is the invariant scalar representation with $D^0(Q)\equiv 1$.

The node-feature space $V_{\le L}\otimes\mathbb{R}^{C}=\bigoplus_{l=0}^{L}
    V_l\otimes\mathbb{R}^{C}$ of main-text Section~\ref{sec:methods} therefore
carries the real block-diagonal action
\begin{equation}
    D(Q)=\bigoplus_{l=0}^{L}D^l(Q), \label{eq:sm-direct-sum} \end{equation} acting independently on each degree-$l$ block and trivially on the channel factor $\mathbb{R}^{C}$.
Basis coefficients are packed by the linear index
\begin{equation}
    \iota(l,m)=l^2+l+m,\qquad
    m=-l,\ldots,l,\quad l=0,\ldots,L,
    \label{eq:sm-pack-index}
\end{equation}
used throughout this Supplement and in main-text Eq.~\eqref{eq:dpa4-gie};
under this packing, the matrix $D(Q)$ is block-diagonal with the
$(2l+1)\times(2l+1)$ block $D^l(Q)$ occupying rows and columns
$\iota(l,-l),\ldots,\iota(l,+l)$.
A map $f:V_{\le L}\otimes\mathbb{R}^{C}\to V_{\le L}\otimes\mathbb{R}^{C'}$ is \emph{SO(3)-equivariant} if $f\circ D(Q)=D(Q)\circ f$ for every $Q\in\mathrm{SO}(3)$, and \emph{SO(3)-invariant} if it lands in the $l=0$ block, on which $D^0\equiv 1$.

\subsection{Edge-local frame and SO(2) decomposition}
\label{sec:sm-local-frame}

\paragraph{Goal.}
For each directed edge $(i,j)$, DPA4 chooses a rotation $R_{ij}$ satisfying
\begin{equation}
    R_{ij}\widehat{\mathbf{r}}_{ij}=\mathbf{e}_z=(0,0,1)^{\top}.
    \label{eq:sm-local-frame-goal}
\end{equation}
This converts an SO(3) problem into an SO(2) problem in the edge-local frame: after the bond direction is aligned with $\mathbf{e}_z$, the residual symmetry is the subgroup of rotations about $\mathbf{e}_z$.
The in-plane basis is a gauge choice.
The equivariant message does not depend on this choice because the local operator commutes with the residual SO(2) action.

\paragraph{Two quaternion charts.}
Let $\widehat{\mathbf{r}}=(x,y,z)\in \Sph$.
DPA4 uses two smooth unit-quaternion
charts,
\begin{equation}
    \mathbf{q}^{+}(\widehat{\mathbf{r}})
    =
    \frac{(1+z,\;y,\;-x,\;0)}{\sqrt{2(1+z)}},
    \qquad
    \mathbf{q}^{-}(\widehat{\mathbf{r}})
    =
    \frac{(-x,\;0,\;1-z,\;y)}{\sqrt{2(1-z)}}.
    \label{eq:sm-quaternion-charts}
\end{equation}
The first chart is regular away from the south pole and the second is regular away from the north pole.
In the overlap, the sign of $\mathbf{q}^{-}$ is chosen
so that $\langle\mathbf{q}^{+},\mathbf{q}^{-}\rangle\ge 0$, and the blended
quaternion is
\begin{equation}
    \mathbf{q}_{ij}
    =
    \frac{\lambda\mathbf{q}^{+}(\widehat{\mathbf{r}}_{ij})
        +(1-\lambda)\mathbf{q}^{-}(\widehat{\mathbf{r}}_{ij})}
    {\left\|\lambda\mathbf{q}^{+}(\widehat{\mathbf{r}}_{ij})
        +(1-\lambda)\mathbf{q}^{-}(\widehat{\mathbf{r}}_{ij})\right\|},
    \qquad
    \lambda=\frac{1+z}{2}.
    \label{eq:sm-quaternion-blend}
\end{equation}
The denominator is bounded away from zero after shortest-arc sign alignment in the chosen chart overlap, so the resulting gauge is smooth on the chart used by the implementation.
The Wigner-D matrix in the edge-local gauge is denoted $D_{ij}=D(R(\mathbf{q}_{ij}))$.
The underlying unit-quaternion rotation matrix is
\begin{equation}
    R(\mathbf{q})=
    \begin{pmatrix}
        1-2(q_y^2+q_z^2) & 2(q_xq_y-q_wq_z) & 2(q_xq_z+q_wq_y) \\
        2(q_xq_y+q_wq_z) & 1-2(q_x^2+q_z^2) & 2(q_yq_z-q_wq_x) \\
        2(q_xq_z-q_wq_y) & 2(q_yq_z+q_wq_x) & 1-2(q_x^2+q_y^2)
    \end{pmatrix}.
    \label{eq:sm-quaternion-rotation}
\end{equation}
During training, an additional random roll about the local $z$ axis may be composed with $R_{ij}$.
This roll is a gauge augmentation: it changes the in-plane basis of the local gauge but leaves the bond direction fixed.
Because the SO(2) stack is gauge equivariant, the lifted global message is unchanged by this roll apart from the prescribed SO(3) transformation law.

\paragraph{SO(2) decomposition.}
Let $\{\mathbf{e}^{(l)}_m\}_{m=-l}^{l}$ denote the real-spherical-harmonic basis of $V_l$.
Under restriction to the SO(2) subgroup of rotations about
the local $z$ axis, $V_l$ decomposes as
\begin{equation}
    V_l|_{\mathrm{SO}(2)} = \mathrm{span}_{\mathbb{R}}\{\mathbf{e}^{(l)}_0\} \;\oplus\; \bigoplus_{m=1}^{l} \mathrm{span}_{\mathbb{R}}\{\mathbf{e}^{(l)}_{-m},\,\mathbf{e}^{(l)}_{+m}\}, \label{eq:sm-so2-decomposition} \end{equation} where $\mathrm{span}_{\mathbb{R}}\{\cdot\}$ denotes the real linear span of the indicated basis vectors.
The $m=0$ summand is a one-dimensional trivial SO(2) sub-representation, and each $m=1,\ldots,l$ summand is a two-dimensional real SO(2) sub-representation on which a $z$-axis rotation by angle $\theta$ acts as a planar rotation by angle $m\theta$, equivalent to multiplication by the complex phase $e^{im\theta}$.
This is the algebraic reason DPA4 can use SO(2)-equivariant local operators instead of Clebsch--Gordan tensor products.

\subsection{Closed-form cutoff coefficients}
\label{sec:sm-radial}

The cutoff envelopes of the learned branch (main-text Section~\ref{sec:gie}) are the polynomials
\begin{equation}
    s_p(r) =
    \begin{cases}
        1+x^p\bigl(a_p+b_p x+c_p x^2+d_p x^3\bigr), & x=r/r_{\mathrm{c}}\in[0,1), \\
        0,                                          & x\ge 1,
    \end{cases}
    \label{eq:sm-envelope}
\end{equation}
whose four boundary conditions
$s_p(r_{\mathrm{c}})=s_p'(r_{\mathrm{c}})=s_p''(r_{\mathrm{c}})=s_p'''(r_{\mathrm{c}})=0$
uniquely fix the coefficients to
\begin{align}
    a_p & =-\frac{(p+1)(p+2)(p+3)}{6}, &
    b_p & =\frac{p(p+2)(p+3)}{2},        \\
    c_p & =-\frac{p(p+1)(p+3)}{2},     &
    d_p & =\frac{p(p+1)(p+2)}{6}.
    \label{eq:sm-cutoff-coefficients}
\end{align}
For the two values used in DPA4 ($s_5$ for edge weighting and the degree
normalization, $s_7$ in the Bessel radial basis of main-text
Section~\ref{sec:gie}), the explicit polynomials are
\begin{align}
    s_5(r) & =1-56x^5+140x^6-120x^7+35x^8,       \\
    s_7(r) & =1-120x^7+315x^8-280x^9+84x^{10},
\end{align}
with $x=r/r_{\mathrm{c}}$.

\subsection{Truncated local layout and rescaled lift}
\label{sec:sm-truncation}

Inside the edge-local frame, DPA4 retains only
\begin{equation}
    \mathcal{I}_M=
    \{(l,m):0\le l\le L,\ |m|\le \min(l,M)\},
    \qquad
    D_M=|\mathcal{I}_M|.
    \label{eq:sm-trunc-index}
\end{equation}
Let $P_M:\mathbb{R}^{(L+1)^2}\to\mathbb{R}^{D_M}$ be the orthogonal projection onto this reduced layout.
The truncated edge rotation is
\begin{equation}
    D_{ij}^{\le M}=P_M D_{ij}.
    \label{eq:sm-truncated-rotation}
\end{equation}
Because $P_M^{\top}P_M$ is not the identity when $M<L$, the round trip loses degree-block norm.
DPA4 applies the diagonal lift compensation
\begin{equation}
    (\Upsilon_M)_{\iota(l,m),\iota(l,m)}
    =
    \kappa_l
    =
    \sqrt{\frac{2l+1}{2\min(l,M)+1}}.
    \label{eq:sm-rescale}
\end{equation}
Thus $\Upsilon_M= \operatorname{diag}(\kappa_0 I_1,\kappa_1 I_3,\ldots,\kappa_L I_{2L+1})$.

\subsection{Equivariant operator algebra}
\label{sec:sm-operators}

\paragraph{Classification by Schur's lemma.}
The linear operators used by DPA4 follow from the classification of equivariant maps on the relevant representation spaces.
In the global
SO(3) frame, Schur's lemma forbids mixing distinct degrees:
\begin{equation}
    \mathrm{Hom}_{\mathrm{SO}(3)}
    (V_{\le L}\otimes\mathbb{R}^{C},V_{\le L}\otimes\mathbb{R}^{C'})
    \cong
    \bigoplus_{l=0}^{L}\mathbb{R}^{C\times C'}.
    \label{eq:sm-schur-so3}
\end{equation}
In the edge-local SO(2) frame, the $m=0$ lines are trivial representations
and each $|m|>0$ pair is of complex type, giving
\begin{equation}
    \mathrm{Hom}_{\mathrm{SO}(2)}
    (V_{\le L}\otimes\mathbb{R}^{C},V_{\le L}\otimes\mathbb{R}^{C'})
    \cong
    \bigoplus_{l,l'}\mathbb{R}^{C\times C'}
    \oplus
    \bigoplus_{m=1}^{M}\bigoplus_{l,l'\ge m}\mathbb{C}^{C\times C'}.
    \label{eq:sm-schur-so2}
\end{equation}
Equation~\eqref{eq:sm-schur-so2} is the formal reason the local-frame operator may mix degrees $l,l'$ while preserving each $|m|$ stratum.

\paragraph{Global SO(3) linear maps.}
By Eq.~\eqref{eq:sm-schur-so3}, an SO(3)-equivariant \emph{degree-wise} map
has the explicit form
\begin{equation}
    (L^{\mathrm{deg}}_{\Theta}\mathbf{h})_{\iota(l,m),c'}
    =
    \sum_{c=1}^{C}
    W^{(l)}_{c,c'}\mathbf{h}_{\iota(l,m),c}, \label{eq:sm-degree-linear} \end{equation} with one learnable channel matrix $W^{(l)}\in\mathbb{R}^{C\times C'}$ per degree $l$.
Here ``degree-wise'' means \emph{per-$l$ but identical across all $m\in\{-l,\ldots,l\}$ within a degree}: the same $W^{(l)}$ is applied to every $m$-slice of the degree-$l$ block, with different matrices allowed for different $l$.
All ``degree-wise channel-mixing maps'' in main-text Section~\ref{sec:methods} (including $L^{\mathrm{pre}}_{\mathrm{deg}}$, $L^{\mathrm{post}}_{\mathrm{deg}}$, $L^{\mathrm{rad}}_{\mathrm{lift}}$, $L^{\mathrm{ch}}_{\mathrm{in}}$ and $L^{\mathrm{ch}}_{\mathrm{out}}$) are instances of this form.

\paragraph{Edge-local SO(2) linear maps.}
In the edge-local frame, the $m=0$ lines are trivial SO(2) representations, whereas each $|m|>0$ pair is of complex type.
The most general SO(2)-equivariant
linear map on the reduced layout is therefore
\begin{align}
    (L_{\Theta}^{\mathrm{SO2}}\mathbf{x})_{(l,0),c'}
     & =
    \sum_{l'=0}^{L}\sum_{c=1}^{C}
    A_{c,c'}^{(l,l',0)}\mathbf{x}_{(l',0),c}
    +b_{0,c'}\delta_{l,0},
    \label{eq:sm-so2-m0} \\
    \begin{pmatrix}
        (L_{\Theta}^{\mathrm{SO2}}\mathbf{x})_{(l,-m),c'}\\
        (L_{\Theta}^{\mathrm{SO2}}\mathbf{x})_{(l,+m),c'}
    \end{pmatrix}
     & =
    \sum_{l'\ge m}\sum_{c=1}^{C}
    \begin{pmatrix}
        U_{c,c'}^{(l,l',m)} & -V_{c,c'}^{(l,l',m)} \\
        V_{c,c'}^{(l,l',m)} & U_{c,c'}^{(l,l',m)}
    \end{pmatrix}
    \begin{pmatrix}
        \mathbf{x}_{(l',-m),c} \\
        \mathbf{x}_{(l',+m),c}
    \end{pmatrix}.
    \label{eq:sm-so2-mn}
\end{align}

\paragraph{Equivariant RMS normalization.}
For $\mathbf{h}\in V_{\le L}\otimes\mathbb{R}^{C}$, define
\begin{equation}
    \sigma^2(\mathbf{h})
    =
    \sum_{l=0}^{L}\sum_{m=-l}^{l}\sum_{c=1}^{C}
    \frac{(\mathbf{h}_{\iota(l,m),c}-\delta_{l,0}\bar{h})^2}
    {(2l+1)(L+1)C},
    \qquad
    \bar{h}=C^{-1}\sum_c \mathbf{h}_{\iota(0,0),c}.
    \label{eq:sm-rms-stat}
\end{equation}
The normalization
\begin{equation}
    (N_{\gamma,\beta}\mathbf{h})_{\iota(l,m),c}
    =
    \gamma_l
    \frac{\mathbf{h}_{\iota(l,m),c}-\delta_{l,0}\bar{h}}
    {\sqrt{\sigma^2(\mathbf{h})+\varepsilon}}
    +\delta_{l,0}\beta_c
    \label{eq:sm-rmsnorm}
\end{equation}
commutes with SO(3), because the numerator is equivariant, the denominator is
invariant under the orthogonal representation $D^l$, and $\gamma_l$ is constant
within each degree block.

\paragraph{Scalar-gated nonlinearity.}
For a smooth scalar nonlinearity $\psi$ and degree-wise gate matrices
$G^{(l)}$, define
\begin{equation}
    (\Gamma_{\psi,G}\mathbf{h})_{\iota(l,m),c}
    =
    \begin{cases}
        \psi(\mathbf{h}_{\iota(0,0),c}),                  & l=0,    \\[2pt]
        \mathbf{h}_{\iota(l,m),c}
        \sigma\!\left(\sum_{c'}
        G^{(l)}_{c',c} \mathbf{h}_{\iota(0,0),c'}\right), & l\ge 1.
    \end{cases}
    \label{eq:sm-gated}
\end{equation}
The gate is a scalar function of the invariant $l=0$ slice, so the operation is SO(3)-equivariant.

\paragraph{Scalar-keyed mixtures.}
If $\mathbf{x}^{(1)},\ldots,\mathbf{x}^{(S)}$ are equivariant tensors of the
same type and $\pi$ is an invariant scalar projection, then
\begin{equation}
    \mathcal{A}(\mathbf{x}^{(1)},\ldots,\mathbf{x}^{(S)};\mathbf{q})
    =
    \sum_{s=1}^{S}\alpha_s\mathbf{x}^{(s)},\qquad
    \alpha_s=
    \frac{\exp\langle\mathbf{q},\pi(\mathbf{x}^{(s)})\rangle}
    {\sum_{s'}\exp\langle\mathbf{q},\pi(\mathbf{x}^{(s')})\rangle}
    \label{eq:sm-scalar-keyed-mixture}
\end{equation}
is equivariant, because the weights are invariant scalars.

\subsection{Bias consistency at the smooth cutoff}
\label{sec:sm-so2conv}

If the first SO(2) linear map of the EMFA SO(2) convolution (main-text Sec.~\ref{sec:so2}; Eqs.~\eqref{eq:sm-so2-m0}, \eqref{eq:sm-so2-mn}) includes an additive $l=0$ bias, that bias must vanish with the same smooth envelope as the rest of the edge message.
Otherwise an edge whose radial contribution has gone to zero could still carry a constant scalar offset.
DPA4 therefore treats the first-layer scalar bias as part of the envelope-gated edge response.
If $b_{0,f,c}$ is the scalar bias for focus $f$
and channel $c$, the net contribution to the local $(l,m)=(0,0)$ slot is
adjusted to
\begin{equation}
    b_{0,f,c}\,\widetilde{\rho}_{ij,0,f,c}\,s_5(r_{ij}),
    \label{eq:sm-bias-consistency}
\end{equation}
which is $C^3$ and vanishes at the cutoff.
This correction is only needed at the first SO(2) layer because subsequent layers operate on already modulated local features.

\subsection{Numerical equivariance of truncated local layouts}
\label{sec:sm-truncated-equivariance}

Table~\ref{tab:sm_s2_truncated_equivariance} extends the full-coefficient comparison of main-text Table~\ref{tab:s2_full_equivariance} to the $m$-truncated local layout used inside the EMFA SO(2) convolution, evaluating the maximum equivariance error of product-grid rules and Lebedev quadrature under random rotations about the local $z$ axis.

\begin{table}[p]
    \PaperTableStyle
    \begin{threeparttable}
        \caption{Equivariance of the $m$-truncated spherical-grid nonlinearity under random local $z$-axis rotations.}
        \label{tab:sm_s2_truncated_equivariance}
        \begin{tabular*}{\linewidth}{@{\extracolsep{\fill}}cccccccccc@{}}
            \toprule
            \textbf{$M$} &
            \textbf{$L$} &
            \multicolumn{4}{c}{\textbf{Product grid}} &
            \multicolumn{4}{c}{\textbf{Lebedev quadrature}} \\
            \cmidrule(lr){3-6}\cmidrule(lr){7-10}
            & & \textbf{Rule} & \textbf{\#\,pts} & \textbf{fp64 error} & \textbf{fp32 error}
            & \textbf{$p$} & \textbf{\#\,pts} & \textbf{fp64 error} & \textbf{fp32 error} \\
            \midrule
            1 & 2 & $6\times8$   & 48  & $2.36\times10^{-7}$ & $3.58\times10^{-7}$ & 7  & 26  & $2.31\times10^{-14}$ & $2.38\times10^{-7}$ \\
            1 & 3 & $6\times12$  & 72  & $1.22\times10^{-7}$ & $5.96\times10^{-7}$ & 9  & 38  & $3.55\times10^{-14}$ & $2.98\times10^{-7}$ \\
            1 & 4 & $6\times14$  & 84  & $1.12\times10^{-6}$ & $9.54\times10^{-7}$ & 13 & 74  & $1.04\times10^{-13}$ & $9.54\times10^{-7}$ \\
            1 & 5 & $6\times18$  & 108 & $1.10\times10^{-7}$ & $1.43\times10^{-6}$ & 15 & 86  & $9.34\times10^{-14}$ & $7.15\times10^{-7}$ \\
            1 & 6 & $6\times20$  & 120 & $7.64\times10^{-7}$ & $1.91\times10^{-6}$ & 19 & 146 & $8.56\times10^{-14}$ & $2.15\times10^{-6}$ \\
            1 & 7 & $6\times24$  & 144 & $2.17\times10^{-7}$ & $1.91\times10^{-6}$ & 21 & 170 & $2.08\times10^{-13}$ & $3.34\times10^{-6}$ \\
            \midrule
            2 & 2 & $8\times8$   & 64  & $4.01\times10^{-7}$ & $8.34\times10^{-7}$ & 7  & 26  & $1.50\times10^{-14}$ & $2.38\times10^{-7}$ \\
            2 & 3 & $8\times12$  & 96  & $5.99\times10^{-7}$ & $8.34\times10^{-7}$ & 9  & 38  & $5.71\times10^{-14}$ & $3.58\times10^{-7}$ \\
            2 & 4 & $8\times14$  & 112 & $6.02\times10^{-7}$ & $1.67\times10^{-6}$ & 13 & 74  & $9.15\times10^{-14}$ & $5.96\times10^{-7}$ \\
            2 & 5 & $8\times18$  & 144 & $1.19\times10^{-6}$ & $1.55\times10^{-6}$ & 15 & 86  & $7.83\times10^{-14}$ & $4.77\times10^{-7}$ \\
            2 & 6 & $8\times20$  & 160 & $1.33\times10^{-6}$ & $2.15\times10^{-6}$ & 19 & 146 & $1.29\times10^{-13}$ & $9.54\times10^{-7}$ \\
            2 & 7 & $8\times24$  & 192 & $1.41\times10^{-6}$ & $2.62\times10^{-6}$ & 21 & 170 & $1.56\times10^{-13}$ & $1.43\times10^{-6}$ \\
            \bottomrule
        \end{tabular*}
        \begin{tablenotes}[flushleft]
            \PaperTableNotesStyle
            \item[] $M$ and $L$ are the maximum retained azimuthal order and spherical-harmonic degree, respectively. These cases use the reduced local layout of the EMFA SO(2) convolution.
            \item[] Product-grid rules are $(R_{\phi},R_{\theta})$, and \#\,pts is the total number of grid points $R_{\phi}R_{\theta}$. For Lebedev quadrature, $p$ is the algebraic order of accuracy.
            \item[] fp64 and fp32 denote 64- and 32-bit floating-point arithmetic, respectively; errors are maximum absolute deviations between the two equivariance paths.
        \end{tablenotes}
    \end{threeparttable}
\end{table}

\subsection{Completeness of the Wigner bilinear network}
\label{sec:sm-wbn-complete}

This subsection proves the completeness statement invoked in main-text Section~\ref{sec:s2-method}: networks stacking Wigner bilinear blocks with the \emph{full} frame range $|k|\le l$, interleaved with equivariant linear maps, approximate every continuous SO(3)-equivariant map between band-limited feature spaces, and realize polynomial equivariant maps exactly.
The block analyzed here is the exact continuous model underlying the deployed FFN: it retains all frames, forms products over all pairs of left and right channels, and evaluates the Haar integral exactly.
The deployed block (main-text Eqs.~\eqref{eq:dpa4-wbb-lift}--\eqref{eq:dpa4-ffn}) is its economical specialization, with channel-diagonal products, frames truncated to $|k|\le k_{\max}$, and an exact quadrature in place of the integral; the role of the truncation is discussed in Remark~\ref{rem:sm-wbn-truncation}.

Throughout, $g$ denotes a rotation, $dg$ the normalized Haar measure on $SO(3)$, and $D^{(l)}_{mk}(g)$ the \emph{complex} Wigner matrices, which satisfy the orthogonality relation
\begin{equation}
    (2l+1)\int_{SO(3)}\overline{D^{(l)}_{mk}(g)}\,D^{(l')}_{m'k'}(g)\,dg
    =\delta_{ll'}\delta_{mm'}\delta_{kk'}
    \label{eq:sm-wigner-orth}
\end{equation}
and the product identity~\cite{varshalovich1988quantum}
\begin{equation}
    D^{(l_1)}_{m_1k_1}(g)\,D^{(l_2)}_{m_2k_2}(g) =\sum_{L=|l_1-l_2|}^{l_1+l_2} \langle l_1 m_1\, l_2 m_2|L M\rangle\, \langle l_1 k_1\, l_2 k_2|L \kappa\rangle\, D^{(L)}_{M\kappa}(g), \label{eq:sm-wigner-product} \end{equation} with $M=m_1+m_2$, $\kappa=k_1+k_2$, and $\langle\cdot\,\cdot|\cdot\rangle$ the Clebsch--Gordan coefficients.
The associated coupling map $C^{L}_{l_1l_2}:V_{l_1}\otimes V_{l_2}\to V_L$,
\begin{equation}
    \bigl(C^{L}_{l_1l_2}(x,y)\bigr)_M
    =\sum_{m_1,m_2}\langle l_1 m_1\, l_2 m_2|L M\rangle\,x_{m_1}y_{m_2},
    \label{eq:sm-cg-map}
\end{equation}
is a nonzero SO(3)-intertwiner for every admissible triple $|l_1-l_2|\le L\le l_1+l_2$.
The deployed network uses the real Wigner basis of Section~\ref{sec:sm-representation}; the real and complex conventions are related by an invertible equivariant change of coordinates, which is absorbed into the learnable channel maps, so it suffices to argue in the complex convention.

\paragraph{The full-frame block.}
Let $\mathcal{H}=V_{\le L_{\max}}\otimes\mathbb{C}^{C}$.
A \emph{full-frame Wigner bilinear block} $\mathcal{B}:\mathcal{H}\to\mathcal{H}$ maps a feature $x^{(l)}_{m,c}$ to
\begin{align}
    u^{\mathrm{L},(l)}_{m,k,a}
     & =\sum_{c}
    A^{\mathrm{L},(l)}_{k,a,c}\,x^{(l)}_{m,c}, \qquad u^{\mathrm{R},(l)}_{m,k,b} =\sum_{c}A^{\mathrm{R},(l)}_{k,b,c}\,x^{(l)}_{m,c}, \qquad |k|\le l, \label{eq:sm-wbb-lin} \\ f^{\mathrm{L}}_{a}(g) &=\sum_{l,m,k}u^{\mathrm{L},(l)}_{m,k,a}D^{(l)}_{mk}(g), \qquad f^{\mathrm{R}}_{b}(g) =\sum_{l,m,k}u^{\mathrm{R},(l)}_{m,k,b}D^{(l)}_{mk}(g), \label{eq:sm-wbb-synth}\\ h_{ab}(g) &=f^{\mathrm{L}}_{a}(g)\,f^{\mathrm{R}}_{b}(g), \label{eq:sm-wbb-prod}\\ z^{(l)}_{m,k,ab} &=(2l+1)\int_{SO(3)}\overline{D^{(l)}_{mk}(g)}\,h_{ab}(g)\,dg, \label{eq:sm-wbb-analysis}\\ \mathcal{B}(x)^{(l)}_{m,d} &=\sum_{c}P^{(l)}_{dc}\,x^{(l)}_{m,c} +\delta_{l0}\delta_{m0}\,q_{d} +\sum_{k,a,b}B^{(l)}_{d,k,ab}\,z^{(l)}_{m,k,ab}, \label{eq:sm-wbb-out}\end{align} with the per-degree matrices $A^{\mathrm{L}},A^{\mathrm{R}},P$, the bias $q$, and the mixing coefficients $B$ learnable.
A \emph{Wigner bilinear network} (WBN) is a composition $L_{\mathrm{read}}\circ\mathcal{B}_{K}\circ\cdots\circ\mathcal{B}_{1}\circ L_{\mathrm{emb}}$ with equivariant linear maps $L_{\mathrm{emb}}:\mathcal{V}\to\mathcal{H}$ and $L_{\mathrm{read}}:\mathcal{H}\to\mathcal{W}$.
Each block is equivariant: a rotation translates the argument of the synthesized fields, the point-wise product commutes with this translation, and the Haar projection~\eqref{eq:sm-wbb-analysis} intertwines it back to the coefficient action.

\begin{lemma}[Clebsch--Gordan trees span polynomial equivariants]
    \label{lem:sm-cg-span}
    Let $\mathcal{V},\mathcal{W}$ be finite direct sums of SO(3) irreps.
    Every polynomial SO(3)-equivariant map $P:\mathcal{V}\to\mathcal{W}$ is a finite linear combination of \emph{tree maps}: compositions of equivariant linear leaf maps $\mathcal{V}\to V_{l}$, iterated binary couplings~\eqref{eq:sm-cg-map}, and a final equivariant linear map into $\mathcal{W}$, together with constant terms in the trivial isotypic part of $\mathcal{W}$.
\end{lemma}

\begin{proof}
    Decomposing $P$ into homogeneous parts and comparing degrees in $P(tx)$ shows each part is equivariant; the degree-zero part is an invariant constant.
    A homogeneous part of degree $d$ is, by polarization, the diagonal restriction of a symmetric $d$-linear equivariant map, hence of an element of $\mathrm{Hom}_{SO(3)}(\mathcal{V}^{\otimes d},\mathcal{W})$; it therefore suffices to span this space by trees, and, after decomposing $\mathcal{V}$ and $\mathcal{W}$ into irreps, to span $\mathrm{Hom}_{SO(3)}(V_{l_1}\otimes\cdots\otimes V_{l_d},V_L)$.
    We induct on $d$.
    For $d=1$ Schur's lemma leaves only multiples of the identity.
    For $d>1$, decompose $X=V_{l_1}\otimes\cdots\otimes V_{l_{d-1}}$ into irreducible summands with equivariant projections $p_{J,r}:X\to V_J$ and inclusions $i_{J,r}$ (the equivariant maps extracting and identifying the $r$-th copy of $V_J$ in the decomposition of $X$ into irreducibles, with $\sum_{J,r}i_{J,r}p_{J,r}=\mathrm{id}_X$); any $T\in\mathrm{Hom}_{SO(3)}(X\otimes V_{l_d},V_L)$ can be written as $T=\sum_{J,r}T\circ(i_{J,r}p_{J,r}\otimes\mathrm{id})$, and each factor $T\circ(i_{J,r}\otimes\mathrm{id})$ lies in $\mathrm{Hom}_{SO(3)}(V_J\otimes V_{l_d},V_L)$, which by the multiplicity-free Clebsch--Gordan decomposition is spanned by $C^{L}_{Jl_d}$ when the triple is admissible and vanishes otherwise.
    The induction hypothesis expresses each $p_{J,r}$ by trees on the first $d-1$ factors; composing with $C^{L}_{Jl_d}$ yields trees on all $d$ factors.
    Evaluating on the diagonal $(x,\ldots,x)$ recovers the homogeneous polynomial parts and proves the claim.
\end{proof}

\begin{theorem}[Completeness of the WBN]
    \label{thm:sm-wbn-complete}
    Let $\mathcal{V},\mathcal{W}$ be finite direct sums of SO(3) irreps and let $F:\mathcal{V}\to\mathcal{W}$ be continuous and SO(3)-equivariant.
    For every compact $\Omega\subset\mathcal{V}$ and every $\varepsilon>0$ there exist a band limit $L_{\max}$, a width $C$, a depth $K$, and WBN parameters such that
    \begin{equation}
        \sup_{x\in\Omega}
        \bigl\|L_{\mathrm{read}}\circ\mathcal{B}_{K}\circ\cdots\circ\mathcal{B}_{1}\circ L_{\mathrm{emb}}(x)-F(x)\bigr\|
        <\varepsilon .
        \label{eq:sm-wbn-uat}
    \end{equation}
    Moreover, every \emph{polynomial} SO(3)-equivariant map is realized exactly, with zero error.
\end{theorem}

\begin{proof}
    \emph{Step 1 (reduction to polynomial equivariants).}
    Let $\widehat\Omega=SO(3)\,\Omega$, compact by compactness of the group.
    By the Stone--Weierstrass theorem there is a polynomial map $R:\mathcal{V}\to\mathcal{W}$ with $\sup_{\widehat\Omega}\|R-F\|<\delta$.
    Its Reynolds average $P(x)=\int_{SO(3)}g^{-1}R(gx)\,dg$ is polynomial and equivariant, and since $F(x)=\int g^{-1}F(gx)\,dg$, \[ \sup_{x\in\Omega}\|P(x)-F(x)\| \le \Bigl(\sup_{g}\|g^{-1}\|_{\mathcal{W}}\Bigr)\,\delta , \] which is smaller than $\varepsilon$ for $\delta$ small.
    It therefore suffices to realize polynomial equivariants exactly.

    \emph{Step 2 (one block realizes one coupling).}
    Suppose two hidden channels of the current feature hold covariants $T_1(x)\in V_{l_1}$ and $T_2(x)\in V_{l_2}$, and fix an admissible output degree $L$.
    Since the coupling map~\eqref{eq:sm-cg-map} is nonzero, its coefficient array is not identically zero, so there exist frames $|k_1|\le l_1$, $|k_2|\le l_2$ with
    \begin{equation}
        \langle l_1 k_1\, l_2 k_2|L\,\kappa\rangle\neq0,
        \qquad \kappa=k_1+k_2 .
        \label{eq:sm-frame-choice}
    \end{equation}
    Use the maps~\eqref{eq:sm-wbb-lin} to load the two channels into one product pair, \[ f^{\mathrm{L}}_{a}(g)=\sum_{m_1}(T_1)_{m_1}D^{(l_1)}_{m_1k_1}(g), \qquad f^{\mathrm{R}}_{b}(g)=\sum_{m_2}(T_2)_{m_2}D^{(l_2)}_{m_2k_2}(g), \] all other coefficients set to zero.
    By the product identity~\eqref{eq:sm-wigner-product}, \[ h_{ab}(g) =\sum_{L'}\langle l_1 k_1\, l_2 k_2|L'\kappa\rangle\, \bigl(C^{L'}_{l_1l_2}(T_1,T_2)\bigr)_{M}\,D^{(L')}_{M\kappa}(g), \] and the analysis~\eqref{eq:sm-wbb-analysis} with the orthogonality~\eqref{eq:sm-wigner-orth} extracts \[ z^{(L)}_{M,\kappa,ab} =\langle l_1 k_1\, l_2 k_2|L\,\kappa\rangle\, \bigl(C^{L}_{l_1l_2}(T_1,T_2)\bigr)_{M}.
    \]
    The nonzero scalar~\eqref{eq:sm-frame-choice} is absorbed into $B^{(L)}_{d,\kappa,ab}$ in~\eqref{eq:sm-wbb-out}; choosing $P^{(l)}$ to preserve the existing channels and directing the result into one fresh output channel $d$ stores exactly the coupling $C^{L}_{l_1l_2}(T_1,T_2)$.

    \emph{Step 3 (finite stacks realize all trees).}
    Fix a polynomial equivariant $P$ and expand it by Lemma~\ref{lem:sm-cg-span} into finitely many tree maps.
    Assign one hidden channel to every node of every tree; choose $L_{\max}$ no smaller than the largest degree appearing and $C$ large enough to hold all channels.
    The embedding $L_{\mathrm{emb}}$ places the leaf maps.
    Proceeding by induction on tree depth, each block writes all couplings of the next depth into fresh channels while preserving the channels already filled; this is Step~2, applied with a separate product pair $(a,b)$ and output channel per node, the mixing coefficients $B$ set to ignore all other pairs.
    After finitely many blocks all nodes are present; the readout $L_{\mathrm{read}}$ forms the finite linear combination of tree outputs, and the bias $q$ supplies the constant terms.
    The stack realizes $P$ exactly, and with Step~1 the theorem follows.
\end{proof}

\begin{remark}[The adapters $L_{\mathrm{emb}},L_{\mathrm{read}}$]
\label{rem:sm-wbn-adapters}
The maps $L_{\mathrm{emb}}$ and $L_{\mathrm{read}}$ serve only to connect the blocks to input and output spaces $\mathcal{V},\mathcal{W}$ that may differ from the hidden space $\mathcal{H}$, which Theorem~\ref{thm:sm-wbn-complete} is free to enlarge.
When the three spaces coincide, as for the deployed FFN, whose blocks map the feature space $V_{\le L}\otimes\mathbb{R}^{C}$ to itself, no adapters are needed: $L_{\mathrm{emb}}$ composes with the linear maps $A^{\mathrm{L}},A^{\mathrm{R}},P$ of the first block and $L_{\mathrm{read}}$ with $P$, $B$ and $q$ of the last, each composition again an equivariant linear map of the same form.
A bare stack of Wigner bilinear blocks on a common feature space is therefore itself a WBN.
\end{remark}

\begin{remark}[Frame truncation and sphere projection]
    \label{rem:sm-wbn-truncation}
    The full frame range enters the proof only through the choice~\eqref{eq:sm-frame-choice}.
    Truncating the frames to $|k|\le k_{\max}$ restricts that choice, and at $k_{\max}=0$ it forces $k_1=k_2=\kappa=0$: the surviving coefficient $\langle l_1 0\, l_2 0|L 0\rangle$ vanishes whenever $l_1+l_2+L$ is odd, so a sphere-restricted bilinear nonlinearity cannot realize any parity-odd coupling, however deep the stack; this is the incompleteness discussed in main-text Section~\ref{sec:s2-method}.
    Already at $k_{\max}=1$ nonzero parity-odd choices become available, e.g.\ $\langle 1\,1\,1\,{-1}|1\,0\rangle=1/\sqrt{2}$ for the cross product $1\otimes1\to1$; the hypothesis of Theorem~\ref{thm:sm-wbn-complete}, however, requires the full range.
\end{remark}

\FloatBarrier
\section{Training methods}
\label{sec:sm-training-systems}

\subsection{Training objective}
\label{sec:sm-training-objective}

DPA4 is trained as a conservative interatomic potential with the MAE objective of main-text Eq.~\eqref{eq:dpa4-loss}; the benchmark-specific weights, batch sizes and training lengths are listed in Tables~\ref{tab:sm_main_ablation_config}--\ref{tab:sm_matbench_config}.

\subsection{Warmup--stable--decay learning-rate schedule}
\label{sec:sm-wsd-schedule}

For long training runs, DPA4 uses a warmup--stable--decay schedule in which the learning rate first increases linearly, remains constant for the main training phase and is annealed only near the end of the run~\cite{wen2024wsd}.
Let $T$ be the total number of optimization steps, $T_{\mathrm{w}}$ the warmup length, $T_{\mathrm{d}}$ the decay length and $T_{\mathrm{s}}=T-T_{\mathrm{w}}-T_{\mathrm{d}}$ the stable length.
The
schedule used in the reported DPA4 training runs is
\begin{equation}
    \alpha(t)=
    \begin{cases}
        \alpha_{\mathrm{w}}
        +(\alpha_0-\alpha_{\mathrm{w}})t/T_{\mathrm{w}},
         & 0\le t<T_{\mathrm{w}},                             \\[3pt]
        \alpha_0,
         & T_{\mathrm{w}}\le t<T_{\mathrm{w}}+T_{\mathrm{s}}, \\[3pt]
        \alpha_{\mathrm{d}}(\tau),
         & T_{\mathrm{w}}+T_{\mathrm{s}}\le t<T,              \\[3pt]
        \alpha_{\min},
         & t\ge T,
    \end{cases}
    \label{eq:sm-wsd-piecewise}
\end{equation}
where $\alpha_{\mathrm{w}}$ is the initial warmup learning rate,
$\alpha_0$ is the stable-phase learning rate,
$\alpha_{\min}$ is the final learning rate and
\begin{equation}
    \tau
    =
    \operatorname{clip}
    \left(
    \frac{t-T_{\mathrm{w}}-T_{\mathrm{s}}}{T_{\mathrm{d}}},
    0,1
    \right),
    \label{eq:sm-wsd-tau}
\end{equation}
with cosine annealing in the decay phase,
\begin{equation}
    \alpha_{\mathrm{d}}(\tau)
    =
    \alpha_{\min}
    +
    \frac{\alpha_0-\alpha_{\min}}{2}
    \left(1+\cos \pi\tau\right).
    \label{eq:sm-wsd-cosine}
\end{equation}
Thus the stable phase carries most of the optimization steps, whereas the final cosine decay suppresses high-learning-rate oscillations before checkpoint selection.

\subsection{HybridMuon optimizer}
\label{sec:sm-hybrid-muon}

DPA4 is optimized with HybridMuon, a matrix-aware hybrid optimizer adapted from Muon~\cite{jordan2024muon} and scalable-Muon training studies~\cite{liu2025muon}.
The design separates two classes of trainable parameters.
Matrix-valued hidden transformations are routed to Muon, whereas biases, normalization scales, one-dimensional parameters and explicitly marked auxiliary parameters are routed to Adam or decoupled-weight-decay Adam.
This static routing is built on the first optimizer step and remains fixed during training, so the update rule for each parameter is independent of the current gradient value.

\paragraph{Matrix views and slice mode.}
The effective shape of a trainable tensor is its shape with singleton dimensions removed.
HybridMuon interprets a rank-two effective shape as one matrix.
For higher-rank equivariant tensors, DPA4 uses slice mode: the leading dimensions index independent blocks, and Muon is applied separately to every trailing $(m,n)$ matrix.
Thus, for an effective shape $(b_1,\ldots,b_q,m,n)$, the
optimizer constructs
\begin{equation}
    B=\prod_{a=1}^{q}b_a \quad\text{independent matrix blocks}\quad G_t^{(b)}\in\mathbb{R}^{m\times n}, \qquad b=1,\ldots,B.
    \label{eq:sm-muon-slice-view}
\end{equation}
This choice is important for SO(2)-structured equivariant weights: degree- and order-indexed blocks are updated independently, rather than being flattened into one large matrix that would mix unrelated representation strata.
In the degree-wise SO(3) channel maps used by the SO(2) pre- and post-projections and by the equivariant FFN, the weight tensor has shape $(L+1,C_{\mathrm{in}},F C_{\mathrm{out}})$; slice mode therefore applies Muon separately to the $(C_{\mathrm{in}},F C_{\mathrm{out}})$ matrix of each degree $l$.
Local SO(2) linear maps keep separate matrix parameters for the $m=0$ block and the constrained $|m|>0$ blocks, so the optimizer acts on these structured matrix blocks without collapsing distinct equivariant subspaces.

\paragraph{Muon update.}
For each Muon-routed block, the optimizer maintains a momentum buffer and forms
a Nesterov-style update
\begin{align}
    M_t^{(b)} & = \beta M_{t-1}^{(b)}+(1-\beta)G_t^{(b)}, \\ U_t^{(b)} & = \beta M_t^{(b)}+(1-\beta)G_t^{(b)}.
    \label{eq:sm-muon-nesterov}
\end{align}
The matrix $U_t^{(b)}$ is then orthogonalized by a Newton--Schulz polar iteration.
For the standard square or batched path, the iteration starts from
$X_0=U_t^{(b)}/\|U_t^{(b)}\|_{\mathrm{F}}$ and applies
\begin{equation}
    X_{k+1} = a_k X_k + \left(b_k A_k+c_k A_k^2\right)X_k, \qquad A_k=X_kX_k^{\mathsf{T}}.
    \label{eq:sm-newton-schulz}
\end{equation}
DPA4 uses a two-stage schedule: eight fast iterations with $(a,b,c)=(3.4445,-4.7750,2.0315)$ followed by two Newton polishing iterations with $(a,b,c)=(2,-1.5,0.5)$.
The resulting polar factor is denoted $Q_t^{(b)}$.

\paragraph{Rectangular Gram path.}
For rectangular matrices, HybridMuon uses a compiled Gram Newton--Schulz path following the fast polar-decomposition formulation used for Muon~\cite{dao2026gramnewtonschulz}.
The matrix is oriented so that $m\le n$, normalized in single precision and iterated in half precision using a fixed Polar-Express coefficient schedule.
Since the iteration depends on $XX^{\mathsf{T}}\in\mathbb{R}^{m\times m}$, rectangular blocks with the same smaller dimension can be column-padded, concatenated and orthogonalized in a single grouped call.
Padding only the larger dimension preserves the Frobenius
norm and the Gram matrix,
\begin{equation}
    [X\;0][X\;0]^{\mathsf{T}}=XX^{\mathsf{T}},
    \label{eq:sm-gram-padding}
\end{equation}
so truncating the padded columns after the iteration exactly recovers the
unpadded result while reducing the number of small GPU launches.

\paragraph{Update-RMS matching and Adam-family path.}
In the default match-RMS mode, the Muon update for an $m\times n$ block is
scaled as
\begin{equation}
    \Delta W_t^{(b)}
    =
    -\alpha_t\,\gamma\,\sqrt{\max(m,n)}\,Q_t^{(b)},
    \qquad \gamma=0.18.
    \label{eq:sm-match-rms}
\end{equation}
This coefficient follows the update-RMS calibration used in scalable Muon training~\cite{liu2025muon}: it brings the per-element magnitude of the orthogonalized matrix update onto the same learning-rate scale as AdamW-like updates.
Parameters routed to Adam use first and second moments in single precision, with bias correction and $\epsilon=10^{-20}$.
Decoupled weight decay is applied to Muon-routed matrices and to decay-enabled Adam-routed tensors, but not to one-dimensional Adam-routed parameters.

\subsection{Magma-lite update damping}
\label{sec:sm-magma-lite}

Conservative energy-gradient training produces gradients that include mixed coordinate--parameter derivatives and can exhibit large block-to-block variation.
DPA4 therefore augments the Muon path with a deterministic Magma-lite damping rule, adapted from momentum-aligned gradient masking~\cite{joo2026magma}.
For each Muon block, let $G_t^{(b)}$ be the current gradient and $M_t^{(b)}$ the momentum buffer after the update in Eq.~\eqref{eq:sm-muon-nesterov}.
The block alignment is
\begin{equation}
    \chi_t^{(b)}
    =
    \frac{
        \left\langle M_t^{(b)},G_t^{(b)}\right\rangle_{\mathrm{F}}
    }{
        \|M_t^{(b)}\|_{\mathrm{F}}\,
        \|G_t^{(b)}\|_{\mathrm{F}}+\varepsilon
    },
    \qquad
    \chi_t^{(b)}\in[-1,1].
    \label{eq:sm-magma-cosine}
\end{equation}
The alignment is mapped through a temperature-scaled sigmoid and stretched to
$[0,1]$,
\begin{equation}
    r_t^{(b)}
    =
    \operatorname{clip}
    \left[
        \frac{
            \sigma(\chi_t^{(b)}/\tau)-\sigma(-1/\tau)
        }{
            \sigma(1/\tau)-\sigma(-1/\tau)
        },
        0,1
        \right],
    \qquad \tau=2.
    \label{eq:sm-magma-score}
\end{equation}
An exponential moving average gives
\begin{equation}
    u_t^{(b)}
    =
    \rho_{\mathrm{ema}} u_{t-1}^{(b)}
    +(1-\rho_{\mathrm{ema}})r_t^{(b)},
    \qquad \rho_{\mathrm{ema}}=0.9,
    \label{eq:sm-magma-ema}
\end{equation}
and the final Muon update is rescaled by
\begin{equation}
    \Delta W_t^{(b)}
    \leftarrow
    \left[s_{\min}+(1-s_{\min})u_t^{(b)}\right]\Delta W_t^{(b)},
    \qquad s_{\min}=0.1.
    \label{eq:sm-magma-lite}
\end{equation}
This rule differs from stochastic masking: all blocks remain active, and poorly aligned updates are continuously damped rather than randomly skipped.
The nonzero lower bound prevents the optimizer from freezing a block entirely, which is important for MLIP training, where force labels can make the alignment temporarily noisy without implying that the corresponding representation block should stop learning.

\FloatBarrier
\section{Ablation study}
\label{sec:sm-ablation}

This section provides the complete ablation evidence behind Section~\ref{sec:ablation}.
The first subsection reports the cumulative ablation behind main-text Fig.~\ref{fig:ablation}, and the following six give the full mechanism-level sweeps for graph compilation, attention aggregation, multi-focus design, the low-rank edge--node SO(2)-equivariant degree coupling, the spherical-grid nonlinearity and the Wigner bilinear network.
The remaining subsections report model-selection and robustness studies.
All ablations use the same WBM-subsampled evaluation protocol.
Relative efficiency metrics are normalized within each controlled group under matched H20 hardware, batch-size, data-loading and precision settings.
Unless otherwise noted, \textit{Train time (rel.)
} denotes relative training wall-clock time, and
\textit{Test time (rel.)} denotes the wall-clock time of a full test pass over the
WBM-subsampled test set in DeePMD-kit.
Boldface and underlining denote the best and second-best values only for metrics in which ranking is explicitly highlighted.
In the ablation tables, \cmark{} and \xmark{} denote an enabled and a disabled setting, respectively.
N/A denotes a parameter or option that is not applicable to the corresponding setting.

\subsection{Cumulative ablation}
\label{sec:sm-cumulative-ablation}

This ablation is the numerical basis of main-text Fig.~\ref{fig:ablation}.
Starting from the complete interaction block, the four components that can be switched off without changing the dimensions of the equivariant representation are removed one after another, and every configuration is retrained under the identical protocol given in the cumulative-ablation column of Table~\ref{tab:sm_main_ablation_config}; the resulting errors are collected in Table~\ref{tab:sm_cumulative_ablation}.
All five configurations use a single focus stream, so the multi-focus design contributes nothing to the reported differences; it is assessed separately at matched SO(2) width in Section~\ref{sec:sm-multifocus}.

Removing all four components raises the energy MAE from 31.007 to 39.671~meV/atom and the force MAE from 39.169 to 44.581~meV/\AA{}, increases of 27.9\% and 13.8\% relative to the complete architecture.
It reduces the parameter count from 2.05~M to 1.19~M and shortens the wall-clock time of a full test pass from 20.2 to 17.7~s, a 12.4\% reduction on the same denominator.
The energy error increases monotonically along the chain, whereas the force error does not: removing the message node WBN leaves it marginally lower, so the two WBN insertions are separated by the force metric only in the controlled comparison of Table~\ref{tab:sm_wbn_ablation}.

\begin{table}[ht]
    \PaperTableStyle
    \setlength{\tabcolsep}{1.8pt}
    \begin{threeparttable}
        \caption{Cumulative ablation of the DPA4 architectural components.}
        \label{tab:sm_cumulative_ablation}
        \begin{tabular*}{\linewidth}{@{\extracolsep{\fill}}ccccccccc@{}}
            \toprule
            \makecell{\textbf{Message}\\\textbf{node WBN}}       &
            \makecell{\textbf{WBN in}\\\textbf{FFN}}             &
            \makecell{\textbf{Edge--node}\\\textbf{degree coupling}}     &
            \makecell{\textbf{Attention}\\\textbf{aggregation}}  &
            \textbf{E MAE}$\downarrow$\tnote{a}                  &
            \textbf{F MAE}$\downarrow$\tnote{a}                  &
            \makecell{\textbf{Train time}\\\textbf{(rel.)}$\downarrow$}\tnote{b} &
            \makecell{\textbf{Test time}\\\textbf{(rel.)}$\downarrow$}\tnote{b}  &
            \textbf{Params}                                                                        \\
            \midrule
            \cmark & \cmark & \cmark & \cmark & \textbf{31.007}    & \underline{39.169} & 1.00 & 1.00 & 2.05M \\
            \xmark & \cmark & \cmark & \cmark & \underline{31.585} & \textbf{38.957}    & 0.93 & 0.98 & 1.89M \\
            \xmark & \xmark & \cmark & \cmark & 33.226             & 39.918             & 0.87 & 0.93 & 1.22M \\
            \xmark & \xmark & \xmark & \cmark & 34.529             & 42.004             & 0.80 & 0.90 & 1.21M \\
            \xmark & \xmark & \xmark & \xmark & 39.671             & 44.581             & 0.72 & 0.88 & 1.19M \\
            \bottomrule
        \end{tabular*}
        \begin{tablenotes}[flushleft]
            \PaperTableNotesStyle
            \item[] Check and cross marks denote enabled and disabled components, respectively. A disabled edge--node degree coupling denotes scalar scaling, and disabled attention aggregation denotes scatter-sum aggregation. All rows use one focus stream.
            \item[] Boldface and underlining denote the lowest and second-lowest value in each MAE column.
            \item[a] E MAE and F MAE are mean absolute errors in meV/atom and meV/\AA, respectively; lower values are better.
            \item[b] Train time (rel.) and Test time (rel.) are normalized to the complete architecture in the first row.
        \end{tablenotes}
    \end{threeparttable}
\end{table}

\FloatBarrier

\subsection{Graph compilation and training precision}
\label{sec:sm-compile-ablation}

This ablation quantifies the systems-level benefit of graph compilation and reduced-precision tensor-core execution (Table~\ref{tab:dpa4_compile_ablation}).
Relative to the non-compiled FP32 baseline, bf16 AMP alone gives a 1.43$\times$ training speedup and reduces peak training memory by 59\%.
Graph compilation provides a larger and complementary gain: compiled FP32 training gives a 1.61$\times$ speedup, and enabling TF32 in this compiled path increases the speedup to 1.82$\times$.
The largest practical improvement is obtained by combining compilation with bf16 AMP, which gives a 3.1$\times$ speedup and reduces peak memory by 60\%.
Thus, the compiled mixed-precision path makes conservative energy-gradient training more than three times faster while using only about 40\% of the baseline peak GPU memory.

The corresponding accuracy changes are small compared with the efficiency gain.
Against the FP32 baseline (27.603~meV/atom and 34.246~meV/\AA{}), bf16 AMP changes the energy and force MAEs by 1.7\% and 2.0\%, respectively.
The compiled-bf16 setting, which gives the strongest speed and memory improvement, changes them by 2.7\% and 1.8\%; when TF32 is also enabled, the energy MAE changes by 2.8\% and the force MAE by only 0.1\%.
These differences are consistent with small numerical and stochastic variation rather than systematic accuracy degradation.
Although the random seed is fixed and \texttt{torch.compile} preserves the mathematical computation graph, compilation can change kernel fusion, hardware-specific kernel selection, reduction ordering and tensor-core dispatch; bf16 AMP and TF32 also change the effective arithmetic used by eligible matrix operations.
DPA4 limits this sensitivity by keeping geometric preprocessing and normalization operations such as RMSNorm in FP32, while allowing large matrix operations to use bf16 or TF32 where appropriate.
We therefore enable \texttt{torch.compile}, bf16 AMP and TF32 in the following ablations and benchmark experiments unless otherwise specified.

\begin{table}[ht]
    \PaperTableStyle
    \begin{threeparttable}
        \caption{Ablation of graph compilation and training precision.}
        \label{tab:dpa4_compile_ablation}
        \begin{tabular*}{\linewidth}{@{\extracolsep{\fill}}ccccccc@{}}
            \toprule
            \textbf{Compile}                       &
            \textbf{BF16 AMP}                      &
            \textbf{TF32}                          &
            \textbf{E MAE}$\downarrow$\tnote{a}    &
            \textbf{F MAE}$\downarrow$\tnote{a}    &
            \textbf{Train time (rel.)}$\downarrow$\tnote{b} &
            \makecell{\textbf{Peak train}\\\textbf{mem. (rel.)}$\downarrow$}\tnote{b}                              \\
            \midrule
            \xmark                                  & \xmark & N/A   & 27.603 & 34.246 & 1.00          & 1.00          \\
            \xmark                                  & \cmark  & N/A   & 28.081 & 34.936 & 0.70          & 0.41          \\
            \cmark                                   & \xmark & \xmark & 28.217 & 34.787 & 0.62          & 0.77          \\
            \cmark                                   & \xmark & \cmark  & 28.190 & 34.860 & 0.55          & 1.00          \\
            \cmark                                   & \cmark  & \xmark & 28.355 & 34.864 & \textbf{0.32} & \textbf{0.40} \\
            \cmark                                   & \cmark  & \cmark  & 28.364 & 34.275 & \textbf{0.32} & \textbf{0.40} \\
            \bottomrule
        \end{tabular*}
        \begin{tablenotes}[flushleft]
            \PaperTableNotesStyle
            \item[] Check and cross marks denote enabled and disabled settings, respectively. N/A indicates that TF32 tensor cores are not applicable to the non-compiled execution path.
            \item[] Boldface denotes the minimum value in each efficiency column.
            \item[a] E MAE and F MAE are mean absolute errors in meV/atom and meV/\AA, respectively; lower values are better.
            \item[b] Train time (rel.) and peak train memory are normalized to the non-compiled FP32 baseline measured on NVIDIA H20 hardware. Peak train memory is the peak GPU memory allocated during training.
        \end{tablenotes}
    \end{threeparttable}
\end{table}

\FloatBarrier
\subsection{Attention versus sum aggregation}
\label{sec:sm-attention-aggregation}

The attention ablation isolates the aggregation rule while holding the feature dimension and focus count fixed within each pair of rows (Table~\ref{tab:dpa4_attention_ablation}).
Replacing scatter-sum aggregation with attention-weighted sum consistently reduces both energy and force MAEs across the 64-channel, 96-channel and 96-channel 2-focus settings.
The improvement is substantial and stable: energy MAE decreases by 8.3--9.3\%, and force MAE decreases by 4.7--6.8\% relative to the corresponding scatter-sum controls.
The gain costs little: training time increases by only 5--6\%, and test time stays within 1--5\% of the control models.
This result supports the design choice of computing attention from rotationally invariant scalar channels: the model gains adaptive neighbor selection while preserving equivariance of the higher-order SO(2) features.

\begin{table}[ht]
    \PaperTableStyle
    \begin{threeparttable}
        \caption{Ablation of attention aggregation.}
        \label{tab:dpa4_attention_ablation}
        \begin{tabular*}{\linewidth}{@{\extracolsep{\fill}}ccccccc@{}}
            \toprule
            \textbf{Feature dim.}                  &
            \textbf{No. focuses}                   &
            \textbf{Aggregation}                   &
            \textbf{E MAE}$\downarrow$\tnote{a}    &
            \textbf{F MAE}$\downarrow$\tnote{a}    &
            \makecell{\textbf{Train time}\\\textbf{(rel.)}}$\downarrow$\tnote{b} &
            \makecell{\textbf{Test time}\\\textbf{(rel.)}}$\downarrow$\tnote{b}                                                            \\
            \midrule
            64                                     & 1 & Scatter sum            & 30.691          & 40.072          & 1.00 & 1.00 \\
            64                                     & 1 & Attention-weighted sum & \textbf{27.839} & \textbf{38.184} & 1.05 & 1.01 \\
            \midrule
            96                                     & 1 & Scatter sum            & 30.068          & 38.407          & 1.00 & 1.00 \\
            96                                     & 1 & Attention-weighted sum & \textbf{27.567} & \textbf{36.127} & 1.05 & 1.05 \\
            \midrule
            96                                     & 2 & Scatter sum            & 30.935          & 36.639          & 1.00 & 1.00 \\
            96                                     & 2 & Attention-weighted sum & \textbf{28.083} & \textbf{34.158} & 1.06 & 1.02 \\
            \bottomrule
        \end{tabular*}
        \begin{tablenotes}[flushleft]
            \PaperTableNotesStyle
            \item[] Boldface denotes the lower error within each matched feature-dimension and focus-count pair.
            \item[a] E MAE and F MAE are mean absolute errors in meV/atom and meV/\AA, respectively; lower values are better.
            \item[b] Within each feature-dimension and focus-count setting, Train time (rel.) and Test time (rel.) are normalized to scatter-sum aggregation.
        \end{tablenotes}
    \end{threeparttable}
\end{table}

\FloatBarrier
\subsection{Multi-focus design}
\label{sec:sm-multifocus}

The multi-focus comparison separates per-focus feature width from the number of parallel equivariant focus channels (Table~\ref{tab:dpa4_focus_ablation}).
Increasing the width of a single focus stream rapidly increases parameter count and inference cost, but the corresponding accuracy gains are uneven.
By contrast, multi-focus variants increase the effective SO(2) convolution dimension through several narrower focus channels and often reach better accuracy--cost trade-offs.
At an SO(2) dimension of 192, the 96-channel 2-focus model gives the best overall result, improving energy MAE from 29.418 to 26.994~meV/atom and force MAE from 39.529 to 36.408~meV/\AA{} relative to the 64-channel 1-focus baseline.
It also outperforms the 192-channel 1-focus model with more than 56\% fewer trainable parameters and approximately 23\% and 34\% lower training and inference times, respectively.
All rows in this sweep use the same learning-rate setting; as larger-capacity variants often benefit from smaller tuned learning rates, the reported MAEs of the largest configurations may be mildly conservative.
Under a shared training recipe, however, the 96-channel 2-focus configuration gives the most balanced point in this sweep.
These trends are consistent with the intended role of focus competition, where parallel equivariant sub-channels specialize to different edge-local geometric motifs before the rotate-back step.

\begin{table}[ht]
    \PaperTableStyle
    \begin{threeparttable}
        \caption{Ablation of SO(2) feature width and focus count.}
        \label{tab:dpa4_focus_ablation}
        \begin{tabular*}{\linewidth}{@{\extracolsep{\fill}}cccccccc@{}}
            \toprule
            \textbf{Feature dim.}                  &
            \textbf{No. focuses}                   &
            \textbf{SO(2) dim.}                    &
            \textbf{E MAE}$\downarrow$\tnote{a}    &
            \textbf{F MAE}$\downarrow$\tnote{a}    &
            \makecell{\textbf{Train time}\\\textbf{(rel.)}}$\downarrow$\tnote{b} &
            \makecell{\textbf{Test time}\\\textbf{(rel.)}}$\downarrow$\tnote{b}  &
            \textbf{Params}                                                                                                  \\
            \midrule
            64                                     & 1 & 64  & 29.418             & 39.529             & 1.00 & 1.00 & 1.9M  \\
            96                                     & 1 & 96  & 28.671             & 38.333             & 1.37 & 1.49 & 4.1M  \\
            64                                     & 2 & 128 & 28.126             & 38.123             & 1.55 & 1.55 & 3.2M  \\
            128                                    & 1 & 128 & 28.708             & 37.882             & 1.80 & 1.86 & 7.3M  \\
            64                                     & 3 & 192 & 27.890             & 37.283             & 1.95 & 2.14 & 4.5M  \\
            96                                     & 2 & 192 & \textbf{26.994}    & \textbf{36.408}    & 2.28 & 2.54 & 7.0M  \\
            192                                    & 1 & 192 & 27.286             & 36.477             & 2.98 & 3.86 & 16.0M \\
            64                                     & 4 & 256 & 27.821             & 37.605             & 2.43 & 2.65 & 5.7M  \\
            256                                    & 1 & 256 & 28.409             & 36.670             & 4.50 & 5.17 & 28.4M \\
            96                                     & 3 & 288 & \underline{27.168} & \underline{36.429} & 3.25 & 3.62 & 9.8M  \\
            \bottomrule
        \end{tabular*}
        \begin{tablenotes}[flushleft]
            \PaperTableNotesStyle
            \item[] Boldface and underlining denote the lowest and second-lowest value in each MAE column. Params is the number of trainable parameters; all rows use the same learning-rate setting.
            \item[a] E MAE and F MAE are mean absolute errors in meV/atom and meV/\AA, respectively; lower values are better.
            \item[b] Train time (rel.) and Test time (rel.) are normalized to the 64-channel, 1-focus baseline.
        \end{tablenotes}
    \end{threeparttable}
\end{table}

\FloatBarrier
\subsection{Low-rank edge--node \texorpdfstring{SO(2)}{SO2}-equivariant degree coupling}
\label{sec:sm-radial-coupling}

This ablation tests the low-rank edge--node SO(2)-equivariant degree coupling (A1), namely how the per-degree edge radial profiles couple the angular degrees of the node-side SO(2) message in the local frame (Table~\ref{tab:dpa4_radial_degree_ablation}).
The scalar-scaling baseline applies each degree-$l$ radial profile only to the degree-$l$ channels, whereas degree mixing builds a cross-degree kernel as a learnable function of the same degree-indexed profiles, mixing input and output angular degrees at fixed $|m|$ before the SO(2) stack.
Making this kernel channel-dependent improves expressivity, but a per-channel dense kernel would inflate the parameter count by a factor of the hidden width $H$.
DPA4 instead parameterizes the kernel as a rank-$R$ factorization across the channel index, with $R$ scalar degree-pair coefficients contracted against a learnable channel basis of width $R\leq H$ (main-text Eq.~\ref{eq:dpa4-K-lowrank}).
Even the rank-$R=1$ form captures most of the benefit: it reduces energy MAE from 28.493 to 27.611~meV/atom and force MAE from 38.349 to 35.689~meV/\AA{} relative to scalar scaling, corresponding to 3.1\% and 6.9\% improvements, respectively (Table~\ref{tab:dpa4_radial_degree_ablation}).
This improvement is obtained at low additional cost, with training and test times increasing only to 1.12$\times$ and 1.10$\times$ the scalar-scaling baseline.
A compact low-rank edge--node degree coupling therefore provides a favorable accuracy--throughput trade-off before the kernel is made more expressive.

The remaining rank sweep shows that a more expressive edge--node degree coupling is not monotonically better.
Increasing the rank from 1 to 4 gives the best energy and force MAEs, but also raises the training cost to 1.86$\times$ the baseline; rank 8, rank 16, and the full kernel are still more expensive yet give worse force MAEs (Table~\ref{tab:dpa4_radial_degree_ablation}).
These results suggest that compact per-channel kernels act as a structural regularizer, providing enough channel-specific angular mixing without introducing many weakly constrained radial--angular interactions.

\begin{table}[ht]
    \PaperTableStyle
    \begin{threeparttable}
        \caption{Ablation of the low-rank edge--node SO(2)-equivariant degree coupling.}
        \label{tab:dpa4_radial_degree_ablation}
        \begin{tabular*}{\linewidth}{@{\extracolsep{\fill}}ccccccc@{}}
            \toprule
            \makecell{\textbf{Edge--node}\\\textbf{degree coupling}}\tnote{a} &
            \makecell{\textbf{Per-channel}\\\textbf{coupling}} &
            \makecell{\textbf{Coupling}\\\textbf{rank}} &
            \textbf{E MAE}$\downarrow$\tnote{b} &
            \textbf{F MAE}$\downarrow$\tnote{b} &
            \makecell{\textbf{Train time}\\\textbf{(rel.)}}$\downarrow$\tnote{c} &
            \makecell{\textbf{Test time}\\\textbf{(rel.)}}$\downarrow$\tnote{c} \\
            \midrule
            Scalar scaling & \xmark & N/A & 28.493 & 38.349 & 1.00 & 1.00 \\
            Degree mixing & \xmark & N/A & 28.494 & 37.212 & 1.09 & 1.06 \\
            Degree mixing & \cmark & 1 & 27.611 & \underline{35.689} & 1.12 & 1.10 \\
            Degree mixing & \cmark & 2 & \underline{26.983} & 36.127 & 1.66 & 1.13 \\
            Degree mixing & \cmark & 4 & \textbf{26.556} & \textbf{35.675} & 1.86 & 1.14 \\
            Degree mixing & \cmark & 8 & 27.106 & 36.662 & 2.27 & 1.19 \\
            Degree mixing & \cmark & 16 & 26.984 & 36.401 & 3.09 & 1.34 \\
            Degree mixing & \cmark & Full & 28.011 & 37.400 & 3.36 & 2.22 \\
            \bottomrule
        \end{tabular*}
        \begin{tablenotes}[flushleft]
            \PaperTableNotesStyle
            \item[] Check and cross marks denote enabled and disabled per-channel coupling, respectively; N/A denotes a coupling rank that is not applicable. Boldface and underlining denote the lowest and second-lowest value in each MAE column.
            \item[a] Scalar scaling applies each degree-$l$ radial profile only to the degree-$l$ channels, whereas degree mixing uses the profiles to mix angular degrees in the local SO(2) frame.
            \item[b] E MAE and F MAE are mean absolute errors in meV/atom and meV/\AA, respectively; lower values are better.
            \item[c] Train time (rel.) and Test time (rel.) are normalized to the scalar-scaling baseline.
        \end{tablenotes}
    \end{threeparttable}
\end{table}

\FloatBarrier

\subsection{Spherical-grid nonlinearity and quadrature rule}

The spherical-grid nonlinearity is applied in the equivariant FFN alone, so the remaining design choice is the quadrature rule used to project grid features back to equivariant coefficients (Table~\ref{tab:sm_quadrature_ablation}).
Replacing the latitude--longitude product grid by the Lebedev rule lowers the energy MAE by 0.8\%, from 29.225 to 28.992~meV/atom, and leaves the force MAE and the cost within 1\%: 38.055 against 38.103~meV/\AA{}, with relative training and inference times of 0.98 and 1.01.
The released DPA4 variants therefore use Lebedev quadrature.

\begin{table}[ht]
    \PaperTableStyle
    \begin{threeparttable}
        \caption{Ablation of the quadrature rule for the spherical-grid nonlinearity in the FFN.}
        \label{tab:sm_quadrature_ablation}
        \begin{tabular*}{\linewidth}{@{\extracolsep{\fill}}ccccc@{}}
            \toprule
            \textbf{Quadrature}                    &
            \textbf{E MAE}$\downarrow$\tnote{a}    &
            \textbf{F MAE}$\downarrow$\tnote{a}    &
            \textbf{Train time (rel.)}$\downarrow$\tnote{b} &
            \textbf{Test time (rel.)}$\downarrow$\tnote{b}                                                              \\
            \midrule
            Product & \underline{29.225} & \textbf{38.055}    & 1.00 & 1.00 \\
            Lebedev & \textbf{28.992}    & \underline{38.103} & 0.98 & 1.01 \\
            \bottomrule
        \end{tabular*}
        \begin{tablenotes}[flushleft]
            \PaperTableNotesStyle
            \item[] Both rows apply the nonlinearity only in the equivariant FFN on the $k_{\max}=0$ spherical grid; the SO(2) convolution has no grid nonlinearity. The quadrature rule projects spherical-grid features back to equivariant coefficients.
            \item[] Boldface and underlining denote the lowest and second-lowest value in each MAE column.
            \item[a] E MAE and F MAE are mean absolute errors in meV/atom and meV/\AA, respectively; lower values are better.
            \item[b] Train time (rel.) and Test time (rel.) are normalized to the latitude--longitude product-grid row.
        \end{tablenotes}
    \end{threeparttable}
\end{table}

\FloatBarrier

\subsection{Placement of the Wigner bilinear network}

This ablation isolates the two places where the Wigner bilinear network (WBN) enters an interaction block: the equivariant FFN, whose nonlinearity is either absent, evaluated on the $k_{\max}=0$ spherical grid, or evaluated on the full $k_{\max}=1$ grid on the rotation group, and the post-aggregation message--node path, where a WBN couples the aggregated message with the destination node state (Table~\ref{tab:sm_wbn_ablation}).
Each of the two insertions is assessed with the other held fixed.

With the message node WBN disabled, the FFN nonlinearity produces the largest single effect: adding it on the $k_{\max}=0$ grid lowers the energy MAE by 2.77~meV/atom (7.4\%) and the force MAE by 1.51~meV/\AA{} (3.7\%), and raising $k_{\max}$ to $1$ lowers them by a further 0.96~meV/atom (2.8\%) and 0.63~meV/\AA{} (1.6\%) for a 2\% increase in relative training time.
Enabling the message node WBN helps at either FFN setting.
On top of the $k_{\max}=0$ FFN it lowers the energy MAE by 0.79~meV/atom (2.3\%) and the force MAE by 1.10~meV/\AA{} (2.8\%), and on top of the $k_{\max}=1$ FFN it lowers the force MAE by a further 0.45~meV/\AA{} (1.2\%) while raising the energy MAE by 0.25~meV/atom (0.75\%).

The three configurations that contain a WBN therefore differ by at most 0.75\% in energy MAE and 1.2\% in force MAE, smaller than every effect established above, so this experiment does not separate them.
Enabling both paths attains the lower force MAE of the two configurations that use a WBN in the FFN and is the setting used by the released DPA4 variants.
The cost of the added capacity falls almost entirely on the parameter count rather than on runtime: across the table the parameter count grows from 1.2~M to 2.0~M while the relative test time grows by only 4\%.

\begin{table}[ht]
    \PaperTableStyle
    \begin{threeparttable}
        \caption{Ablation of the Wigner bilinear network.}
        \label{tab:sm_wbn_ablation}
        \begin{tabular*}{\linewidth}{@{\extracolsep{\fill}}ccccccc@{}}
            \toprule
            \makecell{\textbf{FFN}\\\textbf{nonlinearity}}\tnote{a} &
            \makecell{\textbf{Message node}\\\textbf{WBN}}          &
            \textbf{E MAE}$\downarrow$\tnote{b}    &
            \textbf{F MAE}$\downarrow$\tnote{b}    &
            \makecell{\textbf{Train time}\\\textbf{(rel.)}}$\downarrow$\tnote{c} &
            \makecell{\textbf{Test time}\\\textbf{(rel.)}}$\downarrow$\tnote{c}  &
            \textbf{Params}                                                                                  \\
            \midrule
            None            & \xmark & 37.213             & 40.653             & 1.00 & 1.00 & 1.2M \\
            $k_{\max}=0$    & \xmark & 34.448             & 39.142             & 1.14 & 1.01 & 1.4M \\
            $k_{\max}=1$    & \xmark & \textbf{33.485}    & 38.517             & 1.16 & 1.02 & 1.9M \\
            $k_{\max}=0$    & \cmark & \underline{33.662} & \textbf{38.044}    & 1.24 & 1.02 & 1.6M \\
            $k_{\max}=1$    & \cmark & 33.736             & \underline{38.069} & 1.23 & 1.04 & 2.0M \\
            \bottomrule
        \end{tabular*}
        \begin{tablenotes}[flushleft]
            \PaperTableNotesStyle
            \item[] Check and cross marks denote an enabled and disabled message--node WBN, respectively. Boldface and underlining denote the lowest and second-lowest value in each MAE column.
            \item[a] The FFN nonlinearity is evaluated on no grid, the $k_{\max}=0$ spherical grid, or the $k_{\max}=1$ rotation-group grid used by the WBN (main-text Section~\ref{sec:s2-method}).
            \item[b] E MAE and F MAE are mean absolute errors in meV/atom and meV/\AA, respectively; lower values are better.
            \item[c] Train time (rel.) and Test time (rel.) are normalized to the row without an FFN nonlinearity.
        \end{tablenotes}
    \end{threeparttable}
\end{table}

\FloatBarrier

\subsection{Interaction depth and stack depth}

Interaction depth is the strongest determinant of accuracy among the structural hyperparameters (Tables~\ref{tab:sm_layers}--\ref{tab:sm_ffn_stack}).
Both errors fall most rapidly over the first few interaction blocks and then with diminishing returns, while training and inference cost grow approximately linearly with depth (Table~\ref{tab:sm_layers}); the depth of each DPA4 variant is therefore chosen to balance accuracy against cost rather than to minimize the error alone.
Within a block, additional SO(2) sublayers improve accuracy steadily at modest cost ($\sim$44\% training-time increase from 2 to 5 SO(2) sublayers; Table~\ref{tab:sm_so2_stack}), whereas deepening the feed-forward sublayer produces only small and partly non-monotonic changes (Table~\ref{tab:sm_ffn_stack}).
We accordingly retain a single feed-forward sublayer and tune the interaction depth and SO(2)-stack depth per variant.

\begin{table}[ht]
    \PaperTableStyle
    \begin{threeparttable}
        \caption{Ablation of interaction-block depth.}
        \label{tab:sm_layers}
        \begin{tabular*}{\linewidth}{@{\extracolsep{\fill}}ccccc@{}}
        \toprule
        \textbf{No. layers}               &
        \textbf{E MAE}$\downarrow$\tnote{a} &
        \textbf{F MAE}$\downarrow$\tnote{a} &
        \textbf{Train time (rel.)}$\downarrow$\tnote{b} &
        \textbf{Test time (rel.)}$\downarrow$\tnote{b}                  \\
        \midrule
        1                                 & 37.125 & 45.552 & 1.00 & 1.00 \\
        2                                 & 29.816 & 39.968 & 1.23 & 1.36 \\
        3                                 & 28.036 & 37.946 & 1.70 & 1.71 \\
        4                                 & 27.207 & 36.628 & 2.18 & 2.13 \\
        5                                 & 27.419 & 35.853 & 2.55 & 2.66 \\
        6                                 & 27.087 & 35.603 & 2.89 & 3.02 \\
        7                                 & 27.019 & 35.281 & 3.43 & 3.47 \\
        8                                 & 26.549 & 34.628 & 3.80 & 3.84 \\
        9                                 & 26.818 & 34.523 & 4.14 & 4.41 \\
        10                                & 26.488 & 33.752 & 4.59 & 4.89 \\
        11                                & 26.552 & 33.731 & 5.07 & 5.29 \\
        12                                & 26.491 & 33.851 & 5.50 & 5.35 \\
        13                                & 26.488 & 34.282 & 5.91 & 6.20 \\
        14                                & 26.227 & 32.821 & 6.30 & 6.64 \\
        15                                & 26.868 & 33.224 & 6.66 & 7.75 \\
        \bottomrule
        \end{tabular*}
        \begin{tablenotes}[flushleft]
            \PaperTableNotesStyle
            \item[a] E MAE and F MAE are mean absolute errors in meV/atom and meV/\AA, respectively; lower values are better.
            \item[b] Train time (rel.) and Test time (rel.) are normalized to the one-layer model.
        \end{tablenotes}
    \end{threeparttable}
\end{table}

\begin{table}[ht]
    \PaperTableStyle
    \begin{threeparttable}
        \caption{Ablation of SO(2)-stack depth per interaction block.}
        \label{tab:sm_so2_stack}
        \begin{tabular*}{\linewidth}{@{\extracolsep{\fill}}ccccc@{}}
        \toprule
        \textbf{SO(2) layers}             &
        \textbf{E MAE}$\downarrow$\tnote{a} &
        \textbf{F MAE}$\downarrow$\tnote{a} &
        \textbf{Train time (rel.)}$\downarrow$\tnote{b} &
        \textbf{Test time (rel.)}$\downarrow$\tnote{b}                  \\
        \midrule
        2                                 & 30.012 & 41.482 & 1.00 & 1.00 \\
        3                                 & 27.483 & 39.521 & 1.10 & 1.20 \\
        4                                 & 27.407 & 38.822 & 1.28 & 1.26 \\
        5                                 & 27.242 & 38.327 & 1.44 & 1.38 \\
        \bottomrule
        \end{tabular*}
        \begin{tablenotes}[flushleft]
            \PaperTableNotesStyle
            \item[a] E MAE and F MAE are mean absolute errors in meV/atom and meV/\AA, respectively; lower values are better.
            \item[b] Train time (rel.) and Test time (rel.) are normalized to the two-layer SO(2) stack.
        \end{tablenotes}
    \end{threeparttable}
\end{table}

\begin{table}[ht]
    \PaperTableStyle
    \begin{threeparttable}
        \caption{Ablation of FFN-stack depth per interaction block.}
        \label{tab:sm_ffn_stack}
        \begin{tabular*}{\linewidth}{@{\extracolsep{\fill}}ccccc@{}}
        \toprule
        \textbf{FFN layers}              &
        \textbf{E MAE}$\downarrow$\tnote{a} &
        \textbf{F MAE}$\downarrow$\tnote{a} &
        \textbf{Train time (rel.)}$\downarrow$\tnote{b} &
        \textbf{Test time (rel.)}$\downarrow$\tnote{b}                  \\
        \midrule
        1                                & 29.527 & 37.408 & 1.00 & 1.00 \\
        2                                & 29.142 & 36.670 & 1.05 & 1.03 \\
        3                                & 28.889 & 36.627 & 1.13 & 1.05 \\
        4                                & 29.227 & 36.225 & 1.18 & 1.08 \\
        \bottomrule
        \end{tabular*}
        \begin{tablenotes}[flushleft]
            \PaperTableNotesStyle
            \item[a] E MAE and F MAE are mean absolute errors in meV/atom and meV/\AA, respectively; lower values are better.
            \item[b] Train time (rel.) and Test time (rel.) are normalized to the one-layer FFN stack.
        \end{tablenotes}
    \end{threeparttable}
\end{table}

\FloatBarrier
\subsection{Attention-aggregation design variants}

These sweeps vary the attention parameterization around the default single-head design, toggling the value projection, the output projection, the pre-mixing and the number of heads for the 64-channel, 96-channel and 96-channel two-focus configurations (Tables~\ref{tab:sm_attention_64}--\ref{tab:sm_attention_96x2}).
The value projection is a learnable linear map applied to the message before the attention-weighted sum, the output projection is a channel mixing applied to the aggregated equivariant feature, and the pre-mixing is a cross-focus channel mixing applied to the input before computing attention.
Across all three configurations, enabling attention with one or more heads consistently outperforms the no-attention scatter-sum baseline, in agreement with the main attention ablation.
The minimal single-head form, without value, output or pre-mixing projections, attains the lowest energy MAE or comes close to it; the additional projections and extra heads yield no consistent gain while increasing the parameter count.
Single-head attention without further projections is therefore retained as the default.

\begin{table}[ht]
    \PaperTableStyle
    \begin{threeparttable}
        \caption{Ablation of attention aggregation for the 64-channel, 1-focus configuration.}
        \label{tab:sm_attention_64}
        \begin{tabular*}{\linewidth}{@{\extracolsep{\fill}}cccccc@{}}
        \toprule
        \textbf{Attn. heads} &
        \textbf{Attention} &
        \textbf{Attn. value/output proj.}\tnote{a} &
        \textbf{Pre mixing} &
        \textbf{E MAE}$\downarrow$\tnote{b} &
        \textbf{F MAE}$\downarrow$\tnote{b} \\
        \midrule
        0 & \xmark & N/A & N/A & 30.691 & 40.072 \\
        1 & \cmark & \xmark/\xmark & \xmark & \textbf{27.839} & 38.184 \\
        1 & \cmark & \xmark/\cmark & \xmark & 29.214 & 38.250 \\
        1 & \cmark & \cmark/\xmark & \xmark & 28.773 & 37.775 \\
        1 & \cmark & \cmark/\cmark & \xmark & 28.656 & 37.688 \\
        1 & \cmark & \cmark/\cmark & \cmark & 28.985 & 37.242 \\
        2 & \cmark & \xmark/\xmark & \xmark & 28.661 & 37.896 \\
        2 & \cmark & \xmark/\cmark & \xmark & \underline{28.634} & 38.035 \\
        2 & \cmark & \cmark/\xmark & \xmark & 28.958 & \underline{37.056} \\
        2 & \cmark & \cmark/\cmark & \xmark & 28.936 & \textbf{36.939} \\
        2 & \cmark & \cmark/\cmark & \cmark & 28.811 & 37.168 \\
        \bottomrule
        \end{tabular*}
        \begin{tablenotes}[flushleft]
            \PaperTableNotesStyle
            \item[] Check and cross marks denote enabled and disabled settings, respectively; N/A denotes a setting that is not applicable. Boldface and underlining denote the lowest and second-lowest value in each MAE column.
            \item[a] Each pair indicates whether the attention value and output projections are enabled, in that order.
            \item[b] E MAE and F MAE are mean absolute errors in meV/atom and meV/\AA, respectively; lower values are better.
        \end{tablenotes}
    \end{threeparttable}
\end{table}

\begin{table}[ht]
    \PaperTableStyle
    \begin{threeparttable}
        \caption{Ablation of attention aggregation for the 96-channel, 1-focus configuration.}
        \label{tab:sm_attention_96}
        \begin{tabular*}{\linewidth}{@{\extracolsep{\fill}}cccccc@{}}
        \toprule
        \textbf{Attn. heads} &
        \textbf{Attention} &
        \textbf{Attn. value/output proj.}\tnote{a} &
        \textbf{Pre mixing} &
        \textbf{E MAE}$\downarrow$\tnote{b} &
        \textbf{F MAE}$\downarrow$\tnote{b} \\
        \midrule
        0 & \xmark & N/A & N/A & 30.068 & 38.407 \\
        1 & \cmark & \xmark/\xmark & \xmark & \textbf{27.567} & 36.127 \\
        1 & \cmark & \xmark/\cmark & \xmark & 28.301 & 36.065 \\
        1 & \cmark & \cmark/\xmark & \xmark & 28.578 & 36.350 \\
        1 & \cmark & \cmark/\cmark & \xmark & 28.430 & 36.321 \\
        1 & \cmark & \cmark/\cmark & \cmark & 28.975 & 36.225 \\
        2 & \cmark & \xmark/\xmark & \xmark & 27.732 & 36.680 \\
        2 & \cmark & \xmark/\cmark & \xmark & 28.074 & 36.288 \\
        2 & \cmark & \cmark/\xmark & \xmark & \underline{27.571} & 36.119 \\
        2 & \cmark & \cmark/\cmark & \xmark & 27.968 & \textbf{35.457} \\
        2 & \cmark & \cmark/\cmark & \cmark & 28.473 & 35.858 \\
        3 & \cmark & \xmark/\xmark & \xmark & 28.170 & 36.419 \\
        3 & \cmark & \xmark/\cmark & \xmark & 28.016 & 36.017 \\
        3 & \cmark & \cmark/\xmark & \xmark & 27.593 & \underline{35.639} \\
        3 & \cmark & \cmark/\cmark & \xmark & 27.722 & 35.809 \\
        3 & \cmark & \cmark/\cmark & \cmark & 28.126 & 36.100 \\
        \bottomrule
        \end{tabular*}
        \begin{tablenotes}[flushleft]
            \PaperTableNotesStyle
            \item[] Check and cross marks denote enabled and disabled settings, respectively; N/A denotes a setting that is not applicable. Boldface and underlining denote the lowest and second-lowest value in each MAE column.
            \item[a] Each pair indicates whether the attention value and output projections are enabled, in that order.
            \item[b] E MAE and F MAE are mean absolute errors in meV/atom and meV/\AA, respectively; lower values are better.
        \end{tablenotes}
    \end{threeparttable}
\end{table}

\begin{table}[ht]
    \PaperTableStyle
    \begin{threeparttable}
        \caption{Ablation of attention aggregation for the 96-channel, 2-focus configuration.}
        \label{tab:sm_attention_96x2}
        \begin{tabular*}{\linewidth}{@{\extracolsep{\fill}}cccccc@{}}
        \toprule
        \textbf{Attn. heads} &
        \textbf{Attention} &
        \textbf{Attn. value/output proj.}\tnote{a} &
        \textbf{Pre mixing} &
        \textbf{E MAE}$\downarrow$\tnote{b} &
        \textbf{F MAE}$\downarrow$\tnote{b} \\
        \midrule
        0 & \xmark & N/A & N/A & 30.935 & 36.639 \\
        1 & \cmark & \xmark/\xmark & \xmark & \underline{28.083} & 34.158 \\
        1 & \cmark & \xmark/\cmark & \xmark & 28.449 & 34.508 \\
        1 & \cmark & \cmark/\xmark & \xmark & 28.363 & 34.460 \\
        1 & \cmark & \cmark/\cmark & \xmark & 28.212 & 35.045 \\
        1 & \cmark & \cmark/\cmark & \cmark & 28.885 & 35.387 \\
        2 & \cmark & \xmark/\xmark & \xmark & 28.335 & \underline{34.013} \\
        2 & \cmark & \xmark/\cmark & \xmark & 28.401 & 34.361 \\
        2 & \cmark & \cmark/\xmark & \xmark & \textbf{27.814} & \textbf{33.786} \\
        2 & \cmark & \cmark/\cmark & \xmark & 28.440 & 34.440 \\
        2 & \cmark & \cmark/\cmark & \cmark & 28.968 & 35.507 \\
        3 & \cmark & \xmark/\xmark & \xmark & 28.506 & 34.560 \\
        3 & \cmark & \xmark/\cmark & \xmark & 29.274 & 34.624 \\
        3 & \cmark & \cmark/\xmark & \xmark & 28.907 & 34.412 \\
        3 & \cmark & \cmark/\cmark & \xmark & 28.609 & 34.869 \\
        3 & \cmark & \cmark/\cmark & \cmark & 28.957 & 34.719 \\
        4 & \cmark & \cmark/\cmark & \cmark & 28.736 & 34.603 \\
        6 & \cmark & \xmark/\xmark & \xmark & 28.757 & 34.289 \\
        6 & \cmark & \xmark/\cmark & \xmark & 28.968 & 34.778 \\
        6 & \cmark & \cmark/\xmark & \xmark & 29.027 & 34.313 \\
        6 & \cmark & \cmark/\cmark & \xmark & 28.397 & 34.541 \\
        6 & \cmark & \cmark/\cmark & \cmark & 29.483 & 34.889 \\
        \bottomrule
        \end{tabular*}
        \begin{tablenotes}[flushleft]
            \PaperTableNotesStyle
            \item[] Check and cross marks denote enabled and disabled settings, respectively; N/A denotes a setting that is not applicable. Boldface and underlining denote the lowest and second-lowest value in each MAE column.
            \item[a] Each pair indicates whether the attention value and output projections are enabled, in that order.
            \item[b] E MAE and F MAE are mean absolute errors in meV/atom and meV/\AA, respectively; lower values are better.
        \end{tablenotes}
    \end{threeparttable}
\end{table}

\FloatBarrier
\subsection{Normalization placement}

Normalization placement is examined by applying RMSNorm as pre-normalization, post-normalization or both within the SO(2) and feed-forward sublayers (Table~\ref{tab:sm_stability}).
A single normalization per sublayer is consistently more accurate than applying both, which adds layers without benefit.
SO(2) post-normalization combined with feed-forward pre-normalization yields the lowest energy MAE and close to the lowest force MAE, and this placement is used throughout.

\begin{table}[ht]
    \PaperTableStyle
    \begin{threeparttable}
        \caption{Ablation of SO(2) and FFN normalization placement.}
        \label{tab:sm_stability}
        \begin{tabular*}{\linewidth}{@{\extracolsep{\fill}}cccccc@{}}
        \toprule
        \makecell{\textbf{SO(2)}\\\textbf{pre-norm}}  &
        \makecell{\textbf{SO(2)}\\\textbf{post-norm}} &
        \makecell{\textbf{FFN}\\\textbf{pre-norm}}    &
        \makecell{\textbf{FFN}\\\textbf{post-norm}}   &
        \textbf{E MAE}$\downarrow$\tnote{a}          &
        \textbf{F MAE}$\downarrow$\tnote{a}                                                     \\
        \midrule
        \cmark                                               & \xmark & \cmark  & \xmark & \underline{28.493} & 37.683          \\
        \xmark                                              & \cmark  & \cmark  & \xmark & \textbf{28.013} & \underline{37.624} \\
        \cmark                                               & \xmark & \xmark & \cmark  & 30.237          & 38.608          \\
        \xmark                                              & \cmark  & \xmark & \cmark  & 29.321          & \textbf{37.466} \\
        \cmark                                               & \xmark & \cmark  & \cmark  & 29.322          & 38.548          \\
        \cmark                                               & \cmark  & \cmark  & \xmark & 28.987          & 38.250          \\
        \xmark                                              & \cmark  & \cmark  & \cmark  & 28.808          & 37.840          \\
        \cmark                                               & \cmark  & \xmark & \cmark  & 29.841          & 39.068          \\
        \cmark                                               & \cmark  & \cmark  & \cmark  & 30.215          & 39.208          \\
        \bottomrule
        \end{tabular*}
        \begin{tablenotes}[flushleft]
            \PaperTableNotesStyle
            \item[] Check and cross marks denote enabled and disabled normalization, respectively. Boldface and underlining denote the lowest and second-lowest value in each MAE column.
            \item[a] E MAE and F MAE are mean absolute errors in meV/atom and meV/\AA, respectively; lower values are better.
        \end{tablenotes}
    \end{threeparttable}
\end{table}

\FloatBarrier
\subsection{Learning-rate schedule}

Under otherwise identical settings, the warmup--stable--decay (WSD) schedule lowers both the energy and force MAEs relative to a cosine schedule (Table~\ref{tab:sm_training_recipe}).
This advantage is consistent with its extended high-learning-rate phase and short terminal decay, and the WSD schedule is used for the released DPA4 variants.

\begin{table}[ht]
    \PaperTableStyle
    \begin{threeparttable}
        \caption{Ablation of the learning-rate scheduler.}
        \label{tab:sm_training_recipe}
        \begin{tabular*}{\linewidth}{@{\extracolsep{\fill}}ccc@{}}
        \toprule
        \textbf{LR scheduler}          &
        \textbf{E MAE}$\downarrow$\tnote{a} &
        \textbf{F MAE}$\downarrow$\tnote{a}                                             \\
        \midrule
        Cosine                          & 28.638          & 38.607          \\
        WSD                             & \textbf{27.828} & \textbf{37.022} \\
        \bottomrule
        \end{tabular*}
        \begin{tablenotes}[flushleft]
            \PaperTableNotesStyle
            \item[] Boldface denotes the lower value in each MAE column.
            \item[a] E MAE and F MAE are mean absolute errors in meV/atom and meV/\AA, respectively; lower values are better.
        \end{tablenotes}
    \end{threeparttable}
\end{table}

\FloatBarrier
\section{Inference benchmark protocol}
\label{sec:sm-inference}

The inference-throughput benchmarks use periodic diamond supercells at a lattice constant of 3.567~\AA{}.
This system permits a wide size sweep at fixed local geometry; within the 6~\AA{} DPA4 cutoff, every atom has 158 neighbors and therefore represents a comparatively dense inference workload.
For each target size, the cubic eight-atom cell is repeated by a near-cubic integer triplet, so the realized atom count can differ slightly from the target.
Each ASE calculator~\cite{ase-paper} is reset before requesting energy, forces and stress, with CUDA synchronized immediately before and after the timed region.
After 10 warm-up evaluations, 100 evaluations are timed and their mean wall time is used to compute throughput.

The adaptive size search begins at a target of 128 atoms and doubles the target after each successful trial until the first out-of-memory result.
The interval between the last successful target and the first failed target is then refined by three bisection trials; after each trial, the corresponding successful or failed boundary is replaced by the midpoint.
All ASE inference throughputs reported in this work use this size-search and timing protocol.
For the CPS--throughput Pareto analysis in Fig.~\ref{fig:cps-throughput}, saturated throughput is defined as the largest throughput among all successful trials, including the bisection trials; out-of-memory trials are excluded.
CPS values are taken from Table~\ref{tab:matbench}.
The saturated throughput of all six DPA4 size classes, the system size at which it is attained and the largest system evaluated within the memory of one GPU are collected in Table~\ref{tab:sm_saturated_throughput}, together with the corresponding values for ten representative compliant baselines.
The CPS--throughput comparison in Fig.~\ref{fig:cps-throughput} includes the five DPA4 variants evaluated on Matbench Discovery and the ten baselines: DPA-3.1-MPtrj~\cite{zhang2025dpa3}, EquiformerV3+DeNS-MP~\cite{liao2026equiformerv3,liao2024dens}, MatRIS-10M-MP~\cite{zhou2026matris}, Nequix MP PFT~\cite{koker2026pft}, Nequip-MP-L~\cite{batzner2022nequip}, Allegro-MP-L~\cite{musaelian2022learning}, SevenNet-l3i5~\cite{park2024scalable}, MACE-MP-0~\cite{batatia2022mace,batatia2023foundation}, ORB v2 MPtrj~\cite{neumann2024orb} and CHGNet~\cite{deng2023chgnet}.
For each baseline, we use the checkpoint corresponding to the compliant Matbench Discovery entry in Table~\ref{tab:matbench}.
The public artifact identifiers are \texttt{matris\_10m\_mp}, \texttt{nequix-mp-1-pft}, \texttt{SevenNet-l3i5}, the \texttt{medium} MACE-MP-0 checkpoint, the ORB loader \texttt{orb\_mptraj\_only\_v2}, CHGNet 0.3.0, and version 0.1 of \texttt{mir-group/NequIP-MP-L} and \texttt{mir-group/Allegro-MP-L}; the EquiformerV3 artifact and implementation revision are specified below.

\begin{table}[ht]
    \PaperTableStyle
    \sisetup{group-separator={{,}},group-minimum-digits=4}
    \begin{threeparttable}
        \caption{Single-GPU saturated inference throughput on diamond supercells.}
        \label{tab:sm_saturated_throughput}
        \begin{tabular*}{\linewidth}{@{\extracolsep{\fill}}lccc@{}}
        \toprule
        \textbf{Model} &
        {\textbf{Sat.\ throughput (atoms/ms)}} &
        {\textbf{Saturating size (atoms)}\tnote{a}} &
        {\textbf{Largest size (atoms)}\tnote{b}} \\
        \midrule
        DPA4-Pro             & \num{1.36}  & \num{1920}  & \num{1920}  \\
        DPA4-Plus            & \num{5.14}  & \num{6144}  & \num{7200}  \\
        DPA4-Air             & \num{11.6}  & \num{12320} & \num{12320} \\
        DPA4-Neo             & \num{28.0}  & \num{20384} & \num{20384} \\
        DPA4-Mini            & \num{79.2}  & \num{49248} & \num{53176} \\
        DPA4-Nano            & \num{153}   & \num{65664} & \num{97336} \\
        \midrule
        ORB v2 MPtrj         & \num{52.2}  & \num{1440}  & \num{25088} \\
        SevenNet-l3i5        & \num{27.8}  & \num{30056} & \num{30056} \\
        Nequix MP PFT        & \num{18.2}  & \num{16016} & \num{25088} \\
        MACE-MP-0            & \num{11.5}  & \num{5832}  & \num{39304} \\
        Nequip-MP-L          & \num{3.35}  & \num{2016}  & \num{2016}  \\
        Allegro-MP-L         & \num{1.40}  & \num{2016}  & \num{3024}  \\
        DPA-3.1-MPtrj        & \num{0.94}  & \num{720}   & \num{720}   \\
        EquiformerV3+DeNS-MP & \num{0.43}  & \num{288}   & \num{288}   \\
        CHGNet               & \num{0.37}  & \num{96}    & \num{1000}  \\
        MatRIS-10M-MP        & \num{0.055} & \num{512}   & \num{512}   \\
        \bottomrule
        \end{tabular*}
        \begin{tablenotes}[flushleft]
            \PaperTableNotesStyle
            \item[a] Saturating size is the realized atom count at which the largest throughput among all successful trials is attained.
            \item[b] Largest size is the largest supercell successfully evaluated within the memory of one NVIDIA H20 GPU.
        \end{tablenotes}
    \end{threeparttable}
\end{table}

All measurements are performed on a single NVIDIA H20 GPU with CUDA 12.8, Ubuntu 20.04.6 LTS and GCC 9.4.0 through the ASE interface.
DPA4 inference uses our fused Triton operators and custom CUDA kernels with PyTorch 2.11.0.
The baseline environments use model-compatible PyTorch 2.8.0 or 2.10.0 releases.
MACE 0.3.15 and SevenNet 0.12.1 use cuEquivariance 0.10.0~\cite{nvidia_cuequivariance}; the SevenNet flash path is disabled.
ORB 0.5.5 uses its CUDA \texttt{torch.compile} path, whereas Nequix 0.4.5 uses its Torch kernel path with graph compilation disabled.
Nequip-MP-L and Allegro-MP-L are compiled on the H20 with AOTInductor and loaded through the NequIP 0.16.1 compiled-model calculator with matscipy 1.2.0 neighbor lists; the Allegro 0.8.0 package additionally uses its Triton contraction backend.
MatRIS is evaluated at commit \texttt{c16f569c} with its compiled Cython graph converter, and CHGNet 0.3.0 uses its standard CUDA ASE calculator.
These baseline environments use Python 3.11.15, ASE 3.28.0 and NumPy 2.4.4, except for the CHGNet 0.3.0 compatibility environment, which uses ASE 3.22.1 and NumPy 1.24.4.

For Fig.~\ref{fig:infer_diamond}, MACE-Omat-Small and MACE-Omat-Medium are the \texttt{mace-omat-0-small} and \texttt{mace-omat-0-medium} checkpoints of the MACE-OMAT-0 release~\cite{batatia2025crosslearning}; these checkpoints also supply the corresponding rows of Table~\ref{tab:omat24}.
The EquiformerV3 timing uses the MPtrj checkpoint of the $L_{\max}=4$ architecture and the \texttt{atomicarchitects/equiformer\_v3} implementation at commit \texttt{a7300c5}~\cite{liao2026equiformerv3}.

\FloatBarrier
\section{Close-contact probe of the ZBL branch}
\label{sec:sm-zbl}

The benchmark datasets of Sections~\ref{sec:matbench}--\ref{sec:small_mol} do not sample the sub-\AA{} regime, so the close-contact behavior of native ZBL zone bridging (Section~\ref{sec:zbl_method}) is probed directly with a C--Si dimer scan.
Both a DPA3 baseline with the DeePMD DP-ZBL pairwise correction~\cite{wang2019dpzbl} and a DPA4 model with native ZBL zone bridging are trained on a 3C-SiC (cubic silicon carbide) dataset computed with the ABACUS DFT package~\cite{li2016abacus}.
The scan drives the pair into the regime where ordinary DFT training data are sparse and the screened nuclear repulsion should dominate, and therefore tests the transition between the learned potential and the analytical ZBL limit.

Figure~\ref{fig:zbl} shows that the energy curves appear smooth over most of the scan, whereas the force curves expose the difference between the two coupling strategies.
The DPA3 DP-ZBL baseline develops a sharp force excursion near the switching region, producing a local attractive impulse even though the analytical ZBL force~\cite{ziegler1985} is strongly repulsive at these distances.
The location and sign of this excursion are consistent with the switching-force term in Eq.~\eqref{eq:zbl-switching-force}: the extra contribution is governed by the mismatch between the ZBL and learned energies inside the splice window, not by the monotone screened-Coulomb repulsion.

DPA4 carries no such switching term.
The analytical branch is evaluated on the true distance, the learned branch sees the $C^3$ clamped displacement, and the source-freeze gate suppresses the direct learned short-range pair channel (Section~\ref{sec:zbl_method}).
Consequently, DPA4 follows the analytical ZBL force in the inner region and joins smoothly to the learned force as the C--Si distance increases.
This dimer scan is a local close-contact probe; validating long-time energy drift, collision stability and damage evolution requires separate many-body molecular-dynamics tests.

\begin{figure}[ht]
    \centering
    \includegraphics[width=0.9\textwidth]{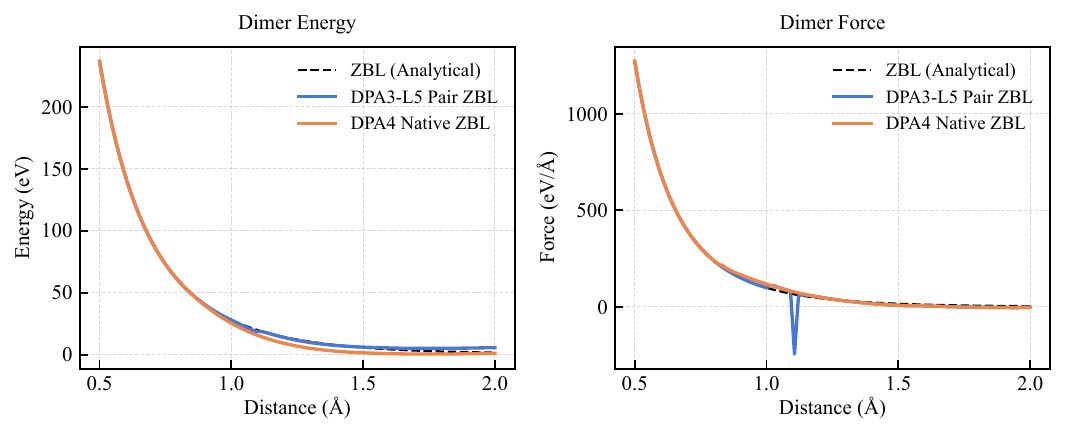}
    \caption{
        Short-range C--Si dimer response for models trained on an ABACUS-computed~\cite{li2016abacus} 3C-SiC dataset.
        The DPA3 baseline uses the DeePMD DP-ZBL pairwise correction~\cite{wang2019dpzbl}, whereas DPA4 uses the native ZBL zone bridging branch described in Section~\ref{sec:zbl_method}.
        The analytical reference is the ZBL screened Coulomb potential~\cite{ziegler1985}.
    }
    \label{fig:zbl}
\end{figure}

\FloatBarrier
\section{Model and training configurations}

\subsection{Ablation model configurations}

Table~\ref{tab:sm_main_ablation_config} gives the configurations for two groups of experiments.
The first is the cumulative ablation of main-text Fig.~\ref{fig:ablation}, whose four controlled switches are the message node WBN, the WBN in the equivariant FFN, the edge--node degree coupling and the attention aggregation.
The second comprises the mechanism ablations: graph compilation, attention aggregation, multi-focus design, the low-rank edge--node SO(2)-equivariant degree coupling, the spherical-grid nonlinearity and the Wigner bilinear network.
Table~\ref{tab:sm_supplementary_ablation_config} gives those for the model-selection and robustness ablations.
Entries marked by ``--'' are the controlled variables within the corresponding experiment family; all other entries define the shared reference setting.
Thin horizontal rules separate related groups of model or training hyperparameters without adding category labels; a thick rule marks the boundary between the model and training blocks.

\begin{table}[p]
    \centering
    \scriptsize
    \renewcommand{\arraystretch}{0.94}
    \setlength{\tabcolsep}{2pt}
    \begin{threeparttable}
        \caption{Configurations for the primary architecture ablations.}
        \label{tab:sm_main_ablation_config}
        \begin{tabular*}{\linewidth}{@{\extracolsep{\fill}}
                L{3.0cm}C{1.42cm}C{1.42cm}C{1.42cm}C{1.42cm}C{1.42cm}C{1.42cm}C{1.42cm}@{}} \toprule \textbf{Hyperparameter} & \makecell{\textbf{Cumulative}\\\textbf{ablation}} & \makecell{\textbf{Compile/}\\\textbf{precision}} & \textbf{Attention}\tnote{a} & \makecell{\textbf{Multi-focus}\\\textbf{SO(2)}} & \makecell{\textbf{Edge--node}\\\textbf{degree}\\\textbf{coupling}} & \makecell{\textbf{Sph.
                    grid/}\\\textbf{quad.}} & \textbf{WBN} \\
            \midrule
            Feature dim.                  & 64                                                & 64                                               & 64/96/96                                                                     & --                                              & 64                                               & 64                                          & 64 \\
            No. focuses                   & 1                                                 & 1                                                & 1/1/2                                                                        & --                                              & 1                                                & 1                                           & 1 \\
            No. layers                    & 2                                                 & 3                                                & 3                                                                            & 2                                               & 3                                                & 3                                           & 2 \\
            SO(2) layers                  & 4                                                 & 4                                                & 4                                                                            & 4                                               & 4                                                & 4                                           & 4 \\
            FFN layers                    & 1                                                 & 1                                                & 1                                                                            & 1                                               & 1                                                & 1                                           & 1 \\
            \midrule
            Radial basis                  & Bessel                                            & Bessel                                           & Bessel                                                                       & Bessel                                          & Bessel                                           & Bessel                                      & Bessel \\
            No. radial bases              & 16                                                & 16                                               & 16                                                                           & 16                                              & 16                                               & 16                                          & 16 \\
            $L_{\max}$                    & 2                                                 & 3                                                & 3                                                                            & 3                                               & 3                                                & 3                                           & 2 \\
            $M_{\max}$                    & 1                                                 & 1                                                & 1                                                                            & 1                                               & 1                                                & 1                                           & 1 \\
            Cutoff (\AA)                  & 6                                                 & 6                                                & 6                                                                            & 6                                               & 6                                                & 6                                           & 6 \\
            \midrule
            Edge--node degree coupling\tnote{b}   & --                                                & Degree mixing                                    & Scalar scaling                                                               & Scalar scaling                                  & --                                               & Degree mixing                               & Degree mixing \\
            Per-channel coupling              & --                                                & True                                             & N/A                                                                          & N/A                                             & --                                               & True                                        & True \\
            Coupling rank               & --                                                & 1                                                & N/A                                                                          & N/A                                             & --                                               & 1                                           & 2 \\
            \midrule
            Attn. heads                   & --                                                & 1                                                & --                                                                           & 1                                               & 1                                                & 1                                           & 1 \\
            Attn. value/output proj.  & False                                             & False                                            & False                                                                        & False                                           & False                                            & False                                       & False \\
            Pre mixing                    & False                                             & False                                            & False                                                                        & False                                           & False                                            & False                                       & False \\
            \midrule
            Message node WBN              & --                                                & False                                            & False                                                                        & False                                           & False                                            & False                                       & -- \\
            WBN in FFN                    & --                                                & False                                            & False                                                                        & False                                           & False                                            & False                                       & -- \\
            Quadrature                    & Lebedev                                           & Lebedev                                          & Lebedev                                                                      & Product                                         & Lebedev                                          & --                                          & Lebedev \\
            \midrule
            Norm. placement\tnote{c}      & Post \& Pre                                       & Post \& Pre                                      & Post \& Pre                                                                  & Pre \& Pre                                      & Post \& Pre                                      & Pre \& Pre                                  & Post \& Pre \\
            Activation func.              & SiLU                                              & SiLU                                             & SiLU                                                                         & SiLU                                            & SiLU                                             & SiLU                                        & SiLU \\
            FFN hidden dim.\tnote{d}      & Auto                                              & Auto                                             & Auto                                                                         & Auto                                            & Auto                                             & Auto                                        & Auto \\
            SO(3) read-out                & False                                             & False                                            & False                                                                        & False                                           & False                                            & False                                       & False \\
            Output fitting dim.           & Auto                                              & Auto                                             & Auto                                                                         & Auto                                            & Auto                                             & Auto                                        & Auto \\
            Output fitting layers         & 1                                                 & 1                                                & 1                                                                            & 1                                               & 1                                                & 1                                           & 1 \\
            \midrule[\heavyrulewidth]
            Compile                       & True                                              & --                                               & True                                                                         & True                                            & True                                             & True                                        & True \\
            BF16 AMP                      & True                                              & --                                               & True                                                                         & True                                            & True                                             & True                                        & True \\
            TF32 matmul                   & True                                              & --                                               & True                                                                         & True                                            & True                                             & True                                        & True \\
            \midrule
            LR scheduler\tnote{e}         & WSD                                               & Cosine                                           & Cosine                                                                       & WSD                                             & Cosine                                           & Cosine                                      & WSD \\
            Max.
            LR & $4\times10^{-4}$ & $4\times10^{-4}$ & \makecell{$4.5/4.2/3.5$\\$\times10^{-4}$} & $4\times10^{-4}$ & $4.5\times10^{-4}$ & $4.5\times10^{-4}$ & $4\times10^{-4}$ \\ Min.
            LR & $1\times10^{-6}$ & $1\times10^{-6}$ & $1\times10^{-6}$ & $1\times10^{-6}$ & $1\times10^{-6}$ & $1\times10^{-6}$ & $1\times10^{-6}$ \\ Warmup ratio & 0.003 & 0.003 & 0.003 & 0.003 & 0.003 & 0.003 & 0.003 \\ Decay ratio & 0.65 & N/A & N/A & 0.65 & N/A & N/A & 0.65 \\ Decay type & Cosine & N/A & N/A & Cosine & N/A & N/A & Cosine \\ Batch size (per GPU)\tnote{f} & $\lceil 900/N\rceil$ & $\lceil 450/N\rceil$ & \makecell{$\lceil 1000/N\rceil$\\$\lceil 700/N\rceil$\\$\lceil 400/N\rceil$} & $\lceil 600/N\rceil$ & $\lceil 1000/N\rceil$ & $\lceil 700/N\rceil$ & $\lceil 900/N\rceil$ \\ Training steps & $1\times10^6$ & $1\times10^6$ & $1\times10^6$ & $2\times10^6$ & $1\times10^6$ & $1\times10^6$ & $5\times10^5$ \\ No.
            GPUs                      & 1                                                 & 1                                                & 1                                                                            & 1                                               & 1                                                & 1                                           & 1 \\
            \midrule
            Loss                          & MAE                                               & MAE                                              & MAE                                                                          & MAE                                             & MAE                                              & MAE                                         & MAE \\
            Loss weights $(E,F,V)$\tnote{g} & 20, 20, 5                                      & 20, 20, 5                                        & 20, 20, 5                                                                    & 20, 20, 5                                       & 20, 20, 5                                        & 20, 20, 5                                   & 20, 20, 5 \\
            Optimizer                     & Hybrid\hspace{0pt}Muon                            & Hybrid\hspace{0pt}Muon                           & Hybrid\hspace{0pt}Muon                                                       & Hybrid\hspace{0pt}Muon                          & Hybrid\hspace{0pt}Muon                           & Hybrid\hspace{0pt}Muon                      & Hybrid\hspace{0pt}Muon \\
            Muon mode                     & Slice                                             & Slice                                            & Slice                                                                        & Slice                                           & Slice                                            & Slice                                       & Slice \\
            Magma Lite                    & True                                              & True                                             & True                                                                         & True                                            & True                                             & True                                        & True \\
            Weight decay                  & $1\times10^{-3}$                                  & $1\times10^{-3}$                                 & $1\times10^{-3}$                                                             & $1\times10^{-3}$                                & $1\times10^{-3}$                                 & $1\times10^{-3}$                            & $1\times10^{-3}$ \\
            \bottomrule
        \end{tabular*}
        \begin{tablenotes}[flushleft]
            \PaperTableNotesStyle
            \item[] A dash indicates a controlled variable within the corresponding ablation family; N/A denotes a parameter that is not used by the corresponding setting.
            \item[a] Multi-value entries in the attention column follow the 64-channel 1-focus, 96-channel 1-focus and 96-channel 2-focus settings.
            \item[b] Scalar scaling applies each degree-$l$ radial profile only to the degree-$l$ channels, whereas degree mixing uses the profiles to mix angular degrees in the local SO(2) frame.
            \item[c] In each normalization-placement entry, the first term refers to the SO(2) subblock and the second to the FFN subblock.
            \item[d] Auto denotes a hidden dimension of $(8/3)d_\mathrm{feat}$, rounded up to a multiple of 32.
            \item[e] LR denotes learning rate, and WSD denotes the warmup--stable--decay schedule.
            \item[f] $N$ denotes the number of atoms in each system; $\lceil\cdot\rceil$ rounds up to the nearest integer.
            \item[g] $E$, $F$ and $V$ denote the energy, force and virial loss contributions, respectively.
        \end{tablenotes}
    \end{threeparttable}
\end{table}

\begin{table}[p]
    \centering
    \scriptsize
    \renewcommand{\arraystretch}{0.97}
    \setlength{\tabcolsep}{3pt}
    \begin{threeparttable}
        \caption{Configurations for the supplementary hyperparameter ablations.}
        \label{tab:sm_supplementary_ablation_config}
        \begin{tabular*}{\linewidth}{@{\extracolsep{\fill}}
                L{3.0cm}C{1.6cm}C{1.6cm}C{1.6cm}C{1.6cm}C{1.6cm}C{1.6cm}@{}} \toprule \textbf{Hyperparameter} & \textbf{Layers} & \makecell{\textbf{SO(2)}\\\textbf{stack}} & \makecell{\textbf{FFN}\\\textbf{stack}} & \textbf{Attention}\tnote{a} & \makecell{\textbf{Norm.
                }\\\textbf{placement}} & \makecell{\textbf{LR}\\\textbf{scheduler}} \\
            \midrule
            Feature dim. & 64 & 64 & 64 & 64/96/96 & 64 & 64 \\
            No. focuses & 1 & 1 & 1 & 1/1/2 & 1 & 1 \\
            No. layers & -- & 2 & 2 & 3 & 3 & 3 \\
            SO(2) layers & 4 & -- & 3 & 4 & 4 & 4 \\
            FFN layers & 1 & 1 & -- & 1 & 1 & 1 \\
            \midrule
            Radial basis & Bessel & Bessel & Bessel & Bessel & Bessel & Bessel \\
            No. radial bases & 16 & 16 & 16 & 16 & 16 & 16 \\
            $L_{\max}$ & 3 & 3 & 3 & 3 & 3 & 3 \\
            $M_{\max}$ & 1 & 1 & 1 & 1 & 1 & 1 \\
            Cutoff (\AA) & 6 & 6 & 6 & 6 & 6 & 6 \\
            \midrule
            Edge--node degree coupling\tnote{b} & Scalar scaling & Degree mixing & Degree mixing & Scalar scaling & Scalar scaling & Degree mixing \\
            Per-channel coupling & N/A & True & True & N/A & N/A & True \\
            Coupling rank & N/A & 1 & 1 & N/A & N/A & 1 \\
            \midrule
            Attn. heads & 1 & 1 & 1 & -- & 1 & 1 \\
            Attn. value/output proj. & False & False & False & -- & False & False \\
            Pre mixing & False & False & False & -- & False & False \\
            \midrule
            Message node WBN & False & False & False & False & False & False \\
            WBN in FFN & False & False & False & False & False & False \\
            Quadrature & Product & Lebedev & Lebedev & Lebedev & Product & Lebedev \\
            \midrule
            Norm. placement\tnote{c} & Pre \& Pre & Post \& Pre & Post \& Pre & Post \& Pre & -- & Post \& Pre \\
            Activation func. & SiLU & SiLU & SiLU & SiLU & SiLU & SiLU \\
            FFN hidden dim.\tnote{d} & Auto & Auto & Auto & Auto & Auto & Auto \\
            SO(3) read-out & False & False & False & False & False & False \\
            Output fitting dim. & Auto & Auto & Auto & Auto & Auto & Auto \\
            Output fitting layers & 1 & 1 & 1 & 1 & 1 & 1 \\
            \midrule[\heavyrulewidth]
            Compile & True & True & True & True & True & True \\
            BF16 AMP & True & True & True & True & True & True \\
            TF32 matmul & True & True & True & True & True & True \\
            \midrule
            LR scheduler\tnote{e} & Cosine & WSD & WSD & Cosine & WSD & -- \\
            Max.
            LR & $5\times10^{-4}$ & $6.5\times10^{-4}$ & $4.5\times10^{-4}$ & \makecell{$4.5/4.2/3.5$\\$\times10^{-4}$} & $4\times10^{-4}$ & $4.5\times10^{-4}$ \\ Min.
            LR & $1\times10^{-6}$ & $1\times10^{-6}$ & $1\times10^{-6}$ & $1\times10^{-6}$ & $1\times10^{-6}$ & $1\times10^{-6}$ \\ Warmup ratio & 0.003 & 0.003 & 0.003 & 0.003 & 0.003 & 0.003 \\ Decay ratio & N/A & 0.65 & 0.65 & N/A & 0.65 & -- \\ Decay type & N/A & Cosine & Cosine & N/A & Cosine & -- \\ Batch size (per GPU)\tnote{f} & $\lceil 512/N\rceil$ & $\lceil 2100/N\rceil$ & $\lceil 900/N\rceil$ & \makecell{$\lceil 1000/N\rceil$\\$\lceil 700/N\rceil$\\$\lceil 400/N\rceil$} & $\lceil 1000/N\rceil$ & $\lceil 1000/N\rceil$ \\ Training steps & $2\times10^6$ & $2\times10^6$ & $1\times10^6$ & $1\times10^6$ & $1\times10^6$ & $1\times10^6$ \\ No.
            GPUs & 1 & 1 & 1 & 1 & 1 & 1 \\
            \midrule
            Loss & MAE & MAE & MAE & MAE & MAE & MAE \\
            Loss weights $(E,F,V)$\tnote{g} & 20, 20, 5 & 20, 20, 5 & 20, 20, 5 & 20, 20, 5 & 20, 20, 5 & 20, 20, 5 \\
            Optimizer & HybridMuon & HybridMuon & HybridMuon & HybridMuon & HybridMuon & HybridMuon \\
            Muon mode & Slice & Slice & Slice & Slice & Slice & Slice \\
            Magma Lite & True & True & True & True & True & True \\
            Weight decay & $1\times10^{-3}$ & $1\times10^{-3}$ & $1\times10^{-3}$ & $1\times10^{-3}$ & $1\times10^{-3}$ & $1\times10^{-3}$ \\
            \bottomrule
        \end{tabular*}
        \begin{tablenotes}[flushleft]
            \PaperTableNotesStyle
            \item[] A dash indicates a controlled variable within the corresponding ablation family; N/A denotes a parameter that is not used by the corresponding setting.
            \item[a] Multi-value entries in the attention column follow the 64-channel 1-focus, 96-channel 1-focus and 96-channel 2-focus settings.
            \item[b] Scalar scaling applies each degree-$l$ radial profile only to the degree-$l$ channels, whereas degree mixing uses the profiles to mix angular degrees in the local SO(2) frame.
            \item[c] In each normalization-placement entry, the first term refers to the SO(2) subblock and the second to the FFN subblock.
            \item[d] Auto denotes a hidden dimension of $(8/3)d_\mathrm{feat}$, rounded up to a multiple of 32.
            \item[e] LR denotes learning rate, and WSD denotes the warmup--stable--decay schedule.
            \item[f] $N$ denotes the number of atoms in each system; $\lceil\cdot\rceil$ rounds up to the nearest integer.
            \item[g] $E$, $F$ and $V$ denote the energy, force and virial loss contributions, respectively.
        \end{tablenotes}
    \end{threeparttable}
\end{table}

\subsection{Benchmark model configurations}

Table~\ref{tab:sm_model_config} lists the model and systems configurations of the six DPA4 size classes considered in this work.
The model settings of each size class are dataset-independent; only the training schedule varies between datasets.
Tables~\ref{tab:sm_matbench_config}, \ref{tab:sm_omat24_config}, \ref{tab:sm_omol25_config} and \ref{tab:sm_spice_mace_off_config} give the training hyperparameters for the Matbench Discovery, OMat24, OMol25 and SPICE-MACE-OFF benchmarks, respectively.

\begin{table}[p]
    \PaperTableStyle
    \renewcommand{\arraystretch}{1.05}
    \begin{threeparttable}
        \caption{Model and systems configurations of the six DPA4 size classes.}
        \label{tab:sm_model_config}
        \begin{tabular*}{\linewidth}{@{\extracolsep{\fill}}
                L{2.6cm}C{1.65cm}C{1.65cm}C{1.65cm}C{1.65cm}C{1.65cm}C{1.65cm}@{}} \toprule \textbf{Hyperparameter} & \textbf{DPA4-Nano} & \textbf{DPA4-Mini} & \textbf{DPA4-Neo} & \textbf{DPA4-Air} & \textbf{DPA4-Plus} & \textbf{DPA4-Pro} \\ \midrule Feature dim.
            & 32                 & 32                 & 64                 & 64                 & 64                 & 128                \\
            No. focuses                  & 1                  & 1                  & 2                  & 1                  & 1                  & 2                  \\
            No. layers                   & 2                  & 2                  & 2                  & 3                  & 4                  & 6                  \\
            SO(2) layers                 & 3                  & 3                  & 3                  & 4                  & 4                  & 4                  \\
            FFN layers                   & 1                  & 1                  & 1                  & 1                  & 1                  & 1                  \\
            \midrule
            Radial basis                 & Bessel             & Bessel             & Bessel             & Bessel             & Bessel             & Bessel             \\
            No. radial bases             & 16                 & 16                 & 16                 & 16                 & 16                 & 16                 \\
            $L_{\max}$                   & 1                  & 2                  & 3                  & 3                  & 4                  & 5                  \\
            $M_{\max}$                   & 1                  & 1                  & 1                  & 1                  & 1                  & 1                  \\
            Cutoff (\AA)                 & 6                  & 6                  & 6                  & 6                  & 6                  & 6                  \\
            \midrule
            Edge--node degree coupling\tnote{a}  & Scalar scaling     & Degree mixing      & Degree mixing      & Degree mixing      & Degree mixing      & Degree mixing      \\
            Per-channel coupling             & N/A                & True               & True               & True               & True               & True               \\
            Coupling rank              & N/A                & 1                  & 1                  & 1                  & 2                  & 2                  \\
            \midrule
            Attn. heads                  & 1                  & 1                  & 1                  & 1                  & 1                  & 1                  \\
            Attn. value/output proj. & False              & False              & False              & False              & False              & False              \\
            Pre mixing                   & False              & False              & False              & False              & False              & False              \\
            \midrule
            Message node WBN             & True               & True               & True               & True               & True               & True               \\
            WBN in FFN                   & True               & True               & True               & True               & True               & True               \\
            Quadrature                   & Lebedev            & Lebedev            & Lebedev            & Lebedev            & Lebedev            & Lebedev            \\
            \midrule
            Norm. placement\tnote{b}     & Post \& Pre        & Post \& Pre        & Post \& Pre        & Post \& Pre        & Post \& Pre        & Post \& Pre        \\
            Activation func.             & SiLU               & SiLU               & SiLU               & SiLU               & SiLU               & SiLU               \\
            FFN hidden dim.\tnote{c}     & Auto               & Auto               & Auto               & Auto               & Auto               & Auto               \\
            SO(3) read-out               & True               & True               & True               & True               & True               & False              \\
            Output fitting dim.          & Auto               & Auto               & Auto               & Auto               & Auto               & Auto               \\
            Output fitting layers        & 1                  & 1                  & 1                  & 1                  & 1                  & 1                  \\
            \midrule[\heavyrulewidth]
            Compile                      & True               & True               & True               & True               & True               & True               \\
            BF16 AMP                     & True               & True               & True               & True               & True               & True               \\
            TF32 matmul                  & True               & True               & True               & True               & True               & True               \\
            \midrule
            Optimizer                    & Hybrid\hspace{0pt}Muon & Hybrid\hspace{0pt}Muon & Hybrid\hspace{0pt}Muon & Hybrid\hspace{0pt}Muon & Hybrid\hspace{0pt}Muon & Hybrid\hspace{0pt}Muon \\
            Muon mode                    & Slice              & Slice              & Slice              & Slice              & Slice              & Slice              \\
            Magma Lite                   & True               & True               & True               & True               & True               & True               \\
            Weight decay                 & $1\times10^{-3}$   & $1\times10^{-3}$   & $1\times10^{-3}$   & $1\times10^{-3}$   & $1\times10^{-3}$   & $1\times10^{-3}$   \\
            \midrule
            Params                       & 0.48M              & 0.66M              & 1.1M               & 5.1M               & 8.8M               & 25.2M              \\
            \bottomrule
        \end{tabular*}
        \begin{tablenotes}[flushleft]
            \PaperTableNotesStyle
            \item[] N/A denotes a setting that is not applicable. Params is the number of trainable parameters.
            \item[a] Scalar scaling applies each degree-$l$ radial profile only to the degree-$l$
            channels, whereas degree mixing uses the profiles to mix angular degrees in the local SO(2)
            frame.
            \item[b] In the normalization-placement entry, the first term refers to the SO(2)
            subblock and the second term to the FFN subblock.
            \item[c] Auto denotes a hidden dimension of $(8/3)d_\mathrm{feat}$, rounded
            up to a multiple of 32.
        \end{tablenotes}
    \end{threeparttable}
\end{table}

\begin{table}[p]
    \PaperTableStyle
    \begin{threeparttable}
        \caption{Matbench Discovery training hyperparameters.}
        \label{tab:sm_matbench_config}
        \begin{tabular*}{\linewidth}{@{\extracolsep{\fill}}
                L{3.2cm}C{1.9cm}C{1.9cm}C{1.9cm}C{1.9cm}C{1.9cm}@{}} \toprule \textbf{Hyperparameter} & \textbf{DPA4-Mini} & \textbf{DPA4-Neo} & \textbf{DPA4-Air} & \textbf{DPA4-Plus} & \textbf{DPA4-Pro} \\ \midrule LR scheduler\tnote{a} & WSD & WSD & WSD & WSD & WSD \\ Max.
            LR & $2\times10^{-3}$ & $7\times10^{-4}$ & $4.2\times10^{-4}$ & $4.7\times10^{-4}$ & $4.3\times10^{-4}$ \\ Min.
            LR & $1\times10^{-6}$ & $1\times10^{-6}$ & $1\times10^{-6}$ & $1\times10^{-6}$ & $1\times10^{-6}$ \\ Warmup ratio & 0.003 & 0.003 & 0.003 & 0.003 & 0.003 \\ Decay ratio & 0.65 & 0.65 & 0.65 & 0.65 & 0.65 \\ Decay type & Cosine & Cosine & Cosine & Cosine & Cosine \\ Batch size (per GPU)\tnote{b} & $\lceil 7000/N\rceil$ & $\lceil 3000/N\rceil$ & $\lceil 3150/N\rceil$ & $\lceil 1500/N\rceil$ & $\lceil 300/N\rceil$ \\ Training epochs & 150 & 100 & 100 & 100 & 140 \\ No.
            GPUs & 1 & 2 & 1 & 4 & 16 \\
            \midrule
            Loss & MAE & MAE & MAE & MAE & MAE \\
            Loss weights $(E,F,V)$\tnote{c} & 20, 20, 5 & 20, 20, 5 & 20, 20, 5 & 20, 20, 5 & 20, 20, 5 \\
            \bottomrule
        \end{tabular*}
        \begin{tablenotes}[flushleft]
            \PaperTableNotesStyle
            \item[a] LR denotes learning rate, and WSD denotes the warmup--stable--decay schedule.
            \item[b] $N$ denotes the number of atoms in each system; $\lceil\cdot\rceil$ rounds up to the nearest integer.
            \item[c] $E$, $F$ and $V$ denote the energy, force and virial loss contributions, respectively.
        \end{tablenotes}
    \end{threeparttable}
\end{table}

\begin{table}[p]
    \PaperTableStyle
    \begin{threeparttable}
        \caption{OMat24 training hyperparameters.}
        \label{tab:sm_omat24_config}
        \begin{tabular*}{\linewidth}{@{\extracolsep{\fill}}
                L{2.65cm}C{1.66cm}C{1.66cm}C{1.66cm}C{1.66cm}C{1.66cm}C{1.66cm}@{}} \toprule \textbf{Hyperparameter} & \makecell{\textbf{DPA4-}\\\textbf{Nano}} & \makecell{\textbf{DPA4-}\\\textbf{Mini}} & \makecell{\textbf{DPA4-}\\\textbf{Neo}} & \makecell{\textbf{DPA4-}\\\textbf{Air}} & \makecell{\textbf{DPA4-}\\\textbf{Plus}} & \makecell{\textbf{DPA4-}\\\textbf{Pro}} \\ \midrule LR scheduler\tnote{a} & Cosine & Cosine & Cosine & Cosine & Cosine & Cosine \\ Max.
            LR & $3\times10^{-3}$ & $1\times10^{-3}$ & $8\times10^{-4}$ & $4\times10^{-4}$ & $3\times10^{-4}$ & $2\times10^{-4}$ \\ Min.
            LR & $1\times10^{-6}$ & $1\times10^{-6}$ & $1\times10^{-6}$ & $1\times10^{-6}$ & $1\times10^{-6}$ & $1\times10^{-6}$ \\ Warmup ratio & 0.003 & 0.003 & 0.003 & 0.003 & 0.003 & 0.003 \\ Batch size (per GPU)\tnote{b} & $\lceil 30000/N\rceil$ & $\lceil 10000/N\rceil$ & $\lceil 5000/N\rceil$ & $\lceil 3900/N\rceil$ & $\lceil 2000/N\rceil$ & $\lceil 400/N\rceil$ \\ Training epochs & 10 & 7 & 7 & 6 & 5 & 5 \\ No.
            GPUs & 1 & 2 & 2 & 4 & 4 & 16 \\
            \midrule
            Loss & MAE & MAE & MAE & MAE & MAE & MAE \\
            Loss weights $(E,F,V)$\tnote{c} & 20, 20, 5 & 20, 20, 5 & 20, 20, 5 & 20, 20, 5 & 20, 20, 5 & 20, 20, 5 \\
            \bottomrule
        \end{tabular*}
        \begin{tablenotes}[flushleft]
            \PaperTableNotesStyle
            \item[a] LR denotes learning rate.
            \item[b] $N$ denotes the number of atoms in each system; $\lceil\cdot\rceil$ rounds up to the nearest integer.
            \item[c] $E$, $F$ and $V$ denote the energy, force and virial loss contributions, respectively.
        \end{tablenotes}
    \end{threeparttable}
\end{table}

\begin{table}[p]
    \PaperTableStyle
    \begin{threeparttable}
        \caption{OMol25 training hyperparameters.}
        \label{tab:sm_omol25_config}
        \begin{tabular*}{\linewidth}{@{\extracolsep{\fill}}
                L{2.65cm}C{1.66cm}C{1.66cm}C{1.66cm}C{1.66cm}C{1.66cm}C{1.66cm}@{}} \toprule \textbf{Hyperparameter} & \makecell{\textbf{DPA4-}\\\textbf{Nano}} & \makecell{\textbf{DPA4-}\\\textbf{Mini}} & \makecell{\textbf{DPA4-}\\\textbf{Neo}} & \makecell{\textbf{DPA4-}\\\textbf{Air}} & \makecell{\textbf{DPA4-}\\\textbf{Plus}} & \makecell{\textbf{DPA4-}\\\textbf{Pro}} \\ \midrule LR scheduler\tnote{a} & WSD & WSD & WSD & WSD & WSD & WSD \\ Max.
            LR & $1\times10^{-3}$ & $5\times10^{-4}$ & $7\times10^{-4}$ & $3\times10^{-4}$ & $3\times10^{-4}$ & $2.5\times10^{-4}$ \\ Min.
            LR & $1\times10^{-6}$ & $1\times10^{-6}$ & $1\times10^{-6}$ & $1\times10^{-6}$ & $1\times10^{-6}$ & $1\times10^{-6}$ \\ Warmup ratio & 0.003 & 0.003 & 0.003 & 0.003 & 0.003 & 0.003 \\ Decay ratio & 0.65 & 0.65 & 0.65 & 0.65 & 0.65 & 0.65 \\ Decay type & Cosine & Cosine & Cosine & Cosine & Cosine & Cosine \\ Batch size (per GPU)\tnote{b} & $\lceil 25000/N\rceil$ & $\lceil 20000/N\rceil$ & $\lceil 6000/N\rceil$ & $\lceil 4000/N\rceil$ & $\lceil 2000/N\rceil$ & $\lceil 700/N\rceil$ \\ Training epochs & 23 & 21 & 19 & 17 & 13 & 10 \\ No.
            GPUs & 2 & 4 & 8 & 16 & 24 & 72 \\
            \midrule
            Loss & MAE & MAE & MAE & MAE & MAE & MAE \\
            Loss weights $(E,F,V)$\tnote{c} & 10, 5, 0 & 10, 5, 0 & 10, 5, 0 & 10, 5, 0 & 10, 5, 0 & 10, 5, 0 \\
            \bottomrule
        \end{tabular*}
        \begin{tablenotes}[flushleft]
            \PaperTableNotesStyle
            \item[a] LR denotes learning rate, and WSD denotes the warmup--stable--decay schedule.
            \item[b] $N$ denotes the number of atoms in each system; $\lceil\cdot\rceil$ rounds up to the nearest integer.
            \item[c] $E$, $F$ and $V$ denote the energy, force and virial loss contributions, respectively.
        \end{tablenotes}
    \end{threeparttable}
\end{table}

\begin{table}[p]
    \PaperTableStyle
    \begin{threeparttable}
        \caption{SPICE-MACE-OFF training hyperparameters.}
        \label{tab:sm_spice_mace_off_config}
        \begin{tabular*}{\linewidth}{@{\extracolsep{\fill}}
                L{3.6cm}C{2.8cm}C{2.8cm}C{2.8cm}@{}} \toprule \textbf{Hyperparameter} & \textbf{DPA4-Neo} & \textbf{DPA4-Air} & \textbf{DPA4-Plus} \\ \midrule LR scheduler\tnote{a} & WSD & WSD & WSD \\ Max.
            LR & $8\times10^{-4}$ & $5\times10^{-4}$ & $3.5\times10^{-4}$ \\ Min.
            LR & $1\times10^{-6}$ & $1\times10^{-6}$ & $1\times10^{-6}$ \\ Warmup ratio & 0.003 & 0.003 & 0.003 \\ Decay ratio & 0.65 & 0.65 & 0.65 \\ Decay type & Cosine & Cosine & Cosine \\ Batch size (per GPU)\tnote{b} & $\lceil 4000/N\rceil$ & $\lceil 2000/N\rceil$ & $\lceil 1000/N\rceil$ \\ Training steps & $2\times10^6$ & $2\times10^6$ & $2\times10^6$ \\ No.
            GPUs & 1 & 1 & 1 \\
            \midrule
            Loss & MAE & MAE & MAE \\
            Loss weights $(E,F,V)$\tnote{c} & 15, 20, 0 & 15, 20, 0 & 15, 20, 0 \\
            \bottomrule
        \end{tabular*}
        \begin{tablenotes}[flushleft]
            \PaperTableNotesStyle
            \item[a] LR denotes learning rate, and WSD denotes the warmup--stable--decay schedule.
            \item[b] $N$ denotes the number of atoms in each system; $\lceil\cdot\rceil$ rounds up to the nearest integer.
            \item[c] $E$, $F$ and $V$ denote the energy, force and virial loss contributions, respectively.
        \end{tablenotes}
    \end{threeparttable}
\end{table}

\end{document}